\documentclass[11pt,  a4paper]{article}

\usepackage{a4, amsmath, amssymb, graphics, subfigure, fullpage, color}
\usepackage[pdftex]{graphicx}
\usepackage[colorlinks]{hyperref}
\hypersetup{colorlinks,citecolor=green,filecolor=black,linkcolor=blue,urlcolor=blue}
\usepackage{fancyhdr}
\usepackage{multirow}
\usepackage{ftnxtra}
\usepackage{amsthm}
\newtheorem{theorem}{Theorem}
\newtheorem{definition}{Definition}
\newtheorem{corollary}{Corollary}

\newtheorem{lemma}{Lemma}

\newtheorem{remark}{Remark}

\newtheorem{problem}{Problem}
\usepackage{upgreek} 
\usepackage{lscape}
\usepackage{booktabs}
\usepackage{float}
\usepackage{comment}

\newcommand{\InclFig}

\usepackage[round]{natbib}

\newcommand{\bed}{\begin{definition}}
\newcommand{\eed}{\end{definition}}

\newcommand{\cB}{\mathcal{B}}

\newcommand{\cE}{\mathcal{E}}

\newcommand{\cF}{\mathcal{F}}
\newcommand{\cG}{\mathcal{G}}
\newcommand{\cH}{\mathcal{H}}
\newcommand{\cT}{\mathcal{T}}

\newcommand{\FF}{\mathbb{F}}
\newcommand{\GG}{\mathbb{G}}
\newcommand{\HH}{\mathbb{H}}

\DeclareMathOperator{\E}{\mathbb{E}}
\newcommand{\Q}{\mathbb{Q}}

\DeclareMathOperator{\e}{\textrm{e}}

\newcommand{\beq}{\begin{equation}}
\newcommand{\eeq}{\end{equation}}
\newcommand{\beqn}{\begin{eqnarray}}
\newcommand{\eeqn}{\end{eqnarray}}
\newcommand{\bfig}{\begin{figure}}
\newcommand{\efig}{\end{figure}}
\newcommand{\btab}{\begin{table}}
\newcommand{\etab}{\end{table}}

\ifdefined \up 
	\renewcommand{\up}{\textrm{up}}
\else 
	\DeclareMathOperator{\up}{\textrm{up}} 
\fi

\title{\bf Affine term-structure models : a time-changed approach with perfect fit to market curves}

\author{Cheikh Mbaye \& Fr\'ed\'eric Vrins\thanks{Voie du Roman Pays 34, B-1348 Louvain-la-Neuve, Belgium. E-mail: \href{mailto:frederic.vrins@uclouvain.be}{frederic.vrins@uclouvain.be}. The research of Cheikh Mbaye is funded by the National Bank of Belgium and an FSR grant. The opinions expressed in this paper are those of the authors and do not necessarily reflect the views of the National Bank of Belgium. The work of F. Vrins was supported by the Fonds de la Recherche Scientifique F.S.R.-FNRS, Grant J.0037.18.} \\[-0.12cm]  
Louvain Finance Center (LFIN), 
UC Louvain, Belgium
 \\[0.2cm]
}
\date{}
\begin{document}

\maketitle
\font\dsrom=dsrom10 scaled 1200
\def \indic{\textrm{\dsrom{1}}}

\begin{abstract} We address the so-called \textit{calibration problem} which 
consists of fitting in a tractable way a given model to a specified term structure like, e.g.,  yield, prepayment or default probability curves. Time-homogeneous jump-diffusions like Vasicek or Cox-Ingersoll-Ross (possibly coupled with compound Poisson jumps, JCIR, a.k.a. SRJD), are tractable processes but have limited flexibility; they fail to replicate actual market curves. The deterministic shift extension of the latter, Hull-White or JCIR++ (a.k.a. SSRJD) is a simple but yet efficient solution that is widely used by both academics and practitioners. However, the shift approach may not be appropriate when positivity is required, a common constraint when dealing with credit spreads or default intensities. In this paper, we tackle this problem by adopting a time change approach, leading to the TC-JCIR model. 
On the top of providing an elegant solution to the calibration problem under positivity constraint, our model features additional interesting properties in terms of variance. 
It is compared to the shift extension on various credit risk applications such as credit default swap, credit default swaption and credit valuation adjustment under wrong-way risk. The TC-JCIR model is able to generate much larger implied volatilities and covariance effects than JCIR++ under positivity constraint, and therefore offers an appealing alternative to the shift extension in such cases.
\end{abstract}


\section{Introduction}
\label{sec:Intro}

\textit{Model calibration} is a standard problem in many areas of finance~\cite{Brigo06,Joshi03,Vero10}. It consists of tuning a model such that it ``best fits'' market quotes at a given time. As an example, financial markets provide a set of prices associated with liquid instruments, that openly trade on the market. Alongside with risk management (hedging), the main purpose of a model here is to act as an ``interpolation/extrapolation'' tool, i.e., to obtain the value of products at a given time $t$ for which the market does not disclose prices in a transparent way. This could happen because either the product to be priced is ``exotic'' (i.e., is too ``special'', it does not quote openly on a platform, only on a bilateral basis) or because its cashflow schedule is not in line with the products that currently trade openly at $t$ (a situation that commonly happens since products that were ``standard'' at inception, may have time-to-expiry or moneyness levels that are no longer ``standard'' afterwards).\medskip 

Mathematically, model calibration is nothing but an optimization problem. Starting from a set of prices quoted on the market (called ``market prices'') for a set of specific financial products (called ``calibration instruments''), model calibration consists of computing the model parameters such that the prices generated by the model (called ``model prices'') best fit to the market prices, according to some error function. Model calibration is crucial in finance; it is strongly related to arbitrage opportunities. In practice, only models that are able to reproduce the market prices of ``simple instruments'' (either in a perfect way, or at least up to the bid-ask spread) are trustworthy enough when it comes to pricing other instruments. For instance, one can price exotic derivatives (like barrier options) using stochastic volatility model like Heston in a semi-analytical way~\cite{Carr99,Heston93}. The parameters of the Heston model will be obtained by ``calibration'' to a volatility surface, i.e., to a set of liquid (``plain vanilla'') options, like European calls and puts of various strikes and maturities. The justification behind this is that in a no-arbitrage, complete market setup, the price of an option can be obtained by computing the cost of setting up a self-financing hedging strategy. This cost depends on the prevailing prices of the hedging instruments. If the model fails to correctly price the latter, there is no chance it can correctly price the option. \medskip

In this work, we focus on financial calibration problems arising in other asset classes: interest-rates and credit \cite{Brigo06,Duffie03}. When specifying an interest rate model to price a derivative on, say, the Libor 3M index, one needs to make sure that the model generates a discount curve that is in line with that extracted from market quotes of simpler Libor 3M-indexed products. In this case, the set of calibration instruments could be forward rate agreements (FRA), interest rate swaps (IRS), as well as vanilla cap/floors or swaptions. Similarly, adjusting the value of derivatives for counterparty risk (a problem known as credit valuation adjustment, or CVA) generally involves a stochastic model to represent the default of the counterparty with whom the trade is executed. The default probability of the counterparty can be extracted from a set of calibration instruments, which prices are driven by the default likelihood of the counterparty, e.g., corporate bonds of credit default swaps (CDS). In this context, the default model must be ``calibrated'' in such a way that the default probability curve generated by the stochastic model agrees with that implied from the prices of the corresponding instruments (see, e.g., \cite{Grego10} and \cite{Stein11} for a general overview of CVA and \cite{brigopallaagocollateral} for a discussion of bilateral CVA in presence of collateralization agreements). \medskip

The calibration constraint rises practical issues. Indeed, the models that are actually used in the industry must have a tractability that is compatible with real-time pricing but, as explained above, must be flexible enough to match the information conveyed by the calibration instruments. Affine term structure models (ATSM) have been extensively used in fixed income modeling because of their analytical tractability. See, e.g., \cite{Duffie03b} and \cite{Duffie96a} for an excellent review and a mathematical analysis of this class of processes. In practice, homogeneous affine jump diffusion (HAJD) models are extremely popular. The Vasicek (Ornstein-Uhlenbeck) model~\cite{Vas77} is a short-rate model being widely used in both industry and academia. It is a time-homogeneous affine diffusion model that postulates Gaussian dynamics. If negative rates need to be ruled out, positive dynamics like the CIR (Cox-Ingersoll-Ross, also known as square-root diffusion, SRD)~\cite{Cox85b} can be preferred, possibly with independent compounded Poisson jumps (JCIR or SRJD). However, it is in general impossible with either models (even in a multi-factor setup) to achieve a perfect fit: the flexibility of HAJD is limited, they are in general unable to generate a given discount curve. The same problem arises when dealing with credit derivatives: it is generally impossible to make sure that the default intensity process, modeled with HAJD dynamics, will generate a default probability curve that is in line with the corresponding curve, exogenously given by the market via the calibration instruments. \medskip

Several routes can be followed to deal with this issue. The first one consists of disregarding this lack of flexibility. Nevertheless, working with a model that fails to yield a perfect fit to the market is often unacceptable in practice. Indeed, as explained above, the models are used to value derivatives positions, and a mismatch with the market can introduce a tremendous bias in the valuation of the book of companies or financial institutions. Another possibility is to significantly increase the complexity of the models. This is often to be avoided in practice, for computational, identification or over-fitting issues. A trade-off consists of extending the ``simple models'' in such a way that they can fit to the market. In fact, several authors show that a great flexibility can be obtained by shifting HAJD models in a deterministic way.~\cite{Dyb97} for instance show that the term structure of interest rates can be reproduced by adding a deterministic shift to~\cite{HoLee86} or Vasicek processes. Later,~\cite{Brigo01} extended this idea to a broader class of models, thereby providing a simple but very clever solution to the calibration problem. Instead of considering a HAJD, one could simply adjust it with a deterministic function $\varphi$: the resulting process will have the required flexibility. When shifted this way, the Vasicek, CIR and JCIR models respectively correspond to the Hull-White, the CIR++ or the JCIR++ (a.k.a. SSRJD) models~\cite{Brigo06}. This trick is actually very powerful: it solves the calibration problem at no cost, since the model's dynamics remain affine. Moreover, the shift function is known analytically, as a function of the parameters of the underlying HAJD and the market curve to be fitted by the model.\medskip

Yet, this approach suffers from an important limitation. Because of the shift, there is no reason that the range of the shifted process agrees with that of the underlying HAJD process. For instance, shifting a positive process with a deterministic function may result in a process that could take on negative values. It all depends on the mismatch between the information conveyed by the calibration instruments on the one hand, and the parameters of the underlying HAJD on the other hand. In general, there is no reason to believe that the implied shift function will preserve the range of the HAJD model. This is problematic in many cases, and in credit risk modeling in particular: negative default intensities, for instance, make no sense. To circumvent this issue, one could think of adding a non-negativity constraint on the shift in the calibration step. But, as we will show, this drastically restricts the parameters of the underlying HAJD, hence the randomness embedded in the model. This explains why this solution is often not considered by practitioners: the shift approach (without positivity constraint) remains the standard approach, even if positivity is required, theoretically speaking. It seems that in absence of a valid alternative, one actually prefers to rely on a model providing a perfect fit, even though the latter suffers from theoretical inconsistencies. \medskip

In this paper we introduce an alternative to the deterministic shift. Using an equally simple \--- but intrinsically different \--- technique, we adjust a HAJD so as to allow for a perfect fit to a given market curve, without affecting the model's tractability, but also without introducing the aforementioned inconsistencies. More specifically, instead of shifting a HAJD, we time-change it. Time change techniques were first studied in 1965 \cite{Dambis65,Dubins65}. The first application to finance dates back from the early 2000. Geman et al. used L\'evy processes and interpreted the new time scale as the business time, in contrast with the calendar time \cite{Geman01}. This was then applied to stochastic volatility models~\cite{Carr03}. Thanks to subordinated L\'evy models, the authors introduced the leverage effect, as well as a long-term skew. Many other financial applications of time change techniques can be found in the review \cite{Swis16}. More recently, Mendoza-Arriaga and Linetsky used stochastic time change processes to introduce two-side jumps in positive processes. The analytical tractability of the resulting model is preserved to some extend. This model has been recently applied to counterparty credit risk~\cite{Mbaye2018}. In this work, we exploit the time change idea in yet another way, to solve a completely different problem. Our purpose is to time-change HAJDs so as to obtain models with the desired calibration flexibility, without affecting tractability and preserving the range of the original process. The intuition is that by slowing down or speeding up the time of the latent HAJD, at the appropriate rate, one would obtain a model that could fit most discount curves, and actually every default probability curve. Moreover, the time change function is easily found using simple numerical methods (namely, inversion of easy functions or ordinary differential equation). Eventually, our time-changed HAJD is proven to feature larger implied volatilities compared to the corresponding valid (i.e., non-negative) shifted HAJD. 
To illustrate the power of our approach, we provide two applications taken from credit: pricing of CDS options and computation of derivatives pricing accounting for counterparty risk under exposure-credit dependence (wrong-way risk, WWR). In either cases, all the considered default models perfectly fit the risk-neutral default probability curve extracted from market quotes associated to the CDS of the reference entity. The obtained results illustrate the nice feature of large implied volatilities : they are able to generate larger option prices compared to the shift approach calibrated on a same probability curve under non-negativity constraint.\medskip 

Eventually, observe that although we focus on examples featuring reduced-form models when pricing of credit-sensitive instruments, our approach is of potential use for other models, in many areas of finance and insurance. Ongoing work suggests that it can be applied to other default models, including the firm-value (structural) models~\cite{Merton74}, but also to linear-rational (polynomial) models~\cite{Filip16}. Other models could be considered as well, like~\cite{Vrins16} or~\cite{Crep12}.In terms of applications, the proposed method can be used in life insurance, to calibrate mortality rate to mortality tables. Our time-changed process could also be used to model prepayment rates in mortgage-backed securities (MBS). These products naturally exhibit a negative convexity due to a the negative relationship between interest and prepayment rates: householders tend to refinance their loans when interest rates drop. This calls for a stochastic prepayment (i.e., positive) rate, that will be negatively correlated with interest rates, and which parameters could be calibrated so as to agree with the averaged values given in the PSA measure, the indicator attached to MBS securities that characterizes the prepayment speed in MBS \cite{Vero10}. Eventually, the proposed method could be applied in many other applications, including the modeling of performance degradation of devices or materials through time, which average outstanding performances evolution are given according to some quality standards.\medskip

The paper is organized as follows. In Section \ref{sec:Problem} the calibration problem is introduced and two specific cases (cashflow discounting and probability curves) are discussed. We then recall in Section  \ref{sec:JD} how a shifted version of time-homogeneous affine jump diffusions can fit every discount curve. We pay specific attention to the case where the resulting process needs to meet a positivity constraint. We then introduce in Section \ref{sec:TC} our alternative model, specifically devoted to this case, focusing on the most common HAJD, namely the Vasicek and JCIR (generalizing the CIR) models. Eventually, we compare in Section \ref{sec:app} our model's performance to that of the shift approach on three different pricing problems taken from credit risk: CDS curve calibration, pricing of CDS options and pricing of credit valuation adjustment under wrong-way risk.

\section{The calibration problem}
\label{sec:Problem}

Consider a given time-$s$ \textit{market curve} $P^{market}_s(t)$,  $t\geq s$. The calibration problem consists of finding, for a given model, the (set of) parameter(s) $\Xi=\Xi^\star$ such that the corresponding \textit{model curve} $P^{model}_s(t)=P^{model}_s(t;\Xi)$ ``best fits'' the market curve, according to some criterion. Mathematically speaking, this is an optimization problem that consists of finding a set of parameters that minimizes an error function between model and market values,
\beq
\Xi^\star := \arg\min_{\Xi} \|P^{model}_s(\cdot;\Xi)-P^{market}_s(\cdot)\|\;,\label{eq:Xistar}
\eeq
where $\|f(\cdot)-g(\cdot)\|$ represents a divergence measure between two functions $f,g$. In practice, one often computes the mean-square error (MSE) between $f$ and $g$ on a set of maturities $\cT:=\{T_1,\ldots,T_n\}$:

\beq
\|f(\cdot)-g(\cdot)\|:=\frac{1}{n}\sum_{i=1}^n \big(f(T_i)-g(T_i)\big)^2\;.\label{eq:MSE}
\eeq

A model with parameter $\Xi$ is said to \textit{perfectly fit} the market up to horizon $T$ whenever $P^{model}_s(t;\Xi)=P^{market}_s(t)$ for all $s\leq t\leq T$ or using a shorthand notation, $P^{model}_s\equiv P^{market}_s$.


\subsection{Setup}\label{sec:Setup}

A model can be either static or dynamic. For instance, the Nelsen-Siegel model \cite{NS87} postulates a parametric form for the yield curve, but is a static model: the resulting curve does not correspond to the yield curve generated by the dynamics of a stochastic model. We focus on continuous-time dynamic models in the sequel.\medskip

We consider a frictionless market free of arbitrage opportunities in which trading takes place continuously over the time interval $[0,T]$, where $T$ is a fixed time horizon. Uncertainty in the market is modelled through a filtered probability space $(\Omega,\cG,\GG,\mathbb{Q})$. In this setup, $\GG=(\cG_t,t\in [0,T])$ represents the information flow and corresponds to the filtration generated by the stochastic market variables (risk factors, prices, interest rates, default intensities, default event, etc), $\cG:=\cG_T$, and $\mathbb{Q}$ denotes the risk-neutral probability measure referred to as the \textit{pricing measure}. In the sequel, we shall focus on a specific class of $P^{model}$ and $P^{market}$ functions: we assume they are \textit{discount curves}, a set of functions that we now define.

\begin{definition}[Discount curve]\label{def:DiscCurve}
A time-$s$ discount curve is any differentiable function of the form $$P_s:[s,\infty)\to\mathbb{R}_0^+, ~~t\mapsto P_s(t)$$ 
satisfying $P_s(s)=1$.
\end{definition}

In the specific $s=0$ case, a \textit{time-0 discount curve} $P_0(t)$ is simply called a \textit{discount curve} and is noted $P(t)$, assuming implicitly that $t\geq 0$. Any time-$s$ discount curve admits an exponential-integral from:

\begin{lemma}\label{lem:FwdCurve}Every time-$s$ discount curve $P_s$ admits a representation in terms of time-$s$ instantaneous forward rate curve $f_s$:
$$P_s(t)=e^{-\int_s^t f_s(u)du}\;,~~t\geq s.$$
If moreover $P_s(t)$ is strictly decreasing on $(s,\infty)$, then $f_s(t)$ is strictly positive for all $t>s$.
\end{lemma}
\begin{proof}
Since $P_s(t)>0$ for all $t\geq s$ and $P_s(t)$ is differentiable with respect to $t$ on $(s,\infty)$, then one can define the time-$s$ instantaneous forward rate function as $f_s(t):=\frac{-1}{P_s(t)} \frac{d}{dt}P_s(t)=-\frac{d}{dt}\ln P_s(t)$ for all $t>s$; the value $f_s(s)$ is not identified but can be defined by, e.g., the limit as $t\downarrow s$. Moreover, if $P_s$ is strictly decreasing on $(s,\infty)$ then $f_s(t)>0$ for all $t>s$. 
\end{proof}

Discount curves are of paramount importance in finance. As suggested by the name, $P_s(t)$ allows one to compute the time-$s$ value of a cashflow paid at time $t\geq s$, both in a credit risk-free and credit risky setup. As a special case of the second framework, they encompass survival probability curves, defined as one minus cumulative distribution functions. This is elaborated in the next two subsections.

\subsubsection{Discounting in a default-free market}
In this application, $P_s(t)$ stands for the time-$s$ price of a risk-free zero-coupon bond (ZCB) with maturity $t$ and face value 1, denominated in a given currency. In particular, $P^{market}_s(t)$ and $P^{model}_s(t;\Xi)$ respectively give the market and the model prices of that instrument.\medskip 

Indeed, in a no-arbitrage setup, the price of a financial instrument paying a single cashflow (payoff) at a given maturity is given by the risk-neutral conditional expectation of the payoff, discounted at the risk-free rate from the payment date (maturity) back to the valuation date. Adopting a short-rate model, $P^{model}_s(t)$ corresponds to the $\Q$-expectation of the stochastic discount factor $D_s(t):=e^{-\int_s^t r_u du}$, the negative exponential of the risk-free short-rate process $r$, integrated from the valuation time $s$ up to the payment time $t$, conditional upon the information prevailing at the pricing time~\cite{Brigo06}:

$$P^{model}_s(t)=\E\left[\left.1D_s(t)\right|\cG_s\right]=\E\left[\left.e^{-\int_s^t r_u du}\right|\cG_s\right]=:P^{r}_s(t)\;.$$

In this context, we aim at finding a model $x$ to depict the risk-free short rate dynamics $r$ that would be tractable enough, and provide a perfect fit to any yield curve, the curve that gives the set of prices of ZCBs with increasing maturities.

\subsubsection{Discounting in a defaultable market}\label{sec:defaultable} 
Adopting the same framework as before, the time-$s$ price of a zero-coupon bond paying one unit of currency at time $t\geq s$ \textit{contingent on the fact that the issuer doesn't default prior to the payment date} is given by a similar expression as before. It suffices to replace the risk-free payoff $1$ by the risky one, namely  $\indic_{\{\tau> t\}}$, where the random variable $\tau$ represents the default time of the issuer and $\indic_{A}$ is the indicator function defined as 1 if $A$ is true and zero otherwise. Mathematically,

$$P^{model}_s(t)=\E\left[\left.\indic_{\{\tau>t\}}D_s(t)\right|\cG_s\right]=:\bar{P}^{r}_s(t)\;.$$

In such a context, $\bar{P}^{r}_s$ corresponds to a risky discounting, where the term \textit{risk} is referring to the possibility for the issuer not to meet her financial obligations. \medskip

To proceed, we need to model the default event. To that end, we consider a \textit{reduced-form} (a.k.a. \textit{intensity}) default model. We refer the reader to~\cite{Duffie99} and~\cite{Lando04} for an extensive exposition of this class of models. In this framework, the default time $\tau:=\tau(\lambda)$ is defined as the passage time of the process $\Lambda:=(\Lambda_t,t\in[0,T])$ defined as $\Lambda_t:=\int_0^t \lambda_s ds$ above a unit-mean exponential random variable $\cE$ independent from every other processes. The process $\lambda$ is an \textit{intensity}, i.e., it is positive, so that $\Lambda$ is almost surely increasing. In this model, the default event $\{\tau\leq t\}$ is modeled as $\{\Lambda_t\geq \cE\}$ and the survival probability is given by
$$\Q\left(\tau>t\right)=\Q\left(\Lambda_t<\mathcal{E}\right)=\Q\left(U<e^{-\Lambda_t}\right)=\E\left[e^{-\Lambda_t}\right]\;,$$
where $U:=e^{-\mathcal{E}}$ is a random variable uniformly distributed on $[0,1]$.\medskip

The function $\bar{P}^{r}_s(t)$ can be proven to be a time-$s$ discount curve in many cases. To show this, we first define a sub-filtration $\FF$ such that all processes are $\FF$-adapted except those featuring $\tau$ (i.e., those featuring $\cE$ or $U$, which are independent from $\cF_T$). We then define a second filtration $\mathbb{H}=(\mathcal{H}_t,t\in [0,T])$, the filtration generated by the default process $\mathcal{H}_t=\sigma(\indic_{\{\tau<u\}}, u<t)$. Eventually, the total filtration $\GG$ is recovered by progressively enlarging $\FF$ with $\HH$: $\mathbb{G}=\mathbb{F}\vee\mathbb{H}$. Hence, $\tau$ is a $\GG$-stopping time, but not an $\FF$-stopping time. In such a case, one can replace $\cG_s$ by $\cF_s$ in the expression providing the time-$s$ price of the risk-free ZCB:
$$P^{r}_s(t)=\E\left[\left.D_s(t)\right|\cG_s\right]=\E\left[\left.D_s(t)\right|\cF_s\vee\cH_s\right]=\E\left[\left.D_s(t)\right|\cF_s\right]\;.$$
A central result in stochastic calculus is the so-called \textit{Key lemma}. This fundamental theorem allows one to write the $\cG_s$-conditional expectation of $X\indic_{\{\tau>t\}}$ as the $\cF_s$-conditional expectation of $Xe^{-\int_s^t \lambda_u du}$, rescaled by $\indic_{\{\tau>s\}}$, for every integrable and $\cF_t$-measurable random variable $X$. It is originally due to Dellacherie and Meyer \cite{Dell80}, although its use in financial applications have been put forward by Bielecki, Jeanblanc and Rutkowski~\cite{Biel02} (see, e.g., ~\cite{Biel11} for numerous examples in credit risk and~\cite{Brigo17} for a specific application in counterparty credit risk). Applying the Key lemma to the risky ZCB formula above yields, with $X\leftarrow e^{-\int_s^t r_u du}$,

$$\bar{P}^r_s(t)=\indic_{\{\tau>s\}}\E\left[\left.e^{-\int_s^t \lambda_u du}D_s(t)\right|\cF_s\right]=\indic_{\{\tau>s\}}\E\left[\left.e^{-\int_s^t (\lambda_u + r_u) du}\right|\cF_s\right]=:\indic_{\{\tau>s\}}P^{\lambda+r}_s(t)\;.$$

Eventually, in the special case where $r\equiv 0$, so that $\bar{P}^{r}_s(t)$ collapses to $$\bar{P}^{r}_s(t)=\E\left[\indic_{\{\tau>t\}}|\cG_s\right]=\Q\left(\tau>t|\cG_s\right)=\indic_{\{\tau>s\}}P^{\lambda}_s(t)\;.$$
Hence, on the event $\{\tau>s\}$, $\bar{P}^{r}_s(t)=P^{\lambda}_s(t)$ agrees with the survival probability function associated with $\tau$, conditional upon $\cG_s$.\medskip 

In this specific context, we are interested in a model $x$ to depict the dynamics of the intensity process $\lambda$ that would be tractable enough, and provide a perfect fit to any valid survival probability curve extracted from the prices of defaultable instruments like corporate bonds or credit default swaps (CDS). 

\begin{remark}
Notice that in contrast to rates, that can \--- and some of them currently do \--- take negative value, non-negativity is a formal requirement when $x$ represents an intensity process $\lambda$. A default model featuring ``negative intensities'' is theoretically flawed, and is problematic. Indeed, modelling the event $\{\tau>t\}$ as $\{\Lambda_t<\cE\}$ yields a survival indicator process $\indic_{\{\tau>t\}}$ that might jump both up and down, i.e., the reference entity could be ``brought back to life''. One could of course think of replacing the default event using a first passage time, thereby revisiting the default time definition as $\tau:=\inf\{t\geq 0:\Lambda_t\geq \cE\}$. However, one looses the analytical tractability for the survival probability since in this case, $\Q(\tau>t)$ does no longer agree with $\Q(\Lambda_t<\cE)=\E\left[e^{-\int_0^t \lambda_u du}\right]=P_0^\lambda(t)$. The $x\geq 0$ constraint is also a natural requirement when it represents a credit spread. 
\end{remark}

\subsection{The perfect fit problems}
\label{sec:PerfectFit}

Equation \eqref{eq:Xistar} suggests that the calibration problem consists of finding the parameters of a given model to minimize the discrepancies between market and model curves up to a time horizon $T$. However, it is clear that the choice of the model class will also have a substantial impact. Indeed, depending on the model chosen, the minimum of the error function could be large, small or even zero, in which case the perfect fit is obtained. 

Inspired by the financial problems mentioned in Section \ref{sec:Setup}, we consider the following problems directly related to the perfect fit constraint (up to a given time horizon $T$, that is implicit in the sequel). The first one does not impose any constraint on the process $x$ to consider.

\begin{problem}\label{prob:prob2}
Find a tractable process $x$ satisfying
$$P_s^x(t):=\E\left[\left.e^{-\int_s^t x_u du}\right|\cF_s\right]=P_s(t)$$
for $t\in[0,T]$ and every given discount curve $P_s$.
\end{problem}

Depending on the application at hand, one may need to impose additional constraints on $x$. As suggested by the risky discounting example, non-negativity is a crucial one. This leads us to consider a second (constrained) problem. 

\begin{problem}\label{prob:prob1}
Find a tractable positive process $x$ (i.e., such that $\Q(x_t\geq 0)=1$ and $\Q(x_t>0)>0$ for all $t\in[0,T]$) satisfying
$$P_s^x(t):=\E\left[\left.e^{-\int_s^t x_u du}\right|\cF_s\right]=P_s(t)$$
for $t\in[0,T]$ and every strictly decreasing discount curve $P_s$.
\end{problem}

In either problems, \textit{tractability} refers to the fact that model calibration \eqref{eq:Xistar} \--- that features an optimization over the parameter space \--- is not too cumbersome, computationally. Solving this optimization problem typically requires many iterations, hence numerous evaluations of the objective function. This suggests that a highly desirable feature of the model is to admit a closed form expression for $P^{model}$ or, at least, that the latter can be computed without having to rely on time-consuming numerical methods like, e.g., Monte Carlo simulations.

\section{Shifted homogeneous affine models}
\label{sec:JD}

In order to solve these two problems, we consider what is probably the most tractable family of models, namely \textit{affine processes} and, more specifically \textit{time-homogeneous affine processes}. Indeed, for a one-factor affine model $y:=(y_t,t\in[0,T])$, many expressions are available analytically, as well as for its integrated version $Y:=(Y_t,t\in[0,T])$, $Y_t:=\int_0^t y_u du$. In particular, $P_s^y(t):=\E\left[\left.e^{-\int_s^t y_u du}\right|\cF_s\right]$ is merely the conditional moment generating function of $Y_t-Y_s$, $t\geq s$.

\subsection{Affine processes and affine jump-diffusions}

As recalled in the introduction, ATSM models are widely used in finance because they offer an appealing modeling framework : they are scarce, and empirical evidences suggest that they depict relatively well the market dynamics. 
Affine models are characterized as follows \cite{Filip05}.

\begin{definition}[Affine process]
An \textit{affine process} is any process $y$ satisfying 

\begin{equation}
\label{eq:affine}
P^{y}_s(t):=\E\left[\left.e^{-\int_s^t y_u du}\right|\cF_s\right]=e^{A^y_s(t;\Xi)-B^y_s(t;\Xi)y_s}=:P^{y}_s(t;\Xi)
\end{equation}
where $\Xi$ is the (set of) parameter(s) governing $y$ and $A^y_s,B^y_s$ are differentiable functions satisfying $A^y_s(s;\Xi)=B^y_s(s;\Xi)=0$. 
\end{definition}

Provided that the $A,B$ functions are known, the analytical form \eqref{eq:affine} facilitates in a tremendous way calibration procedures such as \eqref{eq:Xistar} when the considered $P^{model}$ function takes the form of the conditional expectation in \eqref{eq:affine}, as illustrated on the risk-free and risky discounting applications. This explains why such models are so popular in term-structure modeling.

For such processes, the function $P^y_s$ is thus well-defined for every $s$, is positive, and satisfies $P^y_s(s)=1$. It is therefore a time-$s$ discount curve in the sense of Definition \ref{def:DiscCurve} since it is obviously differentiable on $(s,\infty)$. For instance, $P^r_s$ and $P^\lambda_s$ in the above two examples are time-$s$ discount curves whenever $r$ and $\lambda$ are affine processes, respectively.\medskip

It is known (see, e.g., \cite{Brigo06} and \cite{Duffie96a}) that every diffusion with affine drift and diffusion coefficients, regular enough so that a solution exists, is an affine process. Similarly, every jump-diffusion with such types of drift and variance coefficients and independent compounded Poisson jumps (i.e., exponentially-distributed jumps arriving according to a Poisson process) is also affine.

\begin{definition}[Affine jump-diffusions, AJD]\label{eq:affineJD} A stochastic process $y$ is called an \textit{affine} jump-diffusion if its dynamics take the form
\beq dy_t=(a(t)+b(t)y_t)dt+\sqrt{c(t)+d(t) y_t}dW_t+ dJ_t\label{eq:HAJD}
\eeq
with $W$ an $\FF$-Brownian motion and $J$ an $\FF$-adapted compound Poisson process independent from $W$, defined according to $J_t:=\sum_{j=1}^{N_t} \zeta_i$ where $N$ is a Poisson process with instantaneous jump rate $\omega(t)\geq 0$ and $\zeta_i$'s are i.i.d. exponentially distributed random variables with mean $\alpha\geq 0$. In the special case where the parameters $(a,b,c,d,\alpha,\omega)$ are constant, $y$ is said \textit{time-homogeneous}, or simply \textit{homogeneous}, or HAJD.
\end{definition}

As explained above, affine models are specifically relevant in our context when $A,B$ are known in closed form. This is the case for HAJD. Three important homogeneous cases are the Ornstein-Uhlenbeck, the square-root diffusion and the square-root jump-diffusion. The first one, widely known as the Vasicek (VAS) model, corresponds to the special case where $(a(t),b(t),c(t),d(t),\alpha,\omega(t))=(\kappa\beta,-\kappa,\eta^2,0,0,0)$ and $y_0\in\mathbb{R}$. The second model is the Cox-Ingersoll-Ross with (CIR), and is associated to $(a(t),b(t),c(t),d(t),\alpha,\omega(t))=(\kappa\beta,-\kappa,0,\delta^2,0,0)$, with $y_0,\beta>0$. Eventually, the JCIR is an extension of the CIR, associated with parameters $(a(t),b(t),c(t),d(t),\alpha,\omega(t))=(\kappa\beta,-\kappa,0,\delta^2,\alpha,\omega)$. The speed of mean-reversion $\kappa$ is assumed to be positive in all models. When the initial value $y_0$ is part of the parameters, we note the parameter set $\Xi_0$. In contrast to VAS which is a Gaussian model, the CIR and JCIR models are non-negative. We recall (and derive) some properties of these processes in the Appendix (Section \ref{sec:AppAffine}) for further references.\medskip

Observe that the sum of two affine processes $x,y$ is, generally speaking, not an affine process. Hence, it is not clear whether the risky discounting curve $P^{\lambda+r}_s$ is a time-$s$ discount curve, even in the simple case where both $r,\lambda$ are affine processes. Some special cases are discussed in the Appendix (Section \ref{sec:AppDiscCurve}). In the sequel, we consider a specific pricing time, say $s=0$ without loss of generality, and drop the observation time subscript for conciseness.\medskip

HAJD models like VAS, CIR and JCIR seem appropriate to solve problems \ref{prob:prob2} and \ref{prob:prob1}. Unfortunately, they do not allow for a perfect fit to a given discount curve $P$, except in very special cases. Indeed, it is not possible in general, for such type of processes $x$, to find $\Xi$ (or $\Xi_0$) such that $P^{model}:=P^x(\cdot;\Xi)\equiv P$, even up to a finite horizon $T$. 

\subsection{A deterministic shift extension}

The starting point is to notice that the limited capacities of homogeneous models result from their rigid parametric form. Therefore, an interesting route is to consider a family of models $x$ defined as time-dependent transform of a base HAJD model $y$ in such a way that the model's tractability is not affected. In this section, we recall the general deterministic shift extension approach. The latter has been introduced in the seminal paper~\cite{Brigo01} in order, precisely, to address calibration issues such as Problem \ref{prob:prob2}. In this model, $x:=x^\varphi$ is defined as a HAJD ($y$) that is shifted in a time-dependent way using a deterministic function $\varphi$: 
\begin{equation}
x^\varphi_t:=y_t+\varphi(t)\;.\label{eq:xyshift}
\end{equation}

Interestingly, $P^{model}(t):=P^{x^\varphi}(t;\Xi)$ where $x^\varphi$ remains affine (although no longer homogeneous) and is hence analytically tractable in terms of calibration since :
$$P^{x^\varphi}(t;\Xi)=e^{A^{x^\varphi}(t;\Xi)-B^{x^\varphi}(t;\Xi)x_0}$$
with
\begin{eqnarray}
A^{x^\varphi}(t;\Xi)&=&A^y(t;\Xi)- \int_0^t \varphi(u)du+B^y(t;\Xi)\varphi(0)\;,\nonumber\\
B^{x^\varphi}(t;\Xi)&=&B^y(t;\Xi)\;.\nonumber
\end{eqnarray}

Clearly, the dynamics of $x^\varphi$ are easily obtained from that of $y$. Indeed, assuming
\beq
dy_t=\mu(t,y_t)dt+\sigma(t,y_t)dW_t+dJ_t\;,\label{eq:ySDE}
\eeq
the dynamics of $x^\varphi$ read, when $\varphi$ is differentiable, as
\beq
dx^\varphi_t=dy_t+\varphi'(t)dt=(\mu(t,x^\varphi_t-\varphi(t))+\varphi'(t))dt+\sigma(t,x^\varphi_t-\varphi(t))dW_t+dJ_t\;,~~x^\varphi_0=y_0+\varphi(0)\;.\nonumber
\eeq
It can be shown that in the particular case where $y$ is a HAJD, then $x^\varphi$ remains an AJD, even though no longer homogeneous, unless $\varphi(t)$ is constant. For instance, if the dynamics of $y$ obey \eqref{eq:HAJD}, then $x^\varphi$ is governed by the same type of dynamics since
\beq
dx^\varphi_t=\left(a^\varphi(t)+b(t)x^\varphi_t\right)dt+\sqrt{c^\varphi(t)+d(t) x^\varphi_t}dW_t+ dJ_t\;.\label{eq:ShiftHAJD}
\eeq
where $a^\varphi(t):=a(t)+\varphi'(t)-b(t)\varphi(t)$ and $c^\varphi(t):=c(t)-d(t)\varphi(t)$. 
As already noticed in~\cite{Brigo01}, whatever the base model $y$, the parameter $\Xi$ and the discount curve $P^{market}$, there always exists a shift function $\varphi(t)=\varphi^\star(t;\Xi)$ that provides a perfect fit between the $x^\varphi$-model and the market. This is summarized in the next lemma. 

\begin{remark}
The shift approach may look suspicious: adding a deterministic function to a stochastic process is arguably a somewhat artificial way to fix the model's limitations in terms of calibration. However, as clear from \eqref{eq:ShiftHAJD}, shifting the model in a deterministic way actually amounts to consider an inhomogeneous model. For instance, the Vasicek model $(a(t),b(t),c(t),d(t))=(0,-\kappa,\eta,0)$ shifted with $\varphi(t)\leftarrow \int_0^t\beta(s)e^{-\kappa(t-s)}ds$ yields a HAJD with $(a(t),b(t),c(t),d(t))=(\kappa\beta(t),-\kappa,\eta,0)$, which is known as the Hull-White (HW) model \cite{Hull90}. Moreover, the later is itself a particular case of the Heath-Jarrow-Morton (HJM) model \cite{HJM92} which consists of modeling the entire instantaneous forward curve $f_s(t)$ with $df_s(t)=\mu(s,t)dt+\eta e^{-\kappa(t-s)}dW_s$ where the drift $\mu(s,t)$ is given by no-arbitrage, and provided that the initial discount curve and the long-term mean obey the relationship $\beta(t)=\frac{d}{dt}f_0(t)+\kappa f_0(t)+\frac{\eta^2}{2\kappa}(1-e^{-2\kappa t})$. Therefore, any instantaneous forward curve $f^{market}$ (hence discount curve $P^{market}$) can be fitted with either models provided that one takes $f_0(t)\leftarrow f^{market}(t)$ as initial curve (HJM), the corresponding long-term mean $\beta(t)$ (HW), or the associated shift $\varphi(t)$ (shifted Vasicek). These models became very popular among practitioners, essentially because of their ability to replicate market curves, i.e., to solve Problem~\ref{prob:prob2}. 
\end{remark}

\begin{lemma} The $x$-model defined according to \eqref{eq:xyshift} where $y$ is a HAJD solves Problem~\ref{prob:prob2} provided that
\begin{equation}
\varphi(t)\leftarrow\varphi^\star(t;\Xi):=\frac{d}{dt}\ln \frac{P^y(t;\Xi)}{P^{market}(t)}=f^{market}(t)-f^{y}(t;\Xi)\label{eq:shift}\;.
\end{equation}
where $f^{market}$ and $f^{y}$ are the instantaneous forward rate functions associated with $P^{market}$ and $P^y$, respectively.
\end{lemma}
\begin{proof}
Indeed, because $y$ is a HAJD, $P^y$ is a discount curve and from Lemma~\ref{lem:FwdCurve}, it admits a representation in terms of forward rates $f^y$. By assumption, same holds true for $P^{market}$. Eventually, 
$$P^{x^\varphi}(t;\Xi)=\E\left[e^{-\int_0^t x_u du}\right]=e^{-\int_0^t \varphi^\star(u;\Xi) du}\E\left[e^{-\int_0^t y_u du}\right]=e^{-\int_0^t f^{market}(u)du}=P^{market}(t)\;.$$
The model is tractable since $f^y(t;\Xi)=-\frac{d}{dt}\ln P^y(t;\Xi)$ can be computed in closed form.
\end{proof}

It is worth noting that, for a given model $y$, the perfect fit can be attained for every parameters  $\Xi$. This suggests that the calibration problem \eqref{eq:Xistar} is ill-posed. Indeed, the choice of $\Xi$ is completely arbitrary since the error between $P^{x^\varphi}$ and $P^{market}$ can be set to zero for any $\Xi$, provided that one chooses $\varphi(t)\leftarrow\varphi^\star(t;\Xi)$. In particular, one could take the null process for $y$ 
and $\varphi(t)=f^{market}(t)$. This trivial choice rends $x^\varphi$ deterministic, which is most likely not the desired result. A common practice to circumvent this indeterminacy is thus either (i) to extend the set of calibration instruments, incorporating products that are sensitive to volatility (like interest-rate or credit options in the above asset classes), or (ii) to require the $y$-model to fit the market ``as best as possible'' (to get $\Xi^\star$) and then take $\varphi(t;\Xi^\star)$ as shift function:
\beq
P^{model}(t):=P^{x^\varphi}(t;\Xi^\star)~~\text{where}~~\Xi^\star := \arg\min_{\Xi} \|P^{y}(\cdot;\Xi)-P^{market}(\cdot)\|\;,~~\varphi(t)\leftarrow\varphi^{\star}(t):=\varphi^\star(t;\Xi^\star)\;.\label{eq:CalyNoConstraint}
\eeq

This approach is particularly relevant when no or little ``volatility-sensitive'' instruments are quoted on the market. The role of the shift is thus merely to compensate the remaining discrepancies between the market curve $P^{market}$ and the one generated by the ``best'' parametric model $y$, $P^{y}(\cdot;\Xi^\star)$. Adding a shift to the VAS, CIR or JCIR models yield the Hull-White, CIR++ or JCIR++, respectively~\cite{Brigo06}.\footnote{Notice that the Hull-White model is a Vasicek model where the long-term mean parameter is replaced by a deterministic function of time.}\medskip

\begin{remark}
On the top of the appealing affine structure, the shifted model is highly tractable because many statistical properties of the process are available in closed form. Indeed, as recalled in Section \ref{sec:AppAffine}, the $k$-th moment $m^{y}(k,t):=\E[y_t^k]$ and the moment generating function (MGF) $\psi^{y}(u,t):=\E[e^{uy_t}]$ of a time-homogeneous affine model $y$ are known analytically, as well as those of their time-integrals $Y_t:=\int_0^t y_udu$, $m^{Y}(k,t)$ and $\psi^{Y}(u,t)$. Due to the simple shift structure, the corresponding expressions for $x^\varphi$, the shifted model, are readily available. For instance, the $k$-th moment of $x^\varphi_t$ and $X^\varphi_t$ are given by Newton's binomial formula applied to $(y_t+\varphi(t))^k$ and the MGFs simply collapse to $\psi^{x^\varphi}(u,t)=e^{u\varphi(t)}\psi^{y}(u,t)$ and $\psi^{X^\varphi}(u,t)=e^{u\int_0^t \varphi(s)ds}\psi^{Y}(u,t)$.
\end{remark}

\subsection{Dealing with the positivity constraint}\label{sec:positivity}

As discussed above, the deterministic shift extension nicely solves Problem \ref{prob:prob2}. In order to solve Problem \ref{prob:prob1} however, one first considers a non-negative base process $y$. Yet, there is no reason that the shifted process $x^\varphi$ would remain non-negative. For instance, taking CIR dynamics for $y$, $x^\varphi$ is non-negative on $[s,t]$ if and only if $\min_{u\in[s,t]}\varphi(u)\geq 0$. From \eqref{eq:shift}, the shift function depends both on the $y$ model (and its parameters $\Xi$) and on the market curve. 

\begin{remark}
Observe that the optimization problem  \eqref{eq:CalyNoConstraint} is contradictory with non-negative shift functions. Indeed, by construction of $\Xi^\star$, $P^{y}(\cdot;\Xi^\star)$ passes \textit{through} $P^{market}$. Consequently, the shift $\varphi(t)\leftarrow\varphi^\star(t;\Xi^\star)$ will lead to a perfect fit, but will correct for both negative and positive errors. In other words, $\varphi$ will change of sign. Therefore, this strategy does not provide a valid solution to Problem \ref{prob:prob1}. This will be illustrated on a real example in Section \ref{sec:example}.
\end{remark}

In order to satisfy the non-negativity constraint mentioned in Problem \ref{prob:prob2}, one needs to force the non-negativity constraint on the shift at the optimal parameters. The shift function under positivity constraint is referred to with the notation $\varphi^{\star,+}(t)$ to stress the difference with the unconstrained counterpart, $\varphi^{\star}(t)$.

\begin{lemma} Let $y$ be a HAJD that is non-negative on $[0,T]$ with parameters $\Xi^\star$ given by
\beq
\Xi^{\star,+} := \arg\min_{\Xi} \|P^{y}(\cdot;\Xi)-P^{market}(\cdot)\|~~\text{subject to}~~
f^{y}(t;\Xi)\leq f^{market}(t),
~~\forall\; 0\leq t\leq T\;.\label{eq:CalyConstraint}
\eeq
Then, the $x^\varphi$-model \eqref{eq:xyshift} with $\varphi(t)\leftarrow \varphi^{\star,+}(t):=\varphi^{\star}(t;\Xi^{\star,+})$ solves Problem~\ref{prob:prob1}.
\end{lemma}

\begin{proof}
The condition on the instantaneous forward rates ensures that shift function $\varphi^{\star,+}$ will be non-negative on $[0,T]$; this is obvious from \eqref{eq:shift}. Hence, since $y$ is assumed to be non-negative, so is the shifted process $x^\varphi$. Moreover, taking $\varphi(s)\leftarrow\varphi^\star(s;\Xi)$ yields a perfect fit for every $\Xi$, by construction, including $\Xi=\Xi^{\star,+}$. 
\end{proof}
Notice that there always exists a set of parameters $\Xi$ such that the constraint is met. Indeed, all parameters $\Xi$ associated to the deterministic case $y\equiv 0$ yield $f^y(\cdot;\Xi)\equiv 0$. Clearly, the constraint is met since $f^{market}(t)$ is strictly positive given that $P^{market}$ is strictly decreasing, by assumption. The shift is simply given by the market forward rate $\varphi(t)\leftarrow \varphi^{\star,+}(t)=f^{market}(t)$. However, the trivial process parameter is likely not to be satosfactory.\medskip

In order to deal with Problem~\ref{prob:prob1}, we need to consider a non-negative base model $y$. Given that we focus on HAJDs, we consider the CIR and JCIR models. To make the distinction between the two shifted models, we call S-(J)CIR the (J)CIR process shifted with $\varphi(t)\leftarrow\varphi^\star(t)=\varphi^\star(t;\Xi^\star)$ (i.e., without positivity constraint, and parameter $\Xi^\star$ given by \eqref{eq:CalyNoConstraint}) and PS-(J)CIR the (J)CIR process shifted with $\varphi(t)\leftarrow\varphi^{\star,+}(t)=\varphi^\star(t;\Xi^{\star,+})$ (i.e., under positivity constraint, and parameter $\Xi^\star$ given by \eqref{eq:CalyConstraint}). Although the PS-(J)CIR allows both for a perfect fit and the non-negativity constraint, one may argue that it is not as tractable as the (J)CIR. Indeed, the  optimization problem \eqref{eq:CalyConstraint} is more difficult than \eqref{eq:CalyNoConstraint} due to the constraint on the instantaneous forwards, even if some sufficient conditions on the parameters can be found. Second, and probably more importantly, this constraint is binding, in the sense that it often deeply impacts the optimal parameter $\Xi^{\star,+}$. Even if it is unlikely that the optimal solution corresponds to the deterministic case, it often yields dynamics associated to rates that feature ``little randomness''. This will be illustrated in Section \ref{sec:app}, first by comparing the variance of the integrated S-CIR and PS-CIR processes, as well as the impact when dealing with financial applications. These two points are discussed in \cite[sec. 3.9.3, p.107-109]{Brigo06}. To circumvent this issue in an interest rate framework, the authors suggest to relax the strict positivity constraint. By working in a setup where positivity is expected but not guaranteed, they obtain a process that yields much more realistic results in terms of implied volatility levels. This is perfectly fine in such a context as positivity of rates might be desirable (in some cases), but zero is by no means a strict lower bound (neither theoretically nor practically). Yet, this is more problematic when it comes to model such things as default intensities, because this kind of applications requires both strict positivity and, typically, large volatility. Increasing the variance of the CIR++ process without breaking Feller's constraint\footnote{Increasing the volatility of the CIR++ process by increasing the diffusion paramter $\delta$ just breaks the Feller's condition ($2\kappa\beta\geq\delta^2$) and leads to an intensity process that almost surely equals to zero at a given time interval.} can be achieved by incorporating compounded Poisson jumps (JCIR++) but, unfortunately, increasing the jump activity while maintaining the calibration to a given market curve $f^{market}$ is difficult under the positivity constraint. Indeed, the minimum of the implied shift function is driven down when increasing the jump activity because the difference $f^{\rm JCIR}(t)-f^{\rm CIR}(t)$ is non-negative and increases with $\omega,\alpha$ for $\alpha,\omega>0$ (see Appendix, Section \ref{sec:AppJcir}).
This observation combined with \eqref{eq:shift} leads to a lower shift function $\varphi$ for the JCIR++ than for the corresponding CIR++. For this reason, there is a need for an alternative to CIR++ and JCIR++ that would combine (i) tractability, (ii) the prefect fit feature, (iii) the large implied volatility and (iv) positivity. 

\section{The deterministic time-changed extension}\label{sec:TC}

In order to circumvent the drawbacks of the deterministic shift extension with regards to Problem \ref{prob:prob1}, we propose a different approach. In the same spirit as the shift, we aim at finding a model $x$ by adjusting a  time-homogeneous affine model $y$, that would benefit from a set of desirable properties. 

\subsection{Model setting}\label{sec:model}

The $x$-model is obtained by time-changing a HAJD $y$ using a specific (but deterministic) \textit{clock} $\Theta$ that may differ from the calendar clock. A clock is a time change function that can differ from identity, but having specific properties. 

\bed[Clock] A clock is an application $$\Theta:\mathbb{R}^+\to\mathbb{R}^+,\; t\mapsto \Theta(t)$$ that is a grounded, increasing and differentiable. In other words, a clock is any function $\Theta$ of the form
\beq
\Theta(t):=\int_0^t \theta(u) du~~\text{where}~~\theta(u)> 0\;,~~\forall \;u\geq 0\;.\nonumber
\eeq
\eed

Clearly, $\Theta(t)=t$ is the calendar clock, and  any function of the form $\Theta(t)=kt$, $k>0$, is again a clock, corresponding to a constant rescaling of the calendar time.\medskip 

Similarly to \eqref{eq:xyshift}, we define our model as $x=x^\theta$, obtained from the following transform of the base process $y$:
\beq
x^\theta_t:=\theta(t)y_{\Theta(t)}\;.\label{eq:xytc}
\eeq
The dynamics of $x^\theta$ are given by Ito's product rule. Defining the process $y^{\theta}:=\left(y_{\Theta(t)},t\in[0,T]\right)$, one gets
\beq
dx^\theta_t=y^\theta_td\theta(t)+\theta(t)dy^\theta_t\;,\label{eq:dynTC}
\eeq
where the dynamics of $y^\theta$ are given in the below lemma.
\begin{lemma}
\label{lem2} Let $\Theta$ be a clock and consider a base model $y$ with dynamics \eqref{eq:ySDE}. Then, the dynamics of $y^{\theta}$ take the form
\begin{equation}
\label{eq:ytheta}
dy^{\theta}_t = \mu\left(\Theta(t),y^{\theta}_t\right)\theta(t)dt + \sigma\left(\Theta(t),y^{\theta}_t\right)\sqrt{\theta(t)}dB_t + dJ^\theta_{t},\quad y^\theta_0=y_0
\end{equation}
where $B$ is an $\mathbb{F}^\theta$-Brownian motion, $\mathbb{F}^\theta:=(\mathcal{F}_{\Theta(t)},t\in[0,T])$ and $J^\theta$ an inhomogeneous compounded Poisson process with jump size mean $\alpha$ and time-$t$ intensity $\omega\theta(t)$.
\end{lemma}
\begin{proof}
By definition, we have 
\begin{equation*}
y^\theta_t:=y_{\Theta(t)} = y_0 + \int_0^{\Theta(t)}\mu(u,y_u)du+\int_0^{\Theta(t)}\sigma(u,y_u)dW_u+\int_0^{\Theta(t)}dJ_u\;.
\end{equation*}
Hence,
\begin{equation*} \int_0^{\Theta(t)}\mu(u,y_u)du=\int_0^t\mu\left(\Theta(u),y_{\Theta(u)}\right)\theta(u)du=\int_0^t\mu\left(\Theta(u),y^\theta_u\right)\theta(u)du\;,
\end{equation*}
and
\begin{equation*}
\int_0^{\Theta(t)}\sigma(u,y_u)dW_u=\int_0^t\sigma\left(\Theta(u),y_{\Theta(u)}\right)dW_{\Theta(u)}=\int_0^t\sigma\left(\Theta(u),y^{\theta}_u\right)\sqrt{\theta(u)}dB_u\;.
\end{equation*}
Indeed, $\Theta$ is a clock, hence $\theta>0$ and the process $W^\theta:=(W_{\Theta(t)},t\in[0,T])$ is a local martingale with quadratic variation $\langle W^{\theta},W^{\theta}\rangle_{t}=\Theta(t)$. From  \cite{jyc:3m}, the process $B:=(B_t,t\in[0,T])$ defined as
\beq
B_t:=\int_0^t\frac{1}{\sqrt{\theta(u)}}dW_{\Theta(u)}\label{eq:TCB}
\eeq
is then a Brownian motion. Differentiating $y^{\theta}_t$ leads to \eqref{eq:ytheta}. With regards to the compounded Poisson process, notice that $dJ_t=\zeta_{N_t}dN_t$ and  $dJ^\theta_{t}=dJ_{\Theta(t)}=\zeta_{N_\Theta(t)}dN_{\Theta(t)}$. The process $N^\theta$ defined as $N^\theta_t:=N_{\Theta(t)}$ is a Poisson process with instantaneous intensity $\omega\theta(t)$. Hence, the dynamics of $J^\theta$ are given by $J^\theta_0=0$ and $\zeta_{N^\theta_t}dN^\theta_t$, so that $J^\theta$ is a compounded Poisson process with jump size mean $\alpha$ and time-$t$ instantaneous rate of jumps arrival, $\omega\theta(t)$.
\end{proof}

This model looks appealing for several reasons. First, just as the shift extension, it is a deterministic adjustment of a base model and is hence expected to be tractable when the latter is, say, a HAJD. Second, because $x^\theta_t$ is a positive rescaling of the process $y$ sampled at time $\Theta(t)$, the range of $x^\theta$ is linked to that of $y$. In particular, if the range of $y$ is $\mathbb{R}$, as for the Vasicek model, then so is the range of $x^\theta$. However, if $y$ is non-negative as in the (J)CIR case, then so is $x^\theta$. Hence, this solves the drawback of the shift approach related to Problem \ref{prob:prob1}. Eventually, the time-dependent feature of the clock rate $\theta$ is expected to provide additional flexibility in the calibration properties of $x^\theta$ with respect to that of the homogeneous model $y$. Two questions remain open in this respect. First, we need to clarify the circumstances under which the model provides a perfect fit. Second, in the case where the perfect fit can be achieved, we need to provide an efficient procedure to compute the resulting ``optimal clock'', $\Theta^\star$. The price to pay is that, in contrast with the shift extension, the time-changed model is not fully flexible. Indeed, starting with a given model $y$, the $x^\theta$ model can only generate specific shapes for discount curves. We are thus more dependent on the initial choice of the base model $y$. Fortunately, it turns out that a perfect fit is achievable for a wide set of market curves, including all decreasing discount curves, considered in Problem \ref{prob:prob1}. This is clearly the most important case since (i) it corresponds to the case where the shift approach fails to provide a convincing solution and (ii) it is probably the most common case in practice, since it encompasses the class of discount curves with non-negative rates (or, more generally, with non-negative instantaneous forward rates), as well as the set of all continuous survival probability curves. Moreover, even if the mathematical expression of the clock $\Theta^\star$ is not available in closed form, its numerical computation turns out to be easy.  This leads us to the first fundamental result of the paper.\footnote{When no confusion is possible, the explicit reference to the model parameters $\Xi$ is avoided to ease the notations.}

\begin{theorem}\label{th:th1}
Let $P^{market}$ be a discount curve and $y$ a model such that $P^y$ is a discount curve. Define the $x^\theta$-model as in \eqref{eq:xytc}. Then, $P^{x^\theta}\equiv P^{market}$ provided that $\Theta\leftarrow \Theta^\star$  where $\Theta^\star$ satisfies the first-order ODE
    \beq
\theta^\star(t):=\frac{d}{dt}\Theta^\star(t)=\frac{f^{market}(t)}{f^y(\Theta^\star(t))}\;,\label{eq:ode}
\eeq
with $f^{market},f^y$ the corresponding instantaneous forward curves. Moreover, if $P^{market}$ and $P^y$ are strictly decreasing, the solution to \eqref{eq:ode} exists, is a clock, and is given by
    \beq
\Theta^\star(t):=Q^y\left(P^{market}(t)\right)\;,\label{eq:ThetaStar}
\eeq
where $Q^y$ is the inverse of the base-model discount curve, $P^y$.
\end{theorem}

\begin{proof}See Section  \ref{seq:ProofTh1}.
\end{proof}

Observe that the optimal clock $\Theta^\star$ actually depends from the $y$-model parameters $\Xi$. Just like for the shift, we actually have $\Theta^\star(t)=\Theta^\star(t;\Xi)$. Although other frameworks are possible, we set $\Xi=\Xi^\star$  as in \eqref{eq:CalyNoConstraint}. Similar to the function $\varphi$ in the shift approach, the purpose of the clock $\Theta$ is then to absorb the remaining errors between $P^y(\cdot;\Xi^\star)$ and $P^{market}$. 

\subsection{Time-changed homogeneous affine diffusions}\label{sec:THAD}

As in the shift extension, a time-changed model $x^\theta$ enjoys a similar tractability level to that of the base model $y$. Indeed the $k$-th moment is $m^{x^\theta}(k,t)=\theta(t)^k m^{y}\left(k,\Theta(t)\right)$ and moment generating function is $\psi^{x^\theta}(u,t)=e^{u\theta(t)}\psi^{y}\left(u,\Theta(t)\right)$, whereas those of $X^\theta_t$ coincide with those of $Y_{\Theta(t)}$. Hence, a tractable model $x^\theta$ can be obtained by considering HAJD processes as base model $y$. We illustrate our method by analyzing two calibration problems that can be solved by considering the Vasicek and the JCIR processes.

It is clear from Lemma \ref{lem2} that in the particular case where $y$ is a HAJD, then $x^\theta$ is a scaled version of an inhomogeneous affine jump diffusion (AJD), unless $\theta(t)$ is a positive constant, in which case it remains a HAJD. To see this, suppose that the dynamics of $y$ obey \eqref{eq:HAJD}. From Lemma \ref{lem2}, $y^\theta$ is governed by
\beq
dy^{\theta}_t = \left(a(\Theta(t)) + b(\Theta(t))y^{\theta}_t\right)\theta(t)dt + \sqrt{\left(c(\Theta(t))+d(\Theta(t))y^{\theta}_t\right)\theta(t)}dB_t + dJ^{\theta}_t\;.\label{eq:TCHAJD}
\eeq

Interestingly, $y^\theta$ is still an AJD. 
In the sequel, we focus on the special case where the base model $y$ is a HAJD, i.e., takes the form \eqref{eq:HAJD} with constant parameters $(a(t),b(t),c(t),d(t),\alpha,\omega(t))=(\kappa\beta,-\kappa,\eta^2,\delta^2,\alpha,\omega)$. To simplify the notation, we specify the model parameters using the vector $\Xi=(\kappa,\beta,\eta,\delta,\alpha,\omega)$.

\subsubsection{Time-changed Vasicek}

Our time change approach can be easily used to solve Problem \ref{prob:prob2} in the most common case where the forward curve $f^{market}$ is arbitrary (monotonic, humped, etc) provided that it is positive. As there is no constraint on the range of the process $x^\theta$, let us postulate Vasicek dynamics for the base process with parameters $\Xi=(\kappa,\beta,\eta,0,0,0)$ :
\begin{equation}
dy_t=\kappa(\beta-y_t)dt+\eta dW_t,\quad y_0\in\mathbb{R}\;.\nonumber
\end{equation}

The forward curve associated to this model is given by $f^{y}(t)=f^{\mathrm{VAS}}(t) :=f_0^{\mathrm{VAS}}(t)$ in \eqref{eq:fvast}: 
\beq
f^{\mathrm{VAS}}(t)= (1-\e^{-\kappa t})\frac{\kappa^2\beta-\eta^2/2}{\kappa^2}+\frac{\eta^2}{2\kappa^2}\e^{-\kappa t}(1-\e^{-\kappa t})+y_0\e^{-\kappa t}\;.\label{eq:fvas}
\eeq

It can thus be used to select an appropriate Vasicek model. The next corollary provides guidelines to generate decreasing discount curve, associated with the most common case of positive instantaneous forwards.

\begin{corollary}\label{cor:cor1} Let $P^{market}$ be a strictly decreasing market curve. Then, for every Vasicek model with parameters satisfying $y_0\geq 0$ and $2\kappa^2\beta>\eta^2$ 
, there exists a clock $\Theta^\star$ such that $P^{x^\theta}\equiv P^{market}$.
\end{corollary}
\begin{proof} Because $y$ is a Vasicek process, $P^y$ is a discount curve. Moreover, the conditions $y_0\geq 0$ and $2\kappa^2\beta>\eta^2$ guarantee that the forward curve \eqref{eq:fvas} is strictly positive, hence $P^y$ is strictly decreasing. From Theorem \ref{th:th1}, the clock $\Theta^\star$ exists and is given by \eqref{eq:ode} with $f^y$ given in \eqref{eq:fvas}.
\end{proof}

Notice that the dynamics of the time-changed Vasicek model $x^{\theta}_t$ are given by \eqref{eq:dynTC} with
\begin{equation}
dy^{\theta}_t = \kappa(\beta - y^{\theta}_t)\theta(t)dt + \eta\sqrt{\theta(t)}dB_t,\quad y^{\theta}_0=y_0\;,\nonumber%
\end{equation}
showing that $y^\theta$ remains a Gaussian process. Fitting perfectly a strictly decreasing discount curve (without further constraints on the process) is a special case of Problem \ref{prob:prob2}, that can also be solved using the shift approach \eqref{eq:xyshift} by taking $x\leftarrow x^\varphi$ where $y$ is a Vasicek with arbitrary parameters $\Xi$ and $\varphi(t)\leftarrow\varphi^\star(t;\Xi)=f^{market}(t)-f^{\mathrm{VAS}}(t)$. The main interest of the time-changed approach is actually when considering Problem \ref{prob:prob1}. 

\subsubsection{Time-changed (J)CIR}\label{sec:tc-cir}

The following result is the second main contribution of the paper. It shows that the time change approach $x\leftarrow x^\theta$ provides a solution to Problem \ref{prob:prob1}.

\begin{corollary}\label{cor:cor2} Let $y$ be an almost-surely positive HAJD with parameters $\Xi$. Then, the model $x^\theta$ defined in \eqref{eq:xytc} with $\Theta\leftarrow\Theta^\star(t;\Xi)$ solves Problem \ref{prob:prob1}.
\end{corollary}
\begin{proof} Because $y$ is a HAJD, $P^y$ is a discount curve and is tractable analytically. Moreover, the latter is strictly decreasing since $y$ is almost-surely positive. We conclude the proof by relying on Theorem \ref{th:th1}.
\end{proof}

Let us now consider the JCIR model, i.e., the HAJD with $\Xi=(\kappa,\beta,0,\delta,\alpha,\omega)$. The CIR is recovered as a special case by choosing $(\alpha,\omega)$ such that $\alpha\omega=0$.\medskip

Then,

\begin{equation}
dy_t=\kappa(\beta-y_t)dt+\delta\sqrt{y_t} dW_t + dJ_t,\quad y_0>0\;,\label{eq:JCIR}
\end{equation}
where $\kappa,\beta,\delta$ are strictly positive constants and $\omega,\alpha$ are non-negative. The optimal clock $\Theta^\star$ leading to the perfect fit to a given strictly decreasing curve $P^{market}$ is given by \eqref{eq:ode}  where the forward curve associated to this model is given by $f^y(t)=f^{\mathrm{JCIR}}(t):=f_0^{\mathrm{JCIR}}(t)$ in   \eqref{eq:fjcirt} : 
\beq
f^{\mathrm{JCIR}}(t)=\frac{2\kappa\beta(e^{t\gamma}-1)}{2\gamma+(\kappa+\gamma)(e^{t\gamma}-1)} + y_0\frac{4\gamma^2e^{t\gamma}}{[2\gamma+(\kappa+\gamma)(e^{t\gamma}-1)]^2}+\frac{2\omega\alpha(e^{t\gamma} - 1)}{2\gamma + (\kappa+\gamma+2\alpha)(e^{t\gamma} - 1)} \;,\label{eq:fjcir}
\eeq
where $\gamma:=\sqrt{\kappa^2+2\delta^2}$. 
The dynamics of the time-changed process $x^{\theta}_t=\theta(t)y^{\theta}_t$ are given by \eqref{eq:dynTC} with  
\begin{equation}
dy^\theta_t = \kappa(\beta - y^\theta_t)\theta(t)dt + \delta\sqrt{\theta(t)y^\theta_t}dB_t+dJ^\theta_t,\quad y^\theta_0=y_0\;.\nonumber%
\end{equation}
where $B$ is an $\FF^\theta$-Brownian motion and $J^\theta$ is an inhomogeneous compound Poisson process with jump size mean $\alpha$ and time-$t$ instantaneous rate of arrival $\omega\theta(t)$.

The time change technique applied to a JCIR (TC-JCIR) therefore solves Problem \ref{prob:prob1}. In particular, in contrast to the S-JCIR (that focuses on parameters such that $\varphi$ is positive), the positivity constraint on $x^\theta$ is automatically satisfied for every (strictly decreasing) market curve and every $\Xi$ (such that $y$ is not trivially equal to 0). However, we have shown that it is possible to ensure positivity by considering the PS-JCIR, $x^{\varphi,+}$. Working with $\Xi^{\star,+}$ instead of $\Xi^\star$ can make the job, but at the expenses of having a process $x^{\varphi,+}$ that is, to a large extend, deterministic (i.e., $x^{\varphi,+}_t$ varies in a small neighborhood around $f^{market}(t)$). Consequently, TC-JCIR model are expected to feature a higher volatility compared to the corresponding PS-JCIR, at least up to some time horizon. This is summarized in the next theorem, which is the third main result of the paper.

\begin{theorem}\label{th:th2}
Let $P^{market}$ be a strictly decreasing discount curve and $y$ be a JCIR++ process with parameter $\Xi$ such that the perfect fit JCIR++ model
$x^{\varphi^{\star}}_t$ is positive. Then, the ODE \eqref{eq:ode} with $f^y(t)=f^{\mathrm{JCIR}}(t;\Xi)$ given by \eqref{eq:fjcir} admits a solution that satisfies $\Theta^\star(t)=\Theta^\star(t;\Xi)\geq t\;.$ Moreover, the variance of the corresponding perfect fit TC-JCIR model $x^{\theta^\star}_t$ satisfies:
\begin{itemize}
\item[1)]$\mathbb{V}\left[X^{\theta^\star}_t\right]\geq\mathbb{V}\left[X^{\varphi^\star}_t\right]$, $\forall\; t\geq 0$,
\item[2)] $\mathbb{V}\left[x^{\theta^\star}_t\right]\geq\mathbb{V}\left[x^{\varphi^\star}_t\right]$ if one of the following holds:
\begin{itemize}
\item[i)] $y_0=\beta + \frac{\omega\alpha}{\kappa}$,
\item[ii)] $f^{market}$ constant and $y_0\leq\beta + \frac{\omega\alpha}{\kappa}$,
\item[iii)] $y_0>\beta + \frac{\omega\alpha}{\kappa}$ and $t<{\Theta^\star}^{-1}(t_1)$,
\item[iv)] $(\kappa\beta + \omega\alpha)/\gamma < y_0<\beta + \frac{\omega\alpha}{\kappa}$ and $t>{\Theta^\star}^{-1}(t_2)$
\end{itemize}
\end{itemize}
where
\begin{equation*}
t_1:=\frac{1}{\kappa}\ln\left(1+ \frac{y_0+2\omega\alpha^2/\gamma^2}{y_0-\beta-\omega\alpha/\kappa}\right)\quad \text{and}\quad t_2:=\frac{1}{\gamma}\ln \frac{(\gamma-\kappa)(\kappa\beta+y_0\gamma+\omega\alpha)-2\omega\alpha^2}{(\kappa+\gamma)(y_0\gamma-\kappa\beta-\omega\alpha)-2\omega\alpha^2}.
\end{equation*}
\end{theorem}
\begin{proof} See Section \ref{seq:ProofTh2}. 
\end{proof}

To sum up, the TC-JCIR model (including the TC-CIR) provides an elegant solution to Problem \ref{prob:prob1}: the process $x^{\theta^\star}$ is non-negative (in contrast with the S-JCIR $x^{\varphi^\star}$), is almost as tractable as the simple JCIR diffusion (in contrast with the PS-JCIR $x^{\varphi^{\star,+}}$), provides a perfect fit to every strictly decreasing discount curve (as both JCIR++ models) and features, to some extend, a larger variance (compared to the PS-JCIR $x^{\varphi^{\star,+}}$). In particular, it is observed, empirically, that the variance of the integral of the TC-JCIR remains similar to that of the unconstrained (i.e., flawed, but high-volatility) S-JCIR model $x^{\varphi^{\star}}$. Therefore, when a positivity constraint is required, the TC-JCIR avoids the drawbacks of the JCIR++ models. The only price to pay is that the clock is not available in closed form, but requires a (simple) numerical inversion. The properties of the model, namely the perfect fit and high-variance features, are illustrated in the next section on various applications taken from credit risk modeling.

\section{Application to Credit Risk Modelling}
\label{sec:app}

We consider a reduced-form default model as in Section \ref{sec:defaultable} by using a CIR base model $y$ (i.e., \eqref{eq:JCIR} with $J\equiv 0$). The default intensity $\lambda$ is modelled either as a CIR++ ($\lambda\leftarrow\lambda^{\varphi}_t:=y_t+\varphi(t)$) or using the TC-CIR ($\lambda\leftarrow\lambda_t^{\theta}:=\theta(t)y_{\Theta(t)}$). Observe that depending on the pair ($P^{market},\Xi$), the CIR++ process can feature negative values. This will be the case when taking $\Xi\leftarrow \Xi^\star$ given using the MSE approach \eqref{eq:CalyNoConstraint}, unless there is an explicit constraint as in \eqref{eq:CalyConstraint}, leading to take $\Xi\leftarrow \Xi^{\star,+}$. Bear in mind that when $\lambda$ represents an intensity process, the S-CIR model ($\lambda^\varphi$) is actually flawed as there is a non-zero probability to observe negative intensities, and $P^{\lambda^\varphi}(t)$ cannot be interpreted as a survival probability associated to a Cox model. Yet, we give the results of the model as a benchmark since, as explained in the introduction, it is a very standard approach.\medskip

We compare the CIR++ (S-CIR and PS-CIR) to the TC-CIR on several aspects related to a real case example where the reference entity is Ford Inc. We also discuss the TC-JCIR case when relevant. We first analyze the perfect fit feature of both types of models, as well as the non-negativity property of $\lambda$. We then compare the variance of the integrated processes $\Lambda$. We then analyse their behaviors in two different applications, namely the pricing of various credit default swaptions (a.k.a. CDS options, or CDSO) with Ford as reference entity, or on the credit valuation adjustment (CVA) of prototypical FRA and IRS exposures where Ford is the trade counterparty.\medskip

It it well-admitted that ``pure credit instruments'' like CDS or CDSO are quite insensitive to 
the stochasticity of the interest rates in realistic conditions. 
This has been discussed explicitly for the CIR base model in  \cite{Brigo05} and \cite{Brigo06a}. Hence, we consider a deterministic short rate process, which is stressed by the notation $r_u=r(u)$. In this case, one simply gets $P_s^r(t)=D_s(t)=e^{-\int_s^t r(u) du}$.\footnote{Given that the interest rates have little impact on the figure and that our main objective is to discuss the impact of the default model, we considered zero risk-free rate in the numerical applications below.}\medskip

In the sequel, we first illustrate the perfect fit feature of S-CIR, PS-CIR and TC-CIR when the default model is calibrated on the survival probability curve of Ford Inc. We then use the model to price CDSO and compute CVA figures.

\subsection{Perfect fit of CDS term-structure}
\label{sec:example}
 
We consider the CDS term-structure of Ford Inc, and show that considering a set of parameter $\Xi$, there exist $\varphi$ and $\Theta$ that yield a perfect fit. In the sequel, we drop the star superscript on the shift and clock functions. Hence, $\Xi^\star$ corresponds to the CIR parameter optimized without constraint to a given $P^{market}$ curve, and $\varphi$ and $\Theta$ refer to the corresponding optimal shift and clock functions. The corresponding parameters found under a non-negativity constraint are noted $\Xi^{\star+},\varphi^+$ and $\Theta^+$, respectively. \medskip

A credit default swap (CDS) is a financial instrument used by two parties \--- called the protection buyer and the protection seller \--- to transfer to the protection seller the financial loss that the protection buyer would suffer if a particular default event happened to a third party called the \textit{reference entity}. Typically, we set $\tau$ as the default time of the latter. In a default swap contracted at time $t$, started at time $T_a$ with maturity $T_b$, the protection buyer pays a coupon (of spread) $k$ at a set of payment dates $T_a,\ldots,T_b$ as long as the reference entity does not default. The protection seller agrees to make a single payment $LGD$ to the protection buyer if the default occurs between $T_a$ and $T_b$. When applicable, the protection buyer makes a final payment corresponding to the spread accrued since the last payment date before default. For more details about the mechanics of this product, we refer to \cite{Brigo05} and \cite{Brigo10}.\footnote{For more details about the actual market conventions, we refer the interested reader to \cite{IsdaCDS} and \cite{CDSbigbang}.}\medskip 

The CDS term-structure consists of a set of par spreads associated with CDS of various maturities. The time-$t$ par spread $s_t(T_i)$ of a CDS contract of maturity $T_i$ is defined as the contract spread $k$ that sets the value of the CDS contract to 0 at time $t$. The par spreads have been taken from Bloomberg on November 12, 2018 and are shown on the table below.

\begin{table}[H]
    \centering
    \begin{tabular}{lccccc}
    \hline
     Maturity (years)    & 1 & 3 & 5 & 7 & 10\\
     \hline
     Spread (bps)    & 18.3 & 136.6 & 191.9 & 267.6 & 280.6\\
     \hline
    \end{tabular}
    \caption{CDS spread term structure of Ford Inc. on November 12, 2018. Source: Bloomberg.}
    \label{tab:spread}
\end{table}

In this context, the market curve $P^{market}$ to be fitted is the risk-neutral survival probability curve, defined as $G(t):=\Q(\tau>t)$ associated with the default time $\tau$ of a given reference entity (here, Ford Inc.). It can be extracted from CDS quotes by inverting the no-arbitrage pricing formulae of the corresponding financial instruments. In practice, one only has a couple of calibration equations, say $n$, given by the number of market quotes (here, $n=5$). It is therefore not possible to estimate the full (i.e., infinite-dimensional) market curve $G$ without further assumptions. It is common market practice to consider the CDS model from the International Swap and Derivative Association (ISDA) \--- a.k.a the JP Morgan model \--- \cite{IsdaCDS}, that provides a slightly simplified version of the actual no-arbitrage pricing formula applying to CDSs. In this approach, the curve $G$ is parametrized via a positive hazard rate function $h$, playing a similar role as the instantaneous forward rate $f^{market}$,
$$G(t):=e^{-\int_0^t h(s) ds}\;,$$
where $h$ is itself parametrized by $n$ constants $h_1,h_2,\ldots,h_n$ bootstrapped from the spreads $s_1,s_2,\ldots,s_n$ associated with the maturities $T_1,T_2,\ldots,T_n$. Let us focus on the horizon $T=T_n$. It is market practice to assume that $h$ is piecewise constant between the maturities, i.e., to postulate the parametric form:
$$h(t)=\sum_{i=1}^{n}\indic_{\{T_{i-1}\leq t<T_i\}}h_{i-1}\;,$$
where $T_0:=0$, $h_0:=\frac{s_1}{1-R}$ with $LGD:=1-R$, $R=40\%$ the assumed recovery rate of the firm and $h_i$'s are positive constants. Even if less standard, another specifications like, e.g., a piecewise linear parametrization could be preferred:
$$h(t)=\sum_{i=1}^{n}\indic_{\{T_{i-1}\leq t<T_i\}}\left[\frac{h_i-h_{i-1}}{T_i-T_{i-1}}(t-T_{i-1})+h_{i-1}\right]\;.$$
These two different specifications of the hazard rate function are considered on panels (a) of Figure \ref{fig:Ford1} and \ref{fig:Ford2}, respectively. These frameworks yield similar (yet, slightly different) market curves $G(t)$ (green curves on panels (d)). For each of them, we start by computing the ``best'' base CIR model $y$. In line with market practice, we take $\Xi\leftarrow\Xi^\star$ using \eqref{eq:Xistar} with $P^{model}\leftarrow P^y$ considering \eqref{eq:MSE} as error function and $\cT$ the set of available liquid CDS maturities  available. In each case, we consider the two adjusted intensity models associated to the optimal shift ($\varphi$, given by \eqref{eq:CalyNoConstraint}) or optimal clock ($\Theta$, given by \eqref{eq:ThetaStar}). The latter are shown on panels (b) and (c), respectively. The model curves $P^{\lambda^{\varphi}}$ (S-CIR) and $P^{\lambda^{\theta}}$ (TC-CIR) are shown in magenta on panels (d); they agree with each other, and collapse to $G(t)$ due to the perfect fit. Notice that the parametrization of the hazard rate function has little importance: the survival probability curves $G$, $P^y$ and $P^\lambda$ are very similar in either cases. Similarly, the clock functions $\Theta$ look very similar in both panels (c) of Fig. \ref{fig:Ford1} and \ref{fig:Ford2}. 

\ifdefined \InclFig
\begin{figure}[H]
\centering
\subfigure[Piecewise constant hazard rate ($h$)]{\includegraphics[width=0.45\columnwidth]{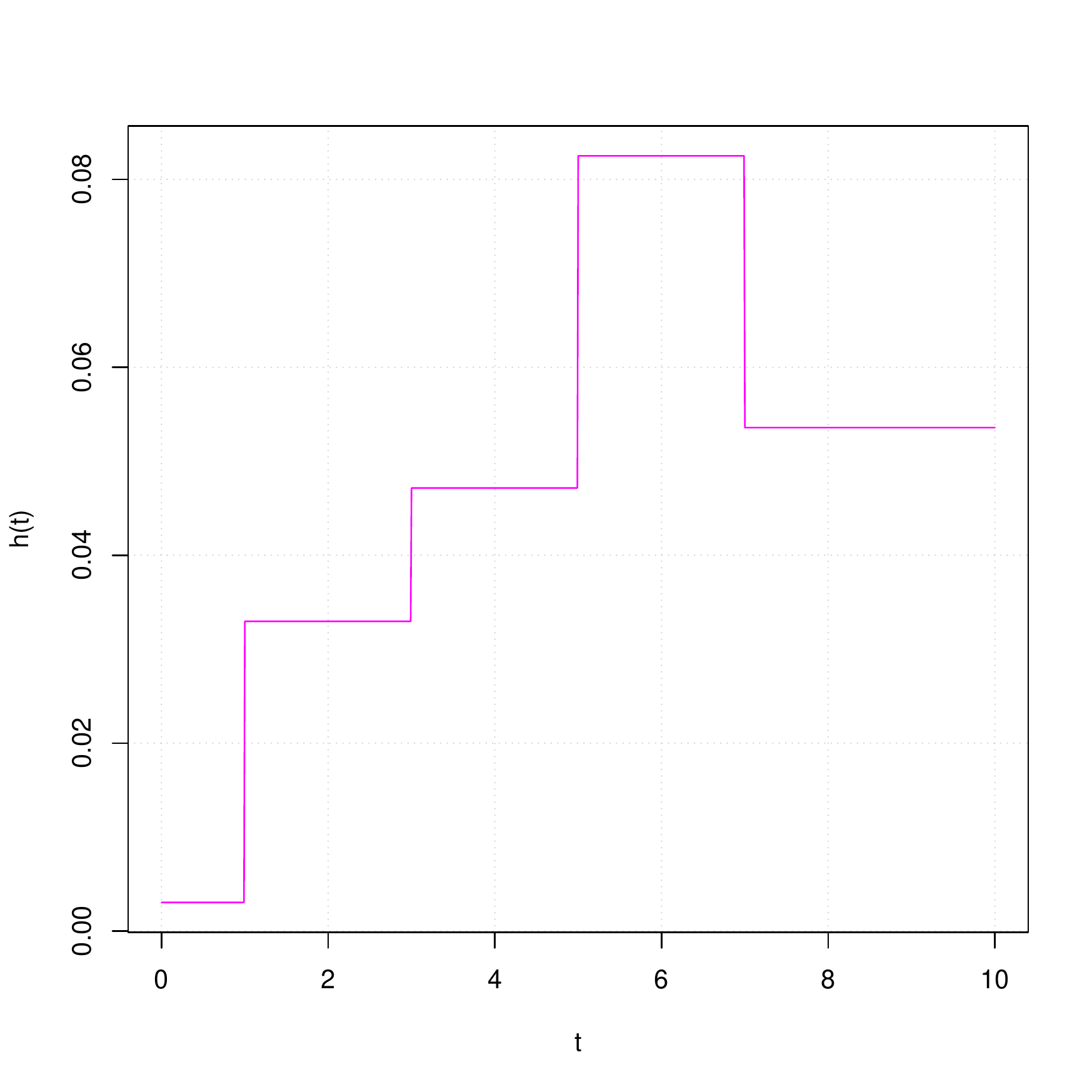}}\hspace{0.2cm}
\subfigure[Optimal shift ($\varphi=\varphi^\star$)]{\includegraphics[width=0.45\columnwidth]{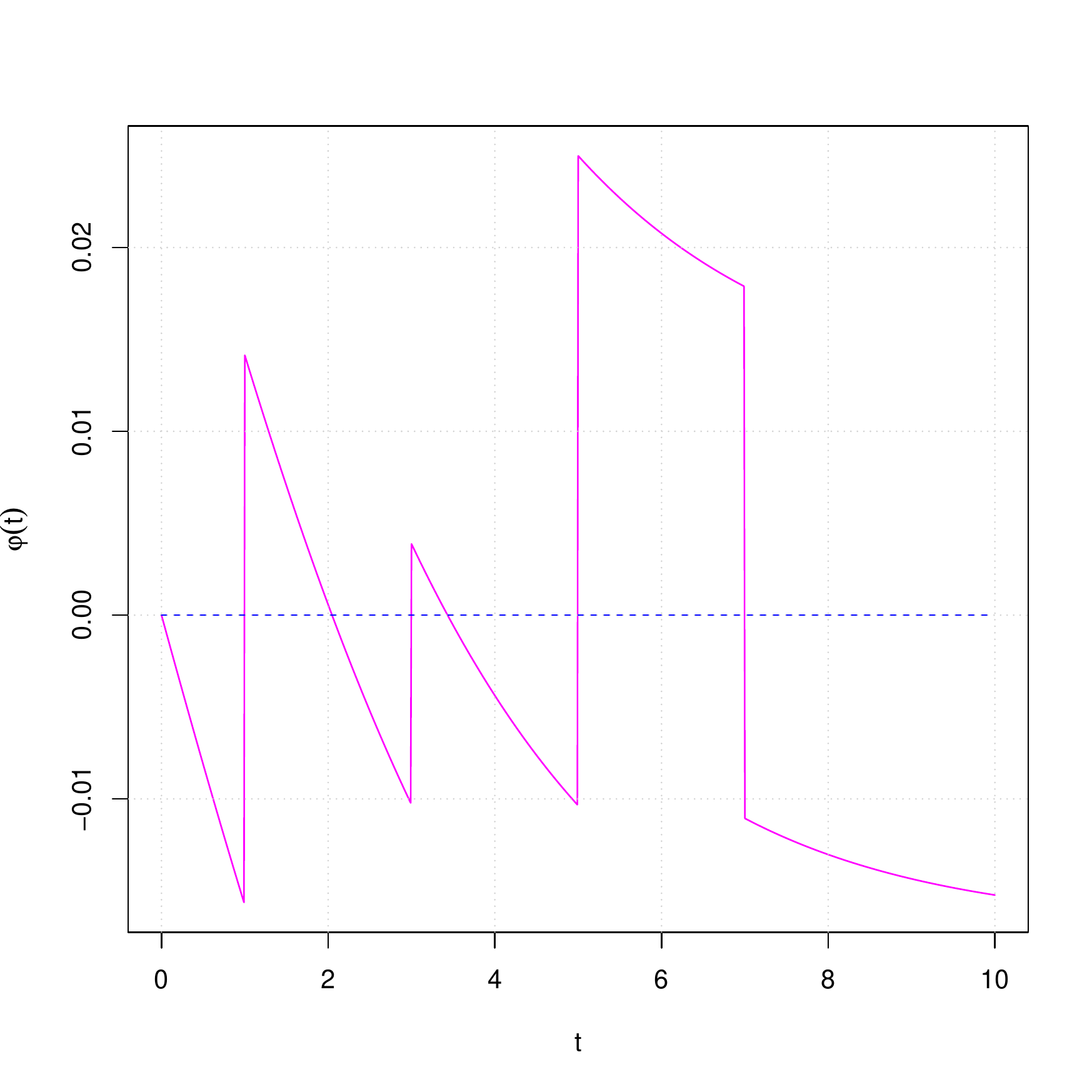}}\hspace{0.2cm}
\subfigure[Optimal clock ($\Theta=\Theta^\star$)]{\includegraphics[width=0.45\columnwidth]{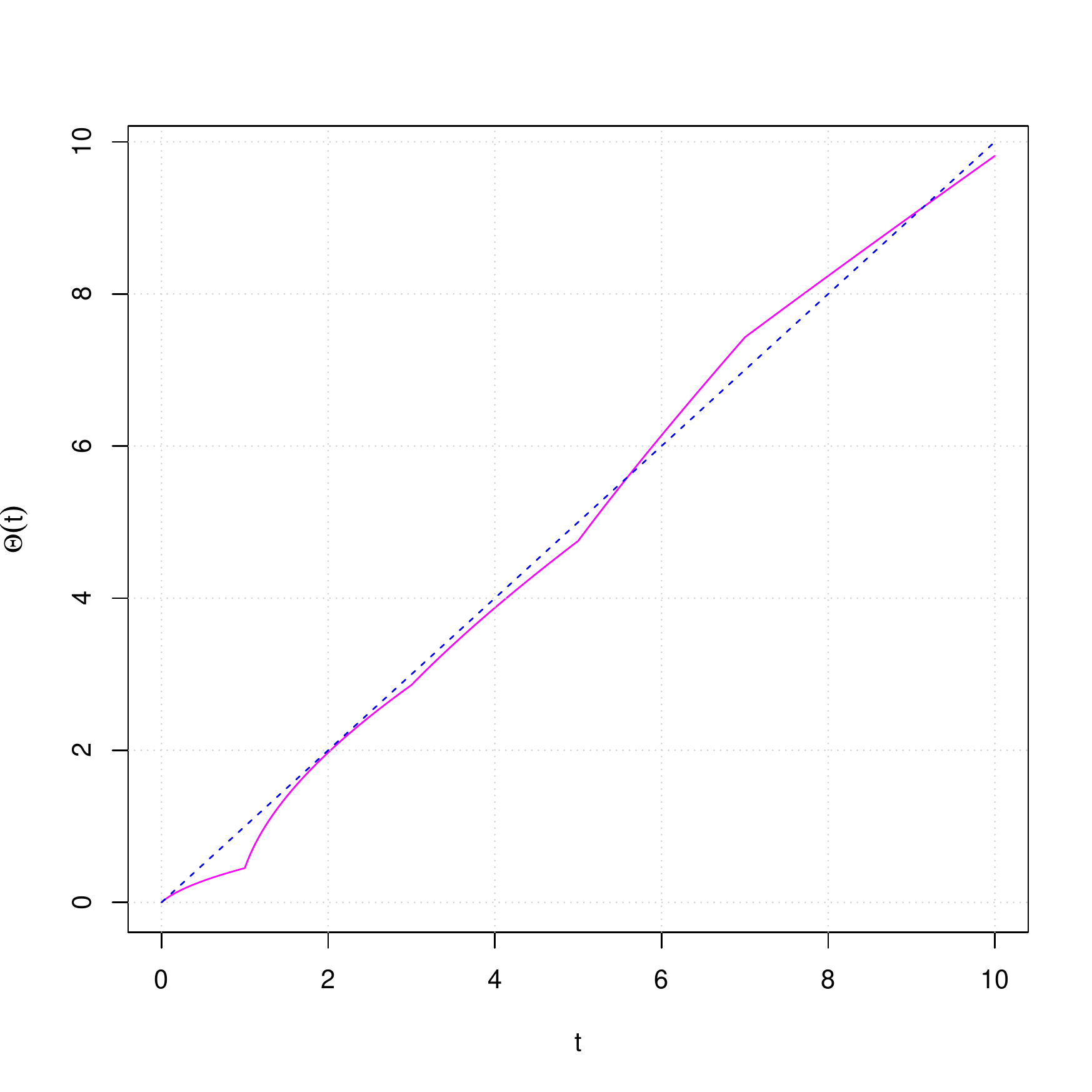}}\hspace{0.2cm}
\subfigure[Survival probability ($P^y,G=P^{\lambda^{\varphi}}=P^{\lambda^{\theta}}$)]{\includegraphics[width=0.45\columnwidth]{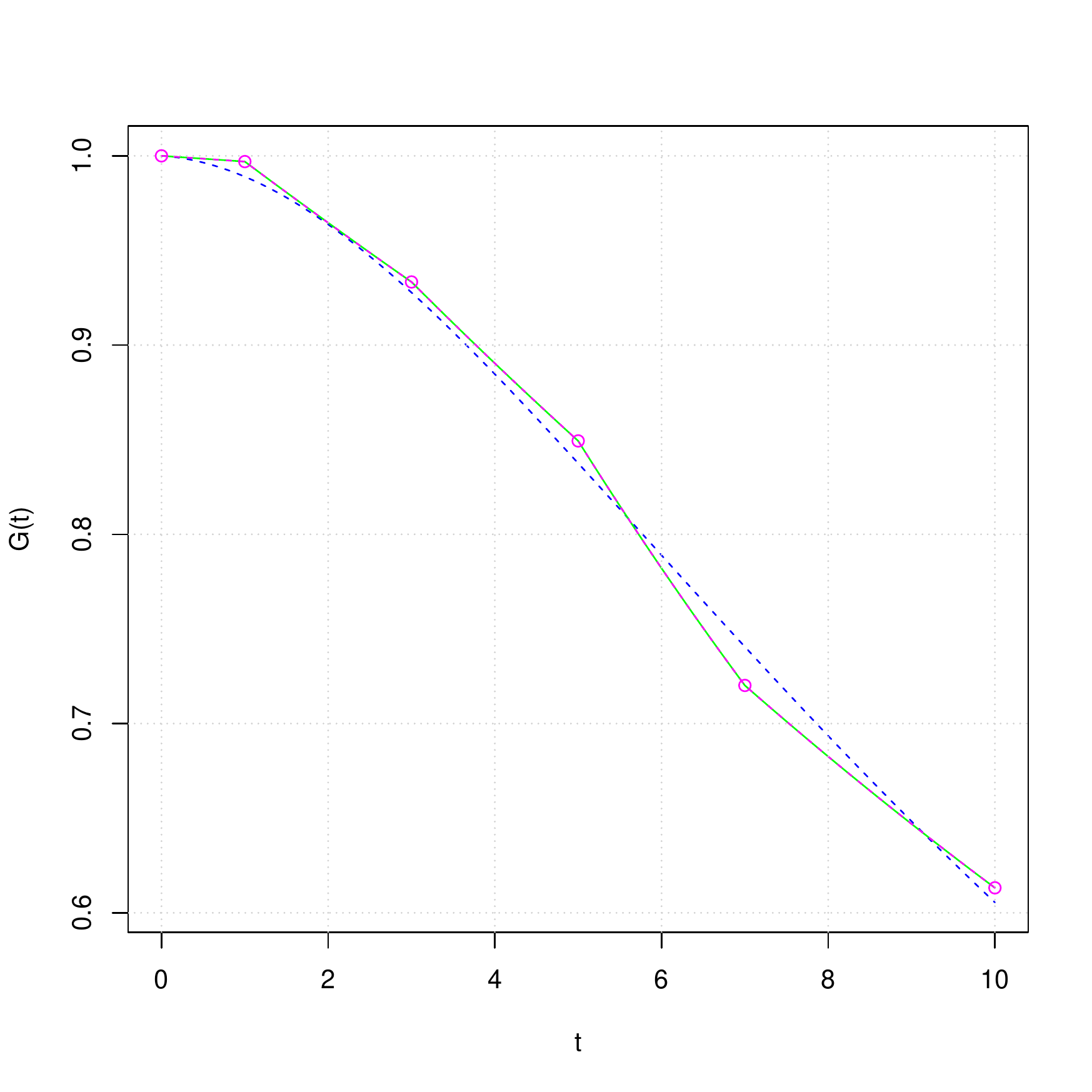}}
\caption{Fitting Ford Inc. CDS term-structure with adjusted CIR models. The survival probability curve $G(t)$ is parametrized with a piecewise constant hazard rate function $h(t)$ extracted from market prices taken from Bloomberg on November 12, 2018, panel (a). The base model $y$ is a CIR with parameters $\Xi=\Xi^\star$ where $\Xi^\star=(\kappa,\beta,\eta,\delta,\alpha,\omega) = (0.0555, 0.3018,0, 0.2939,0,0)$ is obtained from \eqref{eq:CalyNoConstraint} and $y_0=h_0$. The shift function $\varphi(t)\leftarrow\varphi^\star(t,\Xi^\star)$ is shown in panel (b). Panel (c) gives the clock $\Theta(t)\leftarrow\Theta^\star(t;\Xi^\star)$. Eventually, panel (d) yields the survival probability curves given by the market ($G(t)$, green), or associated to $\Q(\tau(\lambda)>t)$ for various intensity models $\lambda$ : the best base model $\lambda\leftarrow y$ (leading to $\Q(\tau(y)>t)=P^y(t,\Xi^\star)$, dashed blue), $\lambda\leftarrow \lambda^{\varphi}$ (S-CIR) and $\lambda\leftarrow \lambda^{\theta}$ (TC-CIR model). By construction of $\varphi$ and $\Theta$, the last two curves coincide (magenta) and agree with $G(t)$.}\label{fig:Ford1}
\end{figure}
\else
\begin{center}
[Figure 2 about here]
\end{center}
\fi 

\ifdefined \InclFig
\begin{figure}[H]
\centering
\subfigure[Piecewise linear hazard rate ($h$)]{\includegraphics[width=0.45\columnwidth]{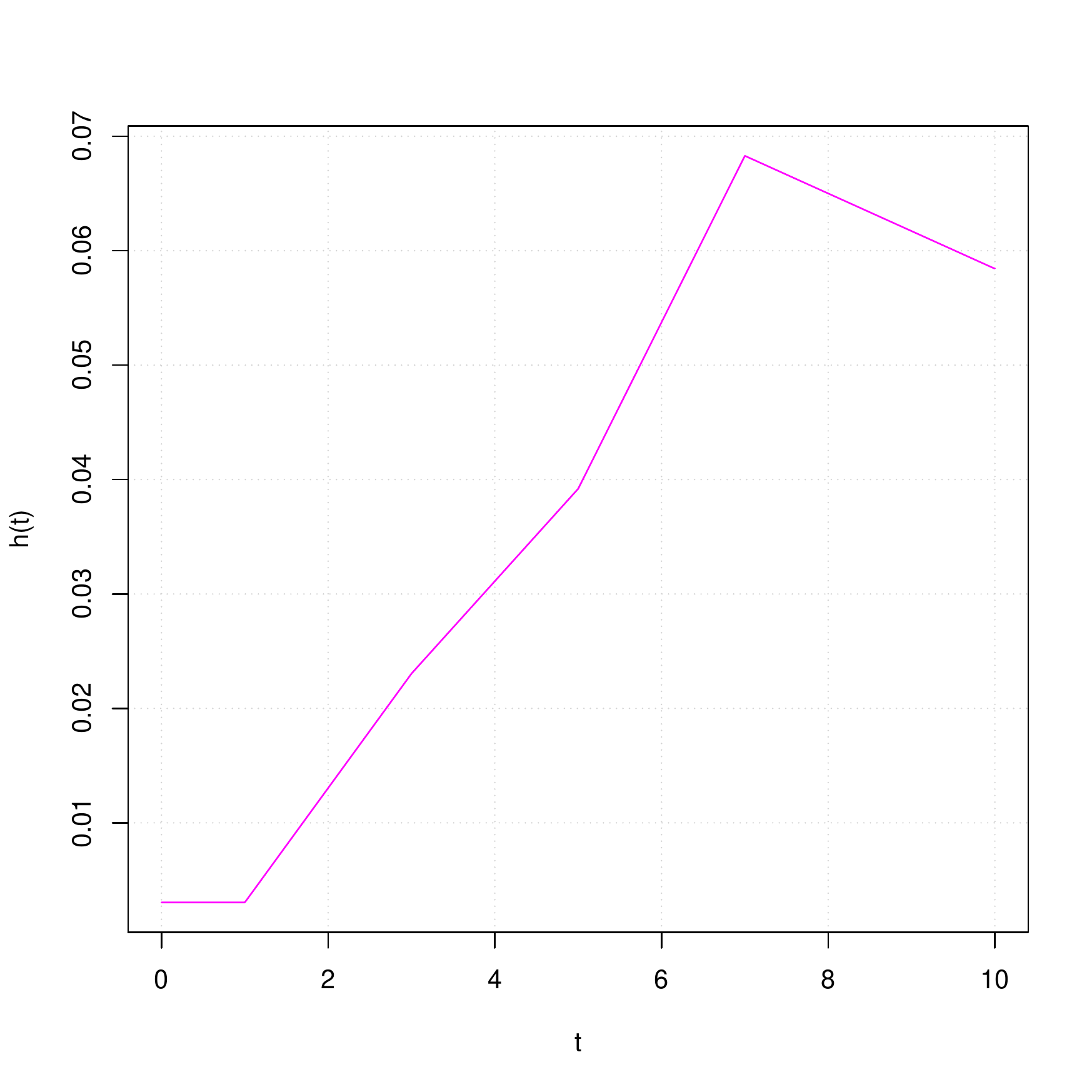}}
\subfigure[Optimal shift ($\varphi=\varphi^\star$)]{\includegraphics[width=0.45\columnwidth]{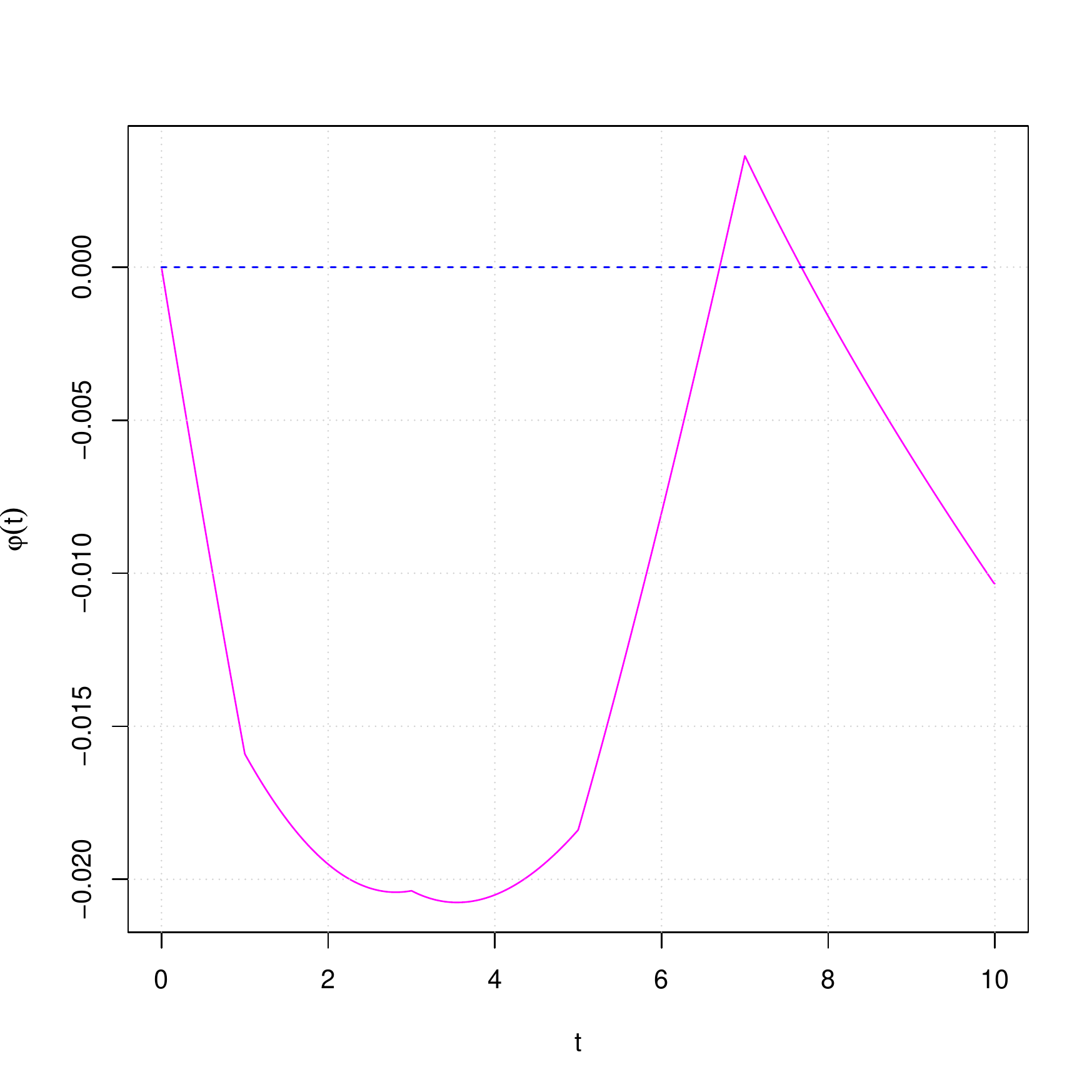}}\hspace{0.2cm}
\subfigure[Optimal clock ($\Theta=\Theta^\star$)]{\includegraphics[width=0.45\columnwidth]{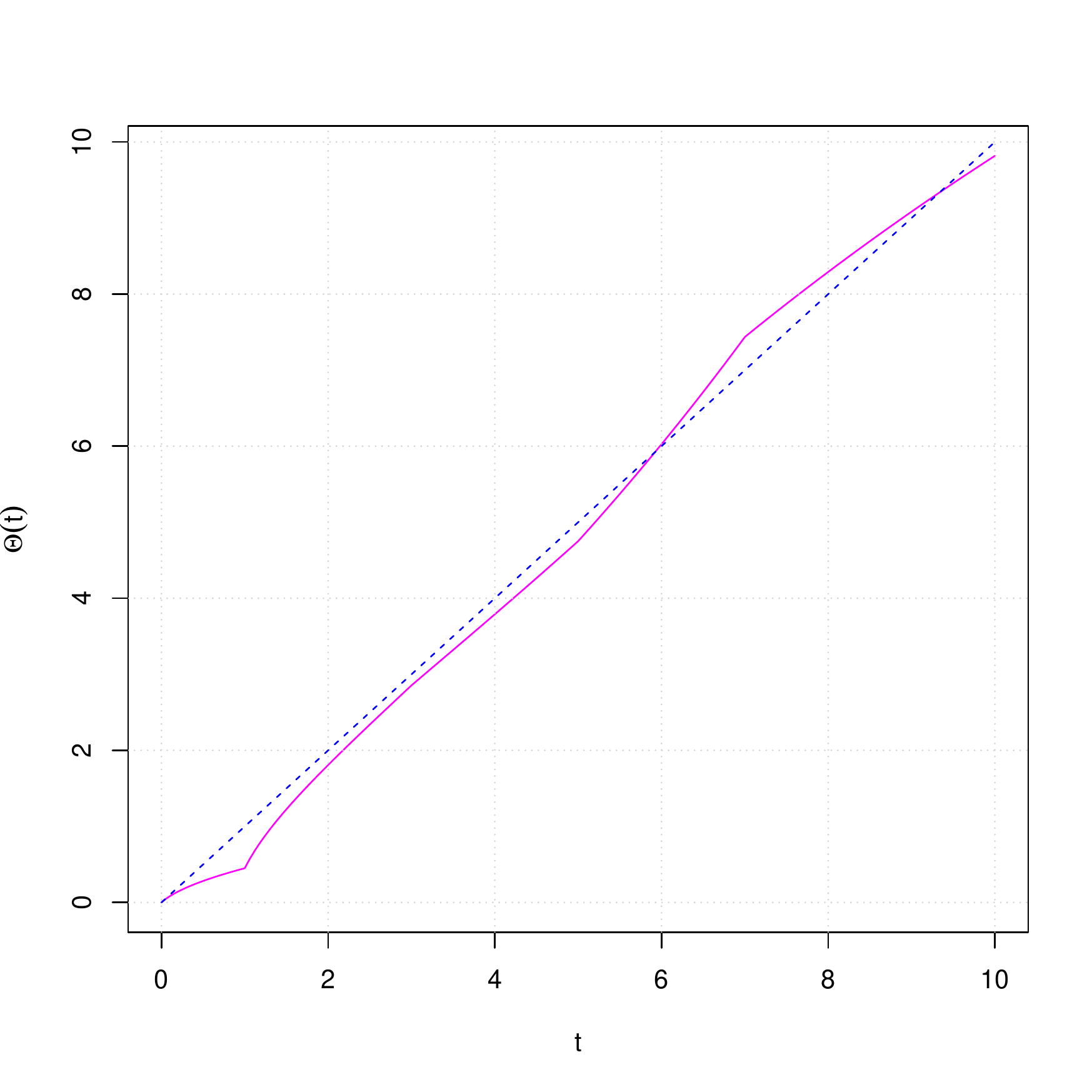}}\hspace{0.2cm}
\subfigure[Survival probability ( $P^y,G=P^{\lambda^{\varphi}}=P^{\lambda^{\theta}}$)]{\includegraphics[width=0.45\columnwidth]{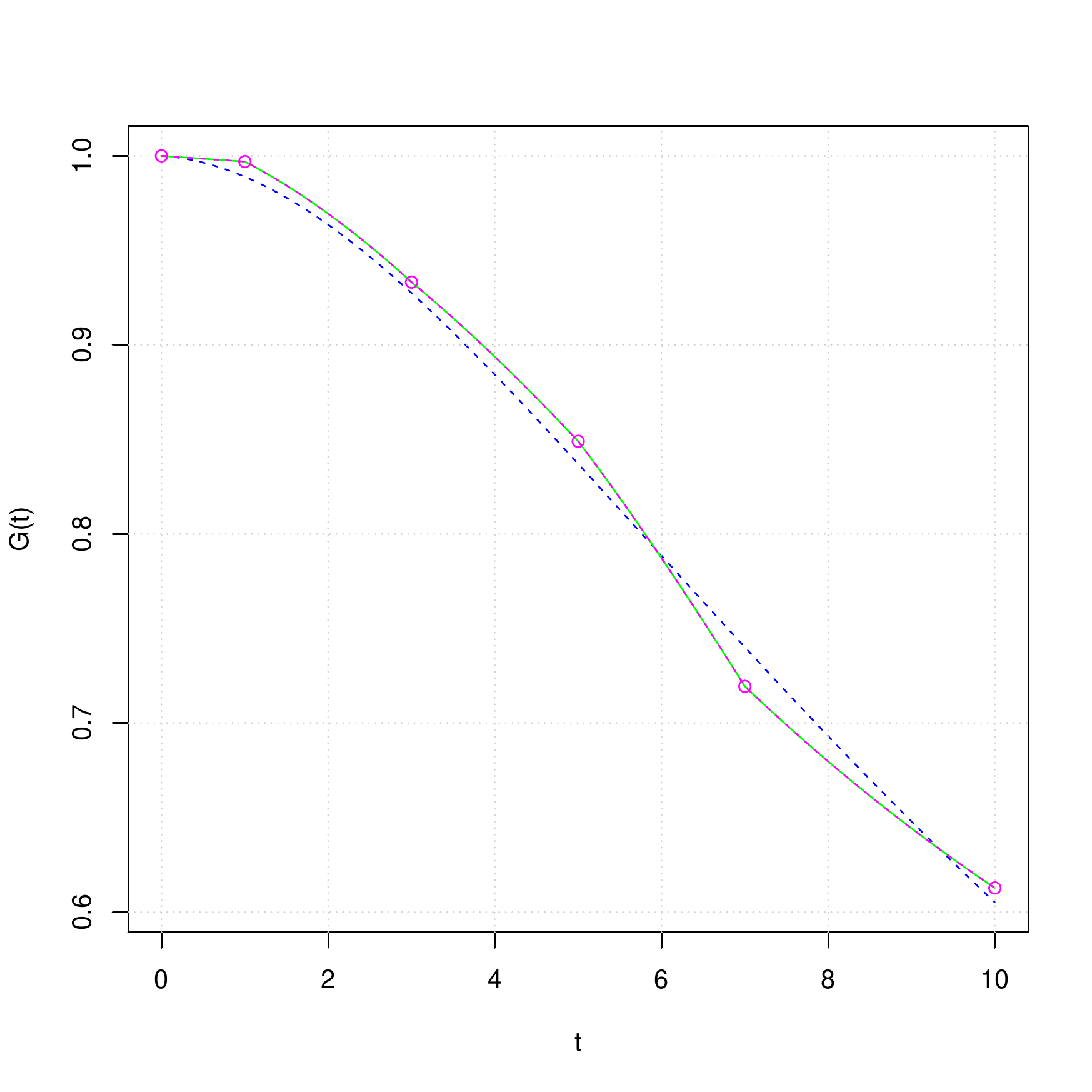}}
\caption{Fitting Ford Inc. CDS term-structure with adjusted CIR. The survival probability curve $G(t)$ is parametrized with a piecewise linear hazard rate function $h(t)$ extracted from market prices taken from Bloomberg on Novermber 12, 2018, panel (a). The base model $y$ is a CIR with parameters $\Xi=\Xi^\star$ where $\Xi^\star=(\kappa,\beta,\eta,\delta,\alpha,\omega) = (0.0620, 0.2729, 0,0.2926,0,0)$ is obtained from eq. \eqref{eq:CalyNoConstraint} and $y_0=h_0$. The shift function $\varphi(t)\leftarrow\varphi^\star(t,\Xi^\star)$ is shown in panel (b). Panel (c) gives the clock $\Theta(t)\leftarrow\Theta^\star(t;\Xi^\star)$. Eventually, panel (d) yields the survival probability curves given by the market ($G(t)$, green), or associated to $\Q(\tau(\lambda)>t)$ for various intensity models $\lambda$ : the best base model $\lambda\leftarrow y$ (leading to $\Q(\tau(y)>t)=P^y(t,\Xi^\star)$, dashed blue), $\lambda\leftarrow \lambda^{\varphi}$ (S-CIR) and $\lambda\leftarrow \lambda^{\theta}$ (TC-CIR model). By construction of $\varphi$ and $\Theta$, the last two curves coincide (magenta) and agree with $G(t)$.}\label{fig:Ford2}
\end{figure}
\else
\begin{center}
[Figure 3 about here]
\end{center}
\fi 

The parameters used in the numerical examples in the rest of the paper are given in Table \ref{tab:parameters}. 

\begin{table}[H]
    \centering
    \begin{tabular}{|l|cccc|}
    \hline
      $\Xi$& $\kappa$ & $\beta$ & $\delta$ & $y_0$ \\
     \hline
     $\Xi^\star$ & 0.0555 & 0.3018 & 0.2939 & $h_0$\\
     \hline
     $\Xi^{\star,+}$ & 0.2118 & 0.0030 & 0.0006 & $h_0$\\
     \hline
     $\Xi_0^\star$ & 0.0624 & 0.2975 & 0.3343 & 0.0000 \\
     \hline
     $\Xi_0^{\star,+}$ & $3.8252.10^{-01}$ & $9.6881.10^{-03}$ & $1.5195.10^{-01}$ & $3.2093.10^{-10}$ \\
     \hline
    \end{tabular}
    \caption{Calibration parameters using Ford piecewise constant hazard rate. Parameters $\Xi^\star$ and $\Xi^{\star,+}$ correspond to the parameters of the CIR model $y$ with and without positivity constraint, where $y_0$ is set exogenously to the first level of the piecewise hazard rate function, $h_0=0.0030$. The other parameters, $\Xi^\star_0$ and $\Xi_0^{\star,+}$, correspond to the similar cases but where $y_0$ is a parameter that enters the optimization procedure. In all cases, we have taken $\alpha=\omega=0$.}
    \label{tab:parameters}
\end{table}

Notice that in both Figure \ref{fig:Ford1} and \ref{fig:Ford2}, the shift function $\varphi$ can take negative values. This means that the shift approach, S-CIR, yields negative default intensities $\lambda^{\varphi}$ and, calibrated that way, is flawed. In particular, we cannot interpret $\lambda^\varphi$ as a default intensity associated to a Cox process.  This contrasts with the TC-CIR approach since $\lambda^{\theta}$ is a positive process if so is  $y$. To fix this issue in a CIR++ framework, one needs to rely on PS-CIR.  We note the corresponding processes $y^+$ and $\lambda^{\varphi,+}$. As illustrated on Figure \ref{fig:FordConstr} with our Ford example, this procedure is very restrictive: it leads to a curve $P^y$ that is decreasing at a very low rate. In particular, the shape of $P^{\lambda^{\varphi,+}}$ essentially results from the shift, not from the base model $y$. This is problematic: it basically amounts to say that $h\approx \varphi$, i.e., that the PS-CIR process $\lambda^{\varphi,+}$ is essentially deterministic. This will put strong limitations on the resulting default model, and will be further discussed in the remaining subsections.

\ifdefined \InclFig
\begin{figure}[H]
\centering
\subfigure[Survival probability ($P^{y^+},G=P^{\lambda^{\varphi,+}}$)]{\includegraphics[width=0.45\columnwidth]{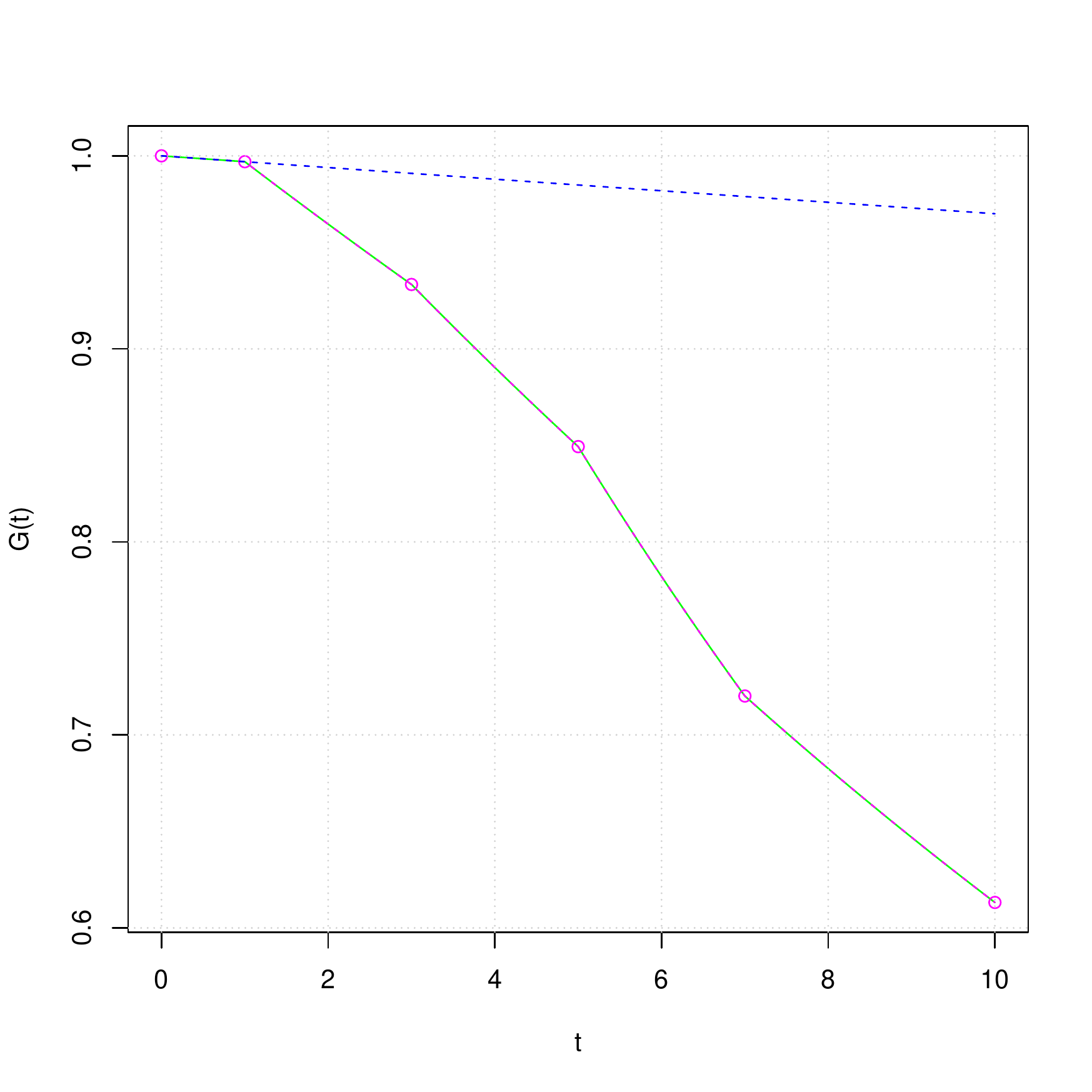}}\hspace{0.2cm}
\subfigure[Survival probability ($P^{y^+},G=P^{\lambda^{\varphi,+}}$)]{\includegraphics[width=0.45\columnwidth]{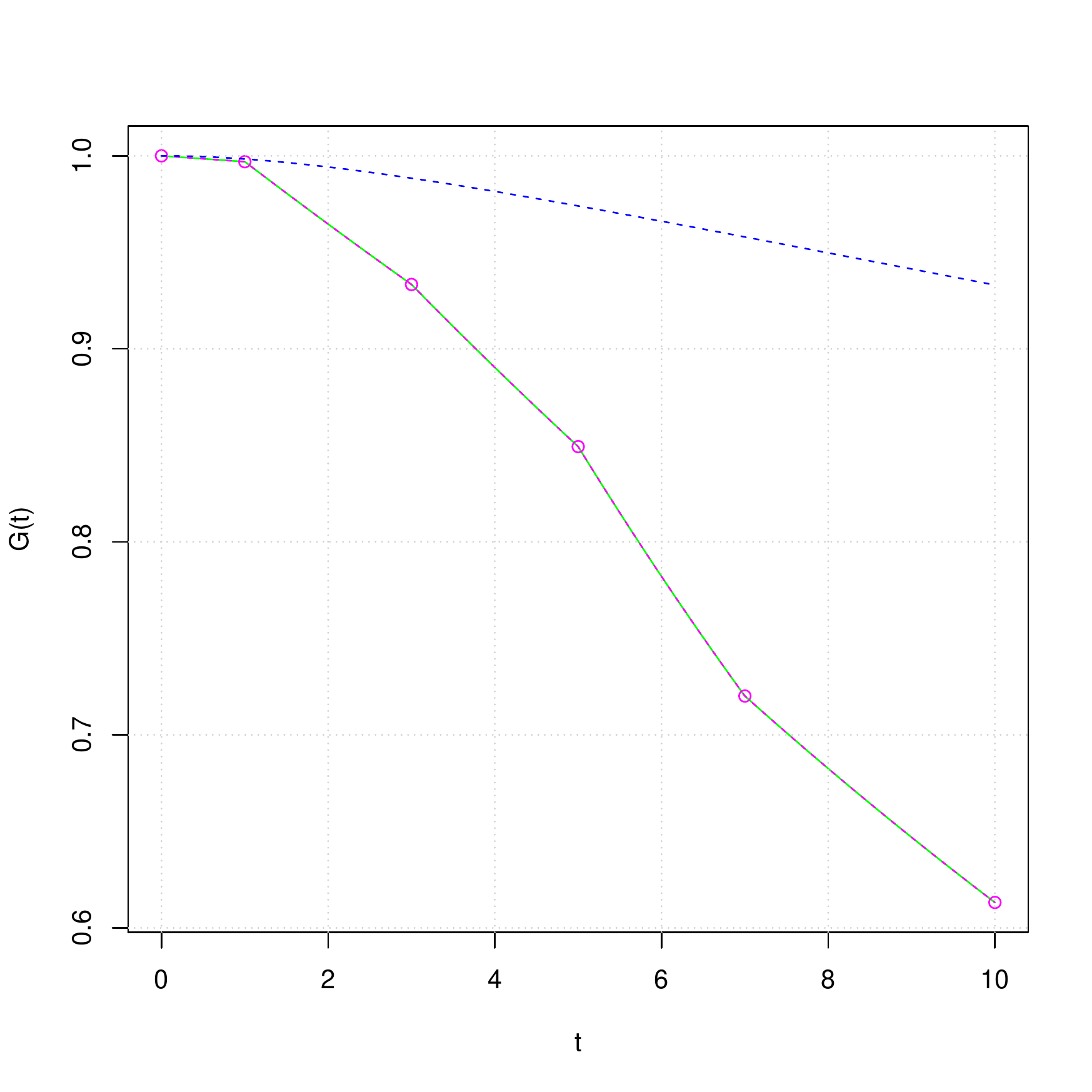}}\hspace{0.2cm}
\subfigure[Shift function ($\varphi=\varphi^{\star,+}$)]{\includegraphics[width=0.45\columnwidth]{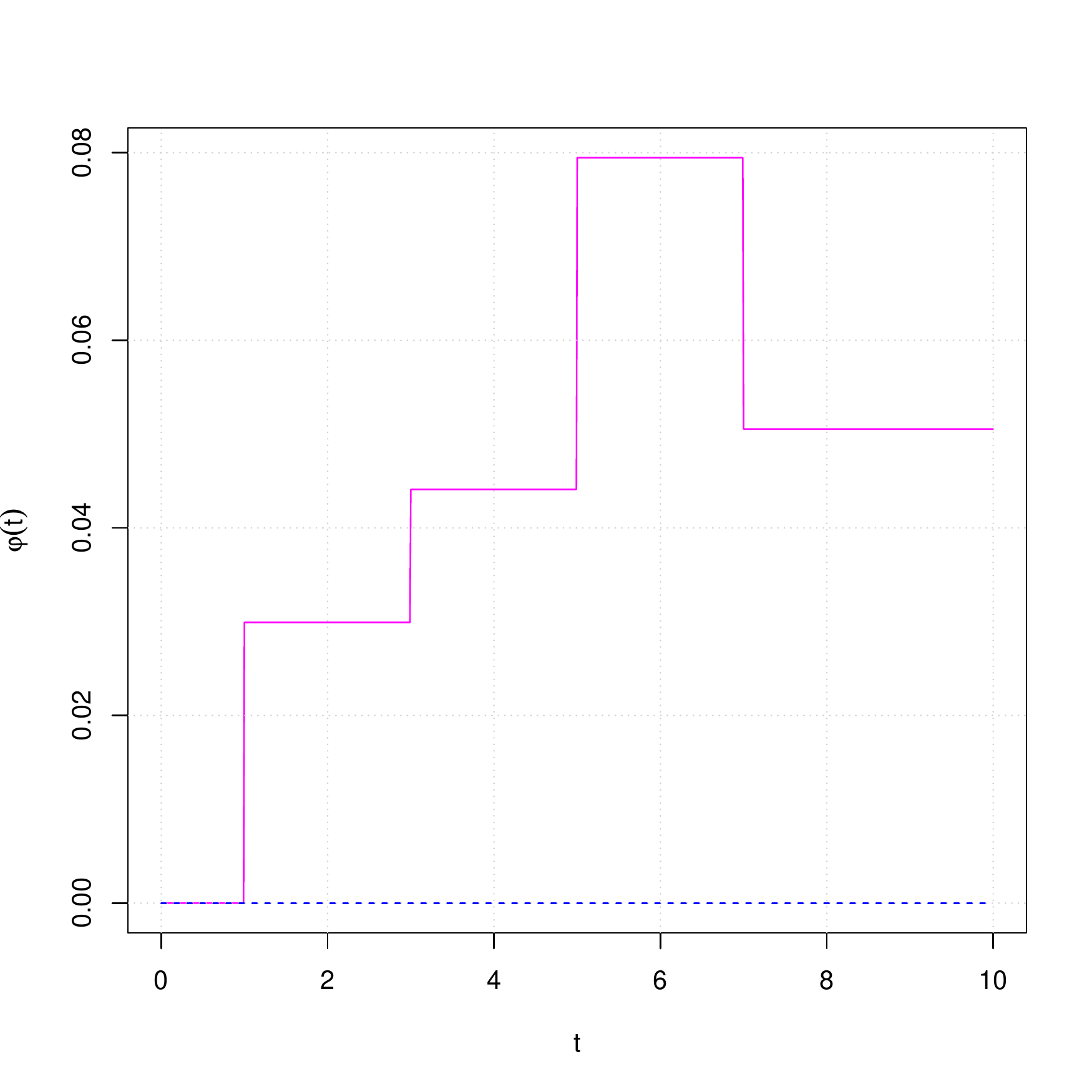}}\hspace{0.2cm}
\subfigure[Shift function ($\varphi=\varphi^{\star,+}$)]{\includegraphics[width=0.45\columnwidth]{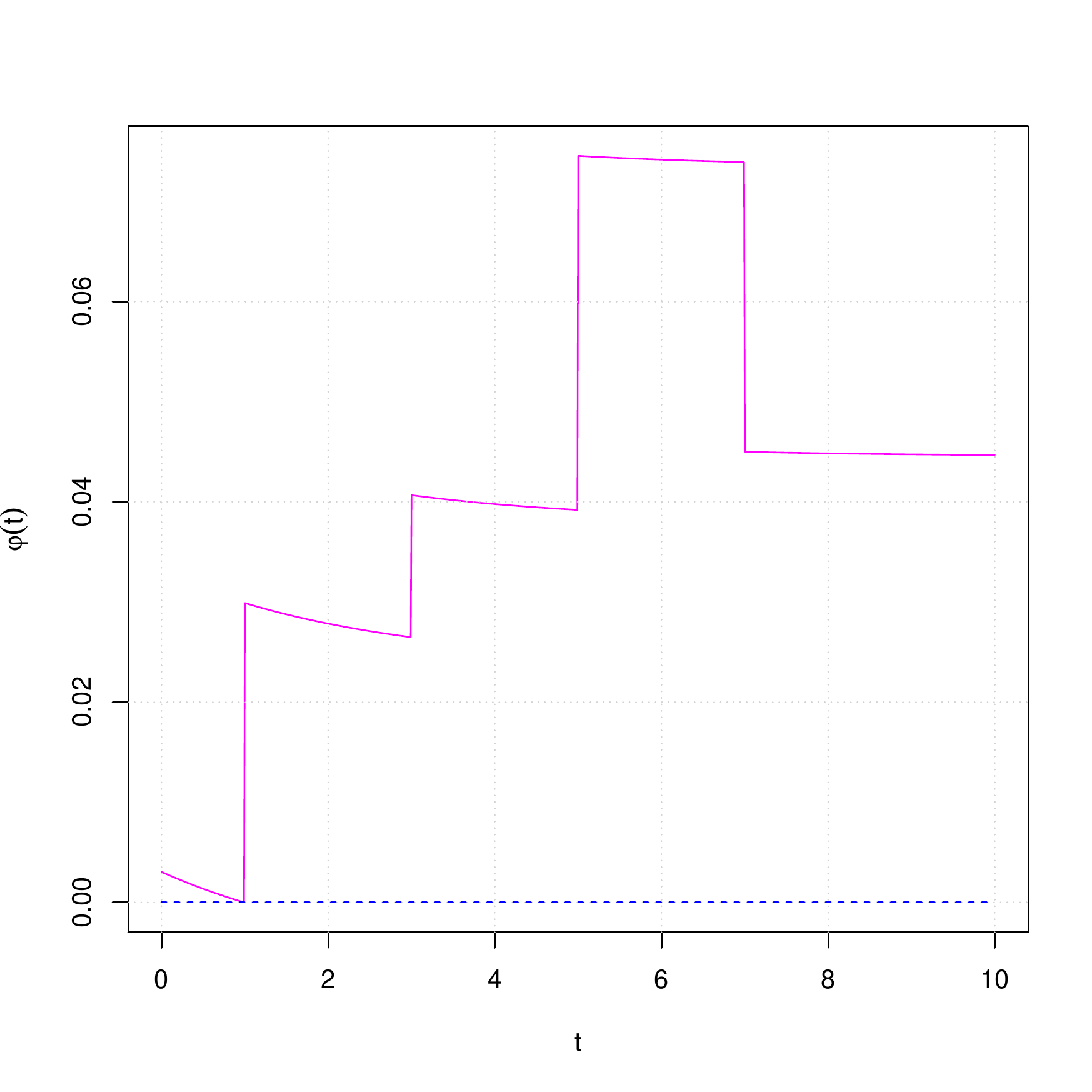}}\hspace{0.2cm}
\caption{Fitting Ford Inc. CDS term-structure using $\lambda^{\varphi^\star,+}$ (PS-CIR). Panels (a) and (c) correspond to $y^+_0=h_0$ whereas panels (b) and (d) correspond to the case where $y^+_0$ is one of the optimized parameters. The survival probability curve $G(t)$ is parametrized with a piecewise constant hazard rate function $h(t)$ extracted from market prices taken from Bloomberg on November 12 2018. The parameters $\Xi^{\star,+}$ are computed under the constraint $\varphi(t)\leftarrow\varphi^\star(t;\Xi^{\star,+})\geq 0$ . The base model $y^+$ is a CIR with parameters $\Xi=\Xi^{\star,+}$ (left) and $\Xi=\Xi_0^{\star,+}$ (right).}\label{fig:FordConstr} 
\end{figure}
\else
\begin{center}
[Figure 4 about here]
\end{center}
\fi 

\subsection{Variance analysis}\label{seq:Variance}

Interesting observations can be made regarding the variance of the various integrated processes. As shown in the next two sections, they will have important consequences when considering financial applications, where $\Lambda$ plays a central role in governing volatility and covariance effects.

First, observe that the integrated CIR process with optimal parameter $\Xi^\star$ is expected to feature a larger variance compared to the integrated CIR with parameter $\Xi^{\star,+}$. 
Because of the shift constraint, the discount curve $P^y$ in the latter case rends to be much flatter than in the former case i.e., one expects to have, in general
$$P^y(t;\Xi^{\star,+})\geq P^y(t;\Xi^\star)\;.$$
This can be observed from panels (a) and (b) of Figure~\ref{fig:FordConstr}. When working with $\Xi^{\star,+}$, a substantial part of the shape of $P^{market}=G$ comes from the deterministic shift. This amounts to limit the randomness of the process. Not surprisingly, this will impact the variance of the integrated process $Y$. Indeed, because the discount curve of the CIR process with parameter $\Xi^{\star,+}$ generally dominates that of the CIR process with parameter $\Xi^\star$, one intuitively expects the variance of the CIR with parameter $\Xi^\star$ to be larger than that of the CIR with parameter $\Xi^{\star,+}$, due to the zero lower bound. In other words, even if it seems difficult to provide a formal proof, one expects intuitively the following to hold, in general:
$$v^{\Lambda^{\varphi}}(t)=v^{Y}(t;\Xi^\star) \geq v^{Y}(t;\Xi^{\star,+})=v^{\Lambda^{\varphi,+}}(t)\;.$$
This is indeed the case on Figure~\ref{fig:FordConstr2}: $v^{\Lambda^{\varphi}}$ (dotted blue) dominates $v^{\Lambda^{\varphi,+}}$ (solid blue).\medskip 

Second, observe that for a given base process $y$, the variance of the integrated TC-CIR is always larger than that of the integrated PS-CIR. Indeed, when working under the positivity constraint (i.e., when $y$ is driven by $\Xi^{\star,+}$), we necessarily have $\Theta^{+}(t):=\Theta(t;\Xi^{\star,+})\geq t$, in agreement with Theorem \ref{th:th2}. Because for any parameter, the variance of $Y$ is an increasing function of time (Lemma \ref{Lem:VarCIR} in the Appendix, Section \ref{sec:AppAffineCIR}) we  have, for $\Xi\leftarrow \Xi^{\star,+}$ in particular, $$v^{\Lambda^{\theta,+}}(t)=v^{Y}(\Theta^{+}(t);\Xi^{\star,+})\geq v^{Y}(t;\Xi^{\star,+})=v^{\Lambda^{\varphi,+}}(t)\;.$$

Third, we observe from  Figure~\ref{fig:FordConstr2} that, in this example at least, the variance of the TC-CIR using $\Xi^\star$ is comparable to the variance of the S-CIR:
$$v^{\Lambda^{\theta}}(t)\approx v^{Y}(t;\Xi^\star)=v^{\Lambda^{\varphi}}(t)\;.$$
The fact that the variance of the S-CIR is expected to be close to that of the corresponding TC-CIR model can be understood intuitively as follows. As explained above the parameter $\Xi\leftarrow \Xi^\star$ computed using \eqref{eq:CalyNoConstraint} leads the HAJD $y$ to best fits the market curve, and the clock is used to absorb the remaining discrepancies. Therefore, one expects the clock not to deviate much from the actual time, i.e $\theta(t)\approx 1$ and the two processes to behave similarly. In particular, the parameters of $y^\theta$ are those of $y$ scaled by $\theta(t)$, and $x^\theta_t=\theta(t) y^\theta_t\approx y^\theta_t$, at least when the fit between $P^{market}$ and the base HAJD model $P^y$ is not too poor; see \eqref{eq:TCHAJD}.\medskip

To sum up, we observe that when dealing with CIR++ under a positivity constraint, one has to choose between a valid (but low-volatility) PS-CIR process $\lambda^{\varphi,+}$, or a flawed (by high-volatility) S-CIR one $\lambda^{\varphi}$. By contrast, the TC-CIR model $\lambda^{\theta}$ is always valid (Corollary~\ref{cor:cor2}), always feature a variance that is larger than the PS-CIR counterpart (Theorem~\ref{th:th2}), and its variance is actually comparable to the large levels generated by the S-CIR. The TC-CIR thus proves to be a solid challenger to CIR++ models. In particular, its features are specifically interesting when dealing with actual credit risk applications, as we now point out based on two case studies.  

\ifdefined \InclFig
\begin{figure}[H]
\centering
\subfigure[$y_0=h_0$]{\includegraphics[width=0.45\columnwidth]{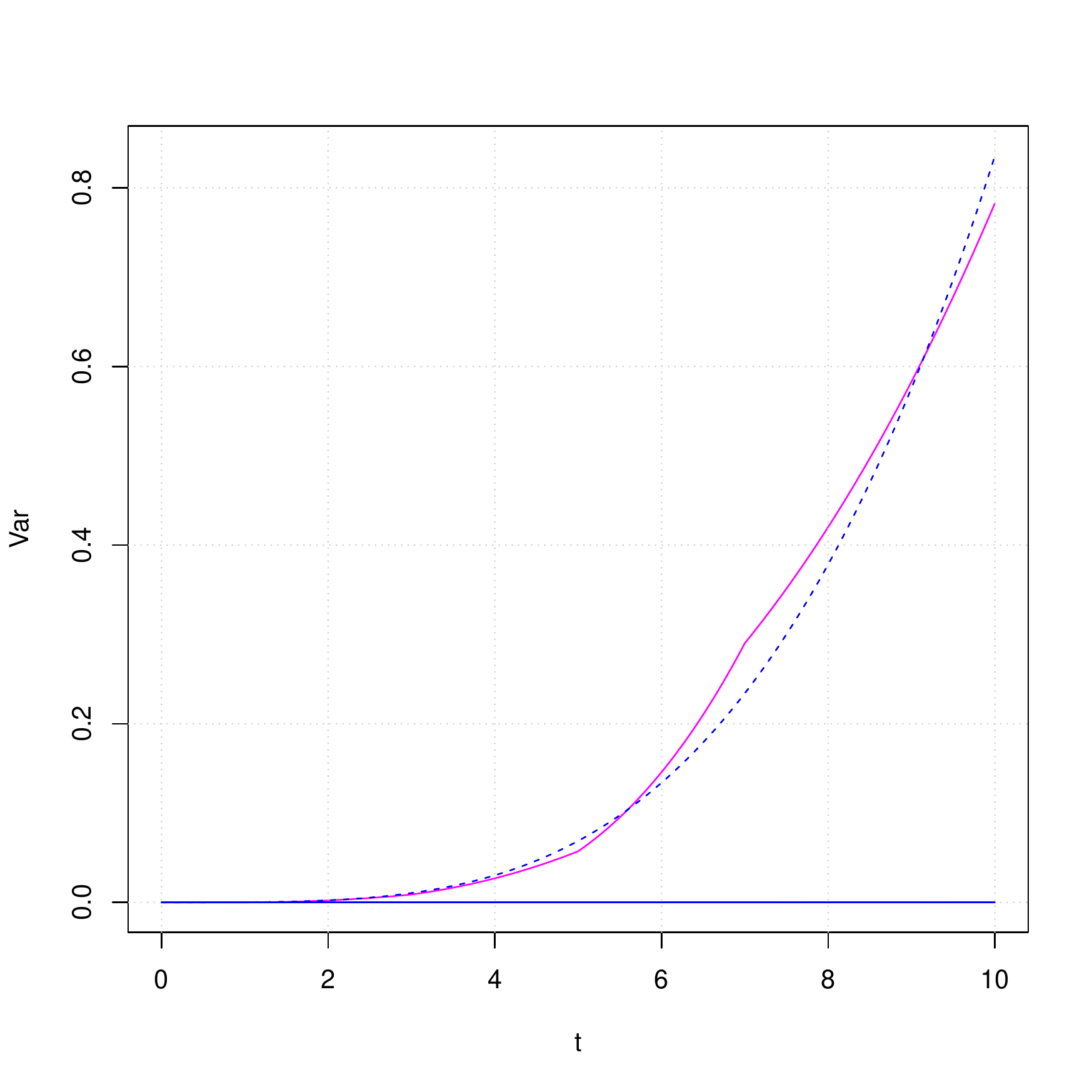}}\hspace{0.2cm}
\subfigure[$y_0$ optimized]{\includegraphics[width=0.45\columnwidth]{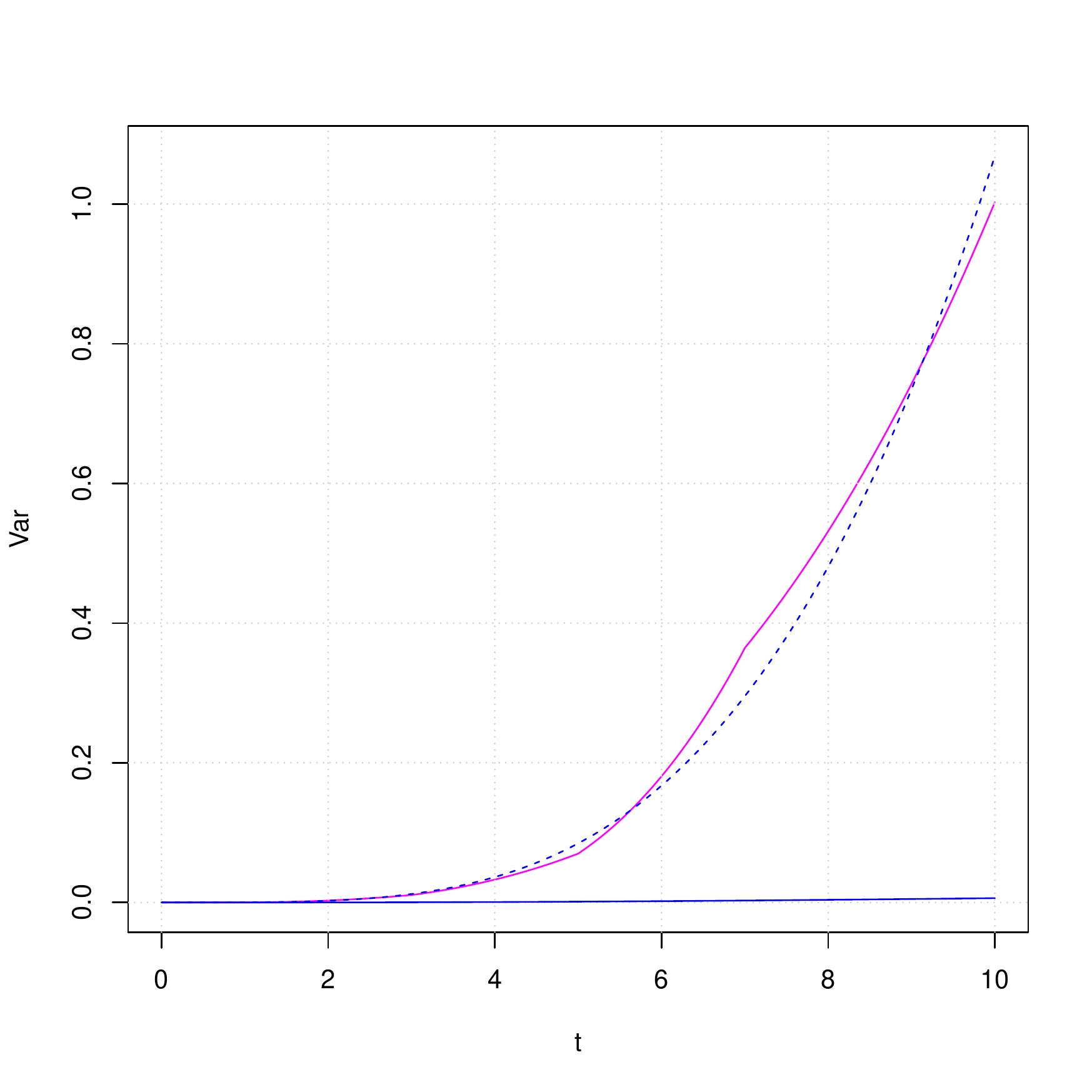}}
\caption{Variances of the integrated versions of  $\lambda^{\varphi,+}$ (PS-CIR $\Xi=\Xi^{\star,+}$ (left) and $\Xi=\Xi_0^{\star,+}$ (right), solid blue), $\lambda^{\varphi}$ (S-CIR $\Xi=\Xi^\star$ (left) and $\Xi=\Xi_0^\star$ (right), dashed blue), and $\lambda^{\theta}$ (TC-CIR with $\Xi=\Xi^\star$ (left) and $\Xi=\Xi_0^\star$ (right), magenta).}\label{fig:FordConstr2}
\end{figure}
\else
\begin{center}
[Figure 5 about here]
\end{center}
\fi 

\subsection{Pricing CDS options}

We deal with the pricing of a CDS option (CDSO). Because CDSO is an option on CDS, we start by recalling the no-arbitrage pricing equation of a CDS. We note $t$ the valuation time and assume $\tau>t$ as it is pointless to price a CDS (or a CDSO) post-default. From the perspective of the protection buyer, the time-$t$ value of a 1 dollar notional CDS $CDS_t(a,b,k)$ starting at time $T_a$ with maturity $T_b$, $t\leq T_a<T_b$, a spread $k$ and (known) loss given default $LGD=(1-R)$ is given by the difference of the conditional risk-neutral expectation of the protection and premium cashflows :
\beqn
CDS_t(a,b,k)&=&\E\left[(1-R)\indic_{\{T_a\leq \tau\leq T_b\}}P_t^r(\tau)|\cG_t\right]\nonumber\\
&&-k\E\left[\left.\sum_{i=a+1}^b \left(\indic_{\{\tau\geq T_i\}}\alpha_i P^r_t(T_i)+ \indic_{\{T_{i-1}\leq \tau< T_i\}}\alpha_i\frac{\tau-T_{i-1}}{T_i-T_{i-1}} P^r_t(\tau)\right)\right|\cG_t\right]\nonumber
\eeqn
with $\alpha_i$ the day count fraction between dates $T_{i-1}$ and $T_i$ which, in a standard CDS, is around $0.25$ (quarterly payment dates). In a reduced-form setup, when the default is triggered by the first jump of a Cox process with intensity $\lambda$, this expression can be developped explicitely thanks to the Key lemma:
\beq
CDS_t(a,b,k)=\indic_{\{\tau >t\}}\left(-(1-R)\int_{T_a}^{T_b}P^r_t(u)\partial_uP_t^\lambda(u)du -k~C_t(a,b)\right)\;,\label{eq:CDS}
\eeq
where $C_t(a,b)$ is the risky duration, i.e., the time-$t$ value of the CDS premia paid during the life of the contract when the spread is 1:
\beq
    C_t(a,b):=\sum_{i=a+1}^b\alpha_iP_t^r(T_i)P_t^\lambda(T_i)-\int_{T_{i-1}}^{T_i}\frac{u-T_{i-1}}{T_{i}-T_{i-1}}\alpha_iP_t^r(u)\partial_uP_t^\lambda(u)du\;.\nonumber
\eeq

The spread which, at time $t$, sets the forward start CDS at 0, called par spread,  is given by: 
\begin{equation}
    \indic_{\{\tau >t\}}s_t(a,b):=\indic_{\{\tau >t\}}\frac{-(1-R)\int_{T_a}^{T_b}P^r_t(u)\partial_uP_t^\lambda(u)du}{C_t(a,b)}\;.\label{eq:ParCDS}
\end{equation}

The no-arbitrage price of a call option on such a contrat at time $t=0$ becomes
\beqn
PSO(a,b,k)&=&\mathbb{E}\left[(CDS_{T_a}(a,b,k))^+P^r(T_a)\right]\nonumber\\&=& P^r(T_a)\mathbb{E}\left[e^{-\Lambda_{T_a}}\left( (1-R)-\sum_{i=a+1}^b\int_{T_{i-1}}^{T_i}g_i(u)P^{r+\lambda}_{T_a}(u)du\right)^+ \right]\;,\nonumber
\eeqn
where $g_i(u):= (1-R)(r(u)+\delta_{T_b}(u))+k\frac{\alpha_i r(u)}{T_i-T_{i-1}}(1-(u-T_{i-1}))$, with $\delta_{s}(u)$ the Dirac delta function centered at $s$.\\

Replacing the base intensity model $(\lambda)$ by its shifted $(\lambda^\varphi,\lambda^{\varphi,+})$ or time-changed $(\lambda^\theta)$ versions leads to model prices noted $PSO^\varphi(a,b,k),PSO^{\varphi,+}(a,b,k)$ and $PSO^\theta(a,b,k)$, respectively. Interestingly, these models are equally tractable as they feature similar expressions that can be written in terms of the base process $\lambda$ or its time integral, $\Lambda$. For instance, dropping the $\Xi$ for short,
\beq
\Lambda^\varphi_{T_a} = \Lambda_{T_a}+\int_0^{T_a} \varphi(u)du\;,~~
\Lambda^\theta_{T_a} = \Lambda_{\Theta(T_a)}\;,\nonumber
\eeq
and
\beqn
P^{\lambda^\varphi}_{T_a}(u)&=&P^{\lambda}_{T_a}(u)e^{\int_{T_a}^u\varphi(s) ds}=e^{A_{T_a}(u)-B_{T_A}(u)\lambda_{T_a}+\int_{T_a}^u\varphi(s) ds}\;,\nonumber\\
P^{\lambda^\theta}_{T_a}(u)&=&P^{\lambda}_{\Theta(T_a)}(u)=e^{A_{\Theta(T_a)}(u)-B_{\Theta(T_a)}(u)\lambda_{\Theta(T_a)}}\;.\nonumber
\eeqn
Recall that these expressions have a closed form when $\lambda$ is a (J)CIR process.\medskip

Such kind of options has little liquidity. Models are then often compared in terms of their capabilities to generate large ``implied volatilities''. Indeed, empirical evidences show that this is a typical feature of CDS option quotes, when disclosed. Therefore, we compare the models in terms of their ``Black volatilities'': the volatility that one needs to plug in a  ``Black-Scholes'' type of model to reproduce the model prices. Black model for $PSO$ is recalled in the Appendix, Section \ref{sec:BlackPSO}. The Black volatility associated to a model price $PSO^{model}(a,b,k)$ is thus the volatility $\bar{\sigma}$ satisfying $PSO^{model}(a,b,k)=PSO^{Black}(a,b,k,\bar{\sigma})$. Recall that in all cases, the intensity process $\lambda$ is calibrated to the market, i.e., $P^\lambda(t)=G(t)$. In other words, choosing, e.g., a CIR process for the base intensity process $\lambda$ combined with the correct shift with ($\varphi^+$) or without ($\varphi$) positivity constraint, or eventually using the correct clock rate $\theta$, all three models yield the same survival probability curve ($P^{\lambda^{\varphi}}(t)=P^{\lambda^{\varphi,+}}(t)=P^{\lambda^{\theta}}(t)=G(t)$). Hence, all these models agree on the par spread:
\beq
s_0(a,b)=\frac{(1-R)\int_{T_a}^{T_b}P^r(u)h(u)G(u)du}{C_0(a,b)}\;.\nonumber
\eeq

We compare the S-CIR, the PS-CIR and the TC-CIR. The base HAJD process $y$ in TC-CIR is taken to be the same as that of the S-CIR. One can see from Table \ref{tab:CDSO-CIR} that the S-CIR features large implied volatilities. Recall however that it allows for negative intensities, hence is not appropriate. The PS-CIR model are not capable of generating large volatility levels, in line with the previous discussion. The TC-CIR fits in between: it rules out negative intensities, while maintaining substantial volatility levels.\medskip 

One might be concerned by the fact that the implied volatilities of the TC-CIR remain  relatively small. This can be addressed in two ways. First, one can play with the parameter $\Xi$. However, the Feller constraint is often required to hold, which sets limits on the process' volatility. Another approach consists of considering a JCIR model as HAJD. Indeed, JCIR is often considered when large volatilities are required. However, as explained in Section \ref{sec:positivity}, increasing the volatility by boosting the jump activity while maintaining the calibration to a given market curve $G$ reinforces the positivity issue. Fortunately, we do not have this problem in the TC-JCIR. One can drastically increase the jump activity without impacting the positivity of the TC-JCIR. As a consequence, the TC-JCIR seems very much appropriate when one needs a positive but yet high-volatility process. This is illustrated on Table \ref{tab:CDSO-JCIR} using the same jump parameters as those given in \cite{Brigo06}. We keep the same parameter $\Xi$ as before for the diffusion part, and play with the jump rate ($\omega$) and jump size ($\alpha$) in the compound Poisson process $J$. In every case, the clock $\Theta$ is chosen such that the model perfectly fits Ford's survival probability curve. Interestingly, the (positive) TC-JCIR model can feature much larger implied volatility levels than the PS-CIR. The results for the S-JCIR are also larger than the PS-CIR, but they are not shown because the negative intensity problem is magnified.

\begin{table}[H]
\centering
\resizebox{15cm}{!}{
\begin{tabular}{ll}
\begin{tabular}{|c|c|c|c|c |}
   \hline
    $T_a$&$T_b$& \multicolumn{2}{c|}{CIR++} & TC-CIR  \\
   \hline
   & & $\lambda^{\varphi}$ & $\lambda^{\varphi,+}$ & $\lambda^{\theta}$  \\
   \hline
   1&3 & 67.12\% & 1.03\% & 43.68\%  \\
   \hline
   1&5 & 45.10\% & 1.25\% & 26.92\% \\
   \hline
  1& 7 & 27.72\% & 1.64\% & 16.86\%  \\
   \hline
  1 &10 & 21.52\% & 1.65\% & 12.86\%  \\
   \hline
   3 &5 & 61.30\% & 0.85\% & 57.16\%  \\
   \hline
   3&7 & 34.88\% & 1.15\% & 36.33\%   \\
   \hline
   3&10 & 27.65\% & 1.08\% & 27.17\%  \\
   \hline
  5& 7 & 34.81\% & 1.16\% & 42.60\%  \\
   \hline
   5&10 & 30.96\% & 0.93\% & 31.19\%  \\
   \hline
   7&10& 45.37\% & 0.63\% & 38.93\% \\
   \hline
\end{tabular}
\hspace{0.25cm}
&
\begin{tabular}{|c|c|c|c|c |}
   \hline
    $T_a$&$T_b$& \multicolumn{2}{c|}{CIR++} & TC-CIR  \\
    \hline
   & & $\lambda^{\varphi}$ & $\lambda^{\varphi,+}$ & $\lambda^{\theta}$  \\
   \hline
   1&3 & 65.21\% & 9.48\% & 44.00\%  \\
   \hline
   1&5 & 42.35\% & 5.88\% & 26.56\% \\
   \hline
  1& 7 & 25.30\% & 4.87\% & 16.30\%  \\
   \hline
  1 &10 & 19.37\% & 2.94\% & 12.27\%  \\
   \hline
   3 &5 & 63.87\% & 8.02\% & 58.67\%  \\
   \hline
   3&7 & 35.20\% & 4.50\% & 36.10\%   \\
   \hline
   3&10 & 27.02\% & 3.44\% & 26.62\%  \\
   \hline
  5& 7 & 35.77\% & 4.59\% & 43.32\%  \\
   \hline
   5&10 & 30.35\% & 3.70\% & 30.61\%  \\
   \hline
   7&10& 45.05\% & 5.16\% & 39.40\% \\
   \hline
\end{tabular}
\end{tabular}
}
\caption{Black volatilities for at-the-money ($k=s_0(a,b)$) CDS options implied by CIR++ models (S-CIR and PS-CIR) and the TC-CIR model with $y_0=h_0$ (left) and $y_0$ optimized (right) using Monte Carlo simulation (500K paths with time step 0.01). In all the considered cases, the CIR++ model without shift is not valid since $\inf\{ \lambda_t^{\varphi},\;t\in[0,T_a]\}< 0$. Among the two valid intensity models ($\lambda_t^{\varphi,+}$ and $\lambda_t^{\theta}$), the latter exhibits a much higher implied volatility.} \label{tab:CDSO-CIR}
\end{table}

\begin{table}[H]
\centering
\resizebox{7.5cm}{!}{
\begin{tabular}{|c|c|c|c|c|}
   \hline
    $T_a$&$T_b$& \multicolumn{3}{c|}{TC-JCIR $(\omega, \alpha$)}   \\
   \hline
   & & $(0,0)$ & $(0.1,0.1)$ & $(0.15,0.15)$    \\
   \hline
   1&3 & 43.68\% & 79.04\% & 100.17\%  \\
   \hline
   1&5 & 26.92\% & 48.07\% & 69.69\% \\
   \hline
  1& 7 & 16.86\% & 30.43\% & 44.00\%  \\
   \hline
  1 &10 & 12.86\% & 23.33\% & 33.75\%  \\
   \hline
   3 &5 & 57.16\% & 65.50\% & 82.60\%  \\
   \hline
   3&7 & 36.33\% & 42.85\% & 53.17\%   \\
   \hline
   3&10 & 27.17\% & 33.17\% & 41.36\%  \\
   \hline
  5& 7 & 42.60\% & 49.13\% & 60.11\%  \\
   \hline
   5&10 & 31.19\% & 37.34\% & 45.78\%  \\
   \hline
   7&10& 38.93\% & 40.59\% & 46.02\% \\
   \hline
\end{tabular}
}
\caption{Black volatilities for at-the-money ($k=s_0(a,b)$) CDS options implied by the TC-JCIR model (jump arrival rate $\omega$ and jump size $\alpha$) using Monte Carlo simulation ($10^6$ paths with time step 0.01) and paramter set $\Xi=\Xi^\star$ but for various jump parameters $(\alpha,\omega)$.} 
\label{tab:CDSO-JCIR}
\end{table}

\subsection{Wrong-way risk impact in credit valuation adjustments}
A major concern of the post-crisis regulation is the modeling of the capital requirement of firms tacking into account some credit adjustment to the valuation under credit risk. 
Counterparty credit risk is defined as the risk that the counterparty of an over-the-counter (OTC) deal will default before the maturity of the contract. The latter can be seen as an option given to the counterparty, and can be priced in a risk-neutral setup by adjusting the OTC derivative, leading to CVA. The latter is nothing but the expected losses due to the missed payments associated to the OTC portfolio. In a risk-neutral specification and assuming $\tau>0$, the current ($t=0$) value of the CVA is expressed as:
\beq
\mathrm{CVA} = \mathbb{E}\left[(1-R)V^+_{\tau}\indic_{\{\tau\leq T\}}\right]\nonumber
\eeq
where $V$ stands for the discounted exposure (i.e., the exposure process rescaled by the stochastic discount factor $D$). A straightforward application of the Key lemma (under some technical conditions that are valid here) yields
\beq\label{eq:CVA}
\mathrm{CVA} = \mathbb{E}\left[(1-R)\int_0^TV^+_{u}\lambda_u e^{-\Lambda_u}du\right]\;.
\eeq

The CVA of the shifted and the time-changed models, $\mathrm{CVA}^\varphi$ and $\mathrm{CVA}^\theta$, correspond to above expression, replacing $(\lambda, \Lambda)$ by $(\lambda^\varphi,\Lambda^\varphi)$ and $(\lambda^\theta,\Lambda^\theta)$, respectively. The purpose of this section is to illustrate the order of magnitude of CVA figures that can be obtained with either models. In particular, we do not aim at representing a specific exposure. Instead, we simplify the analysis by considering two prototypical dynamics:
\beqn
dV_t&=&\nu dW^V_t\;,\nonumber\\  
dV_t&=&\left(\gamma(T-t)-\frac{V_t}{T-t}\right)dt +\nu dW^V_t\;.\nonumber
\eeqn
where $W^V$ is an $\mathbb{F}$-Brownian motion. The first SDE is that of a martingale, and can depict the evolution of the discounted price of a forward contract prior to its cashflow date. The second SDE corresponds to a Brownian bridge with drift, and mimics the dynamics of the discounted price of an asset paying continuous dividends. These two models have been previously used in \cite{Vrins16a} and \cite{Brigo17} to describe, in a schematic way, exposures of FRA and IRS. Calibration to actual exposures give indicative value for the parameters.\medskip

In general, there is no reason to assume that the Brownian motion driving the default intensity ($W$) would be independent of the Brownian motion driving the exposure ($W^V$): it depends on the problem at hand. Usually, we consider the general case of \textit{wrong-way risk} (WWR) effect, obtained by introducing a correlation between the Brownian drivers. For the CIR++ we assume $dW_tdW^V_t=\rho dt$, whereas for the TC-CIR, we apply the synchronisation procedure devised in \cite{Mbaye2018} in order to preserve the correlation after time-changing the intensity process. In the special case where the default time of the counterparty is independent from the discounted exposure (i.e., $\rho=0$, that is no wrong-way risk) one can easily  deduce from \eqref{eq:CVA} the \textit{independent} CVA formula
\beq
\mathrm{CVA}^\perp = -(1-R)\int_0^T\mathbb{E}\left[V^+_{u}\right]d\,\mathbb{E}\left[e^{-\Lambda_u}\right]=(1-R)\int_0^T f^\lambda(u) \E\left[V^+_u\right]P^\lambda(u)du\;.\nonumber
\eeq

Recall that whatever the chosen model, it is assumed to be calibrated to the survival probability curve $G$, extracted from CDS prices. This leads to  $P^\lambda(u)=G(t)$, and to the optimal shift and clock functions, namely $\varphi$ or $\varphi^{+}$ in the S-CIR and PS-CIR cases, and $\Theta$ for the TC-CIR. In this case, $\mathrm{CVA}^\perp$ does not depend on the default model: 
\beq\label{eq:indCVA}
{\rm CVA}^{\perp}=-(1-R)\int_0^T\mathbb{E}\left[V^+_{u}\right]dG(u)=(1-R)\int_0^Th(u)\mathbb{E}\left[V^+_{u}\right]G(u)du\;.\nonumber
\eeq

However, the independent case $\rho=0$ is unrealistic, and may lead to severe over or underestimations of CVA~\cite{Kim16,Brigo17,Breton18}. Under WWR, CVA becomes model-dependent. Figure \ref{fig:CVAplots} shows the evolution of CVA with respect to $\rho$ for three different models: $\lambda^{\varphi}$ (CIR++ without constraint, solid blue), $\lambda^{\varphi,+}$ (CIR++ with constraint, dashed blue) and $\lambda^{\theta}$ (TC-CIR, dashed magenta), all calibrated to Ford's survival probability curve $G$ as before. Under no-WWR, the CVA is equal to the independent CVA (cyan): it is flat, model-free and can be computed using a simple integration. Under WWR, the CVAs are computed using Monte Carlo simulations (100K paths, time step of 0.01) and adaptive control variate \footnote{see \cite{Mbaye2018} for the implementation of the adaptive control variate applied on CVA computation.}. The TC-CIR and S-CIR models exhibit the largest WWR effects and seem therefore appropriate to deal with high WWR applications. Recall that only TC-CIR is valid here as S-CIR gives room to negative intensities. The PS-CIR however is almost flat, equal to the independent CVA.  This can be understood from the fact that WWR is essentially a covariance effect between $V$ and $e^{-\Lambda}$. Hence, the models featuring large variance for $\Lambda$ exhibit larger WWR effects at any (non-zero) fixed correlation level $\rho$. Eventually, TC-CIR provides an appealing trade-off: on the one hand, as the PS-CIR, it rules out the negative intensity problem inherent to the S-CIR model. But on the other hand, it preserves, to some extend, the variance of the S-CIR model, and therefore exhibits a much larger variance compared to PS-CIR. 

\ifdefined \InclFig
\begin{figure}[H]
\centering
\subfigure[$dV_t=\nu dW^V_t$, $y_0=h_0$]{\includegraphics[width=0.48\columnwidth]{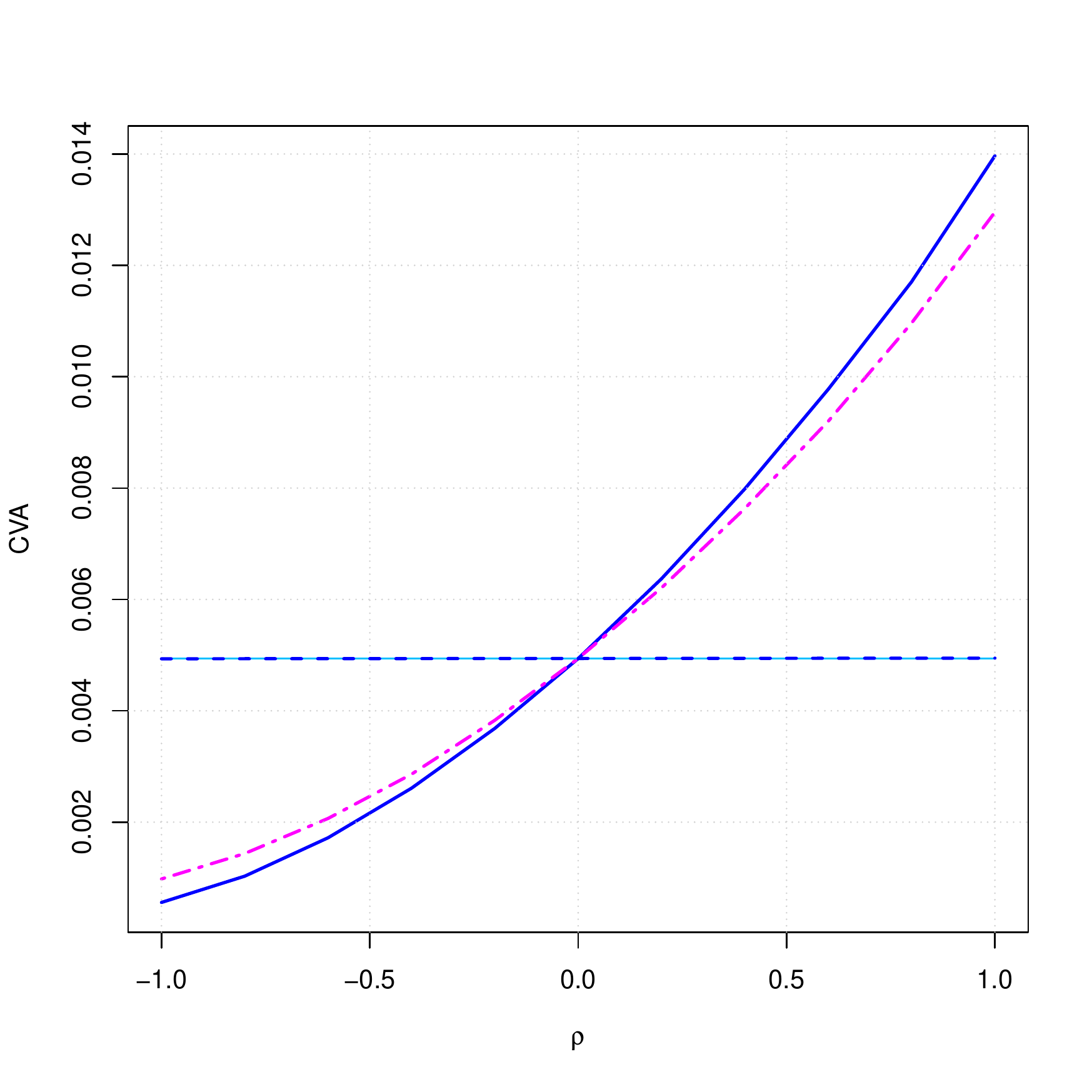}}\hspace{0.20cm}
\subfigure[$dV_t=\nu dW^V_t$, $y_0$ optimized]{\includegraphics[width=0.48\columnwidth]{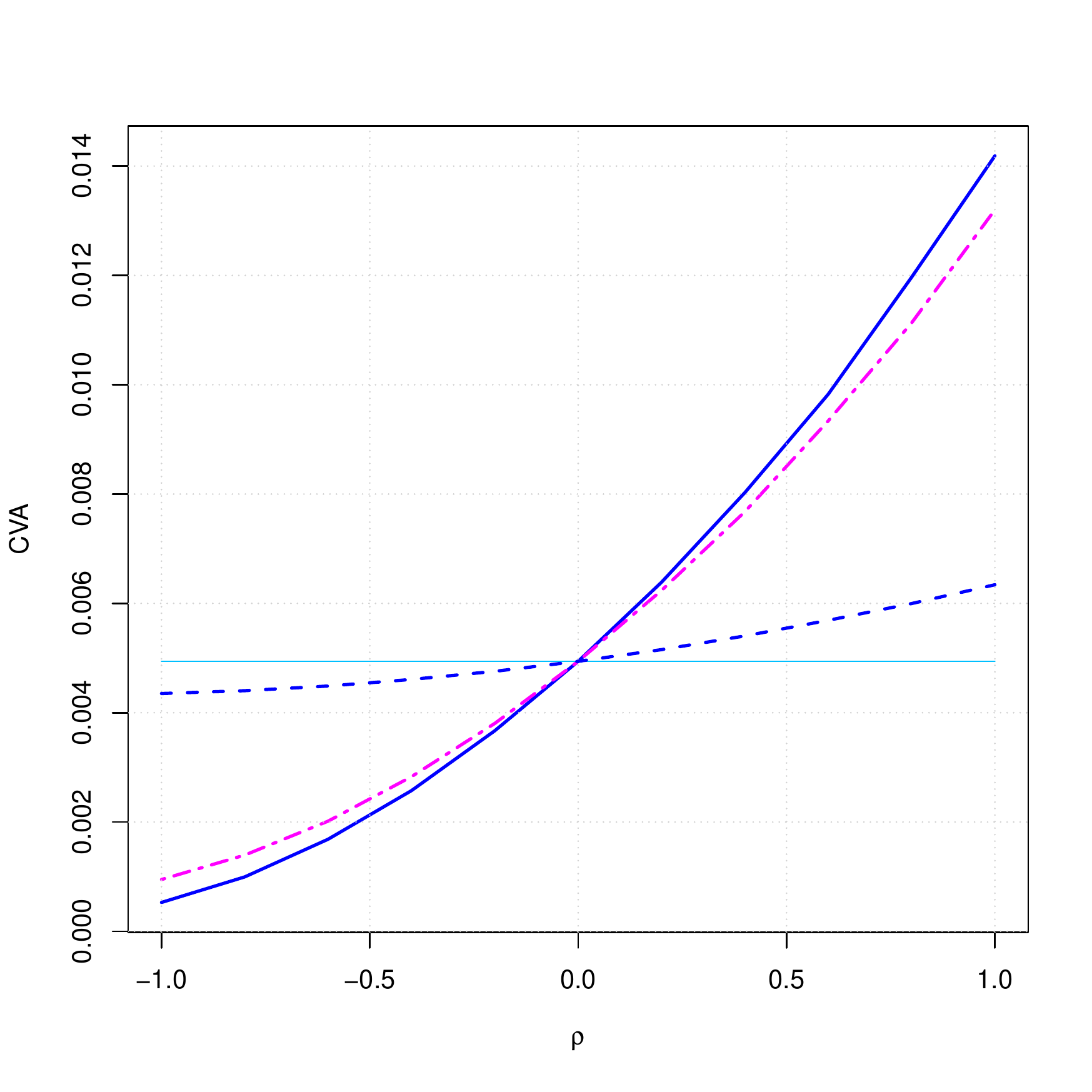}}\hspace{0.20cm}
\subfigure[$dV_t=\left(\gamma(T-t)-\frac{V_t}{T-t}\right) dt + \nu dW^V_t$, $y_0=h_0$]{\includegraphics[width=0.48\columnwidth]{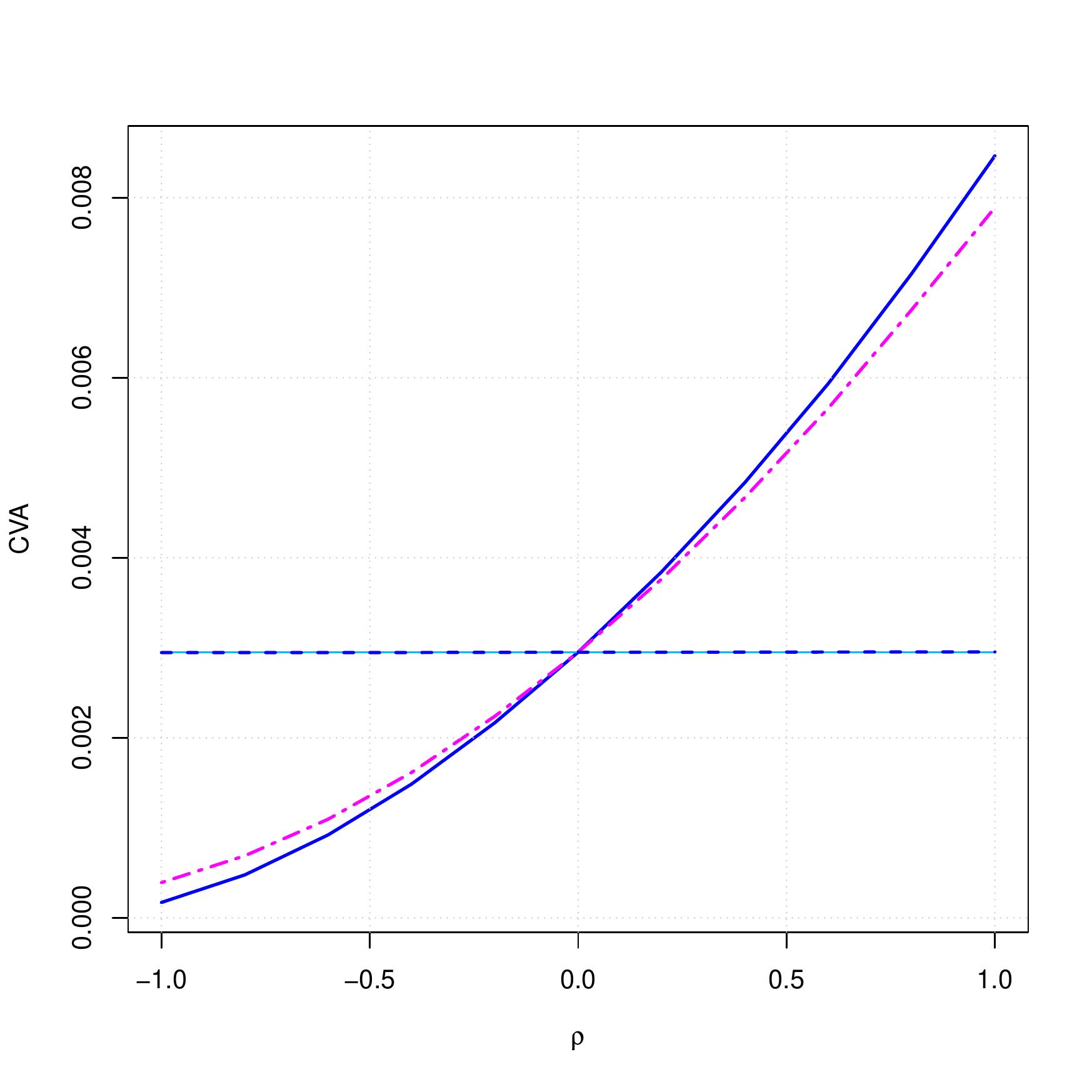}}\hspace{0.20cm}
\subfigure[ $dV_t=\left(\gamma(T-t)-\frac{V_t}{T-t}\right) dt + \nu dW^V_t$, $y_0$ optimized]{\includegraphics[width=0.48\columnwidth]{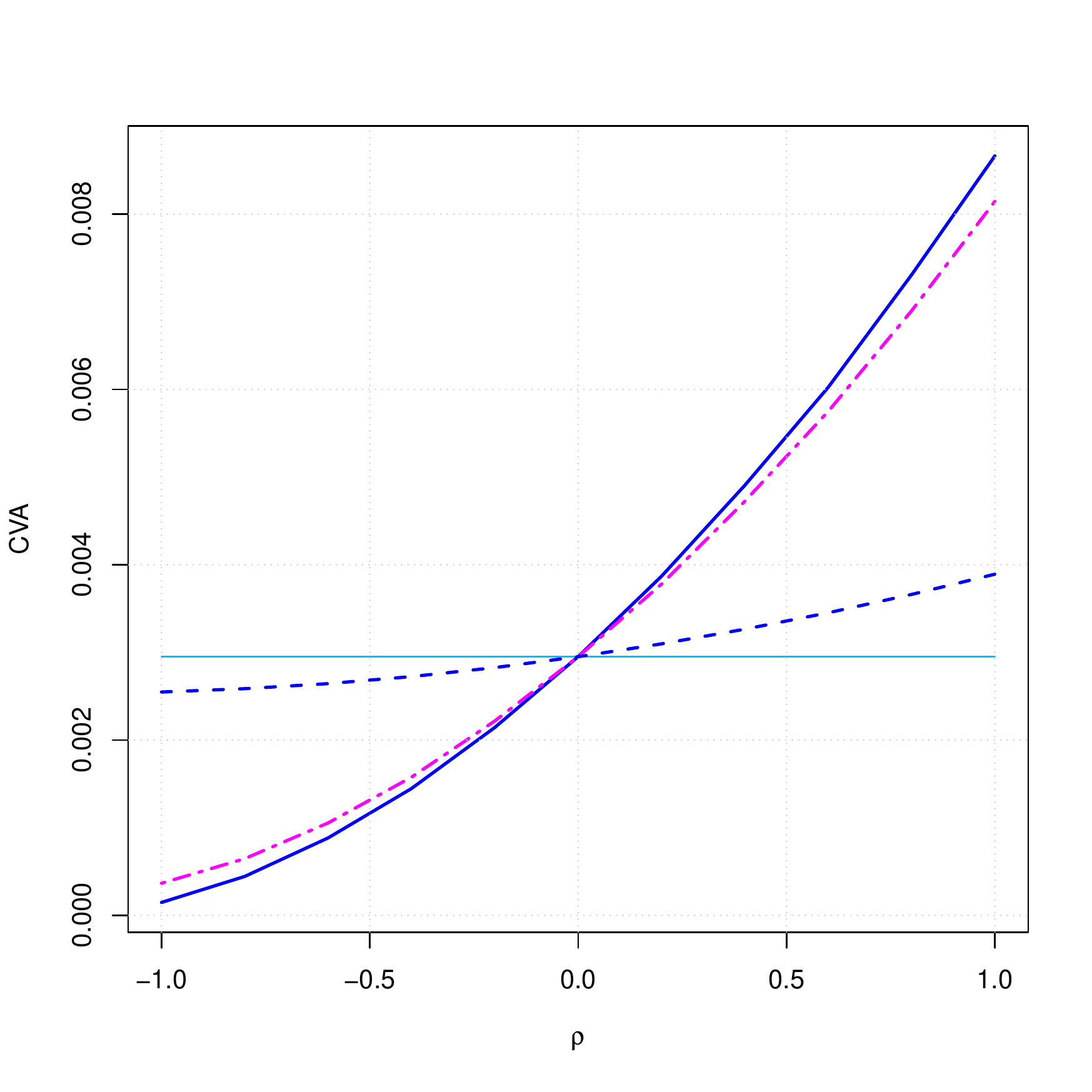}}
\caption{Impact of the exposure-credit correlation $\rho$ on CVA levels for prototypical 5Y Forward (top) and Swap exposure (bottom) with $y_0=h_0$ (left) or $y_0$ optimized (right), $\nu=8\%$ and $\gamma=0.1\%$. The curves correspond to different intensity models: $\lambda^{\varphi,+}$ (PS-CIR $\Xi=\Xi^{\star,+}$ (left) and $\Xi=\Xi_0^{\star,+}$ (right), dotted blue), $\lambda^{\varphi}$ (S-CIR $\Xi=\Xi^\star$ (left) and $\Xi=\Xi_0^\star$ (right), solid blue), and $\lambda^{\theta}$ (TC-CIR $\Xi=\Xi^\star$ (left) and $\Xi=\Xi_0^\star$ (right), dotted magenta). The case without wrong-way risk corresponds to the flat (cyan) line.}\label{fig:CVAplots}
\end{figure}
\else
\begin{center}
[Figure 6 about here]
\end{center}
\fi 

\section{Conclusion}

The calibration problem consists of finding the parameters of a model $x$ so as to perfectly fit a given market curve. The perfect fit is an important feature in a pricing context, that is connected to  no-arbitrage opportunities and corrects valuation of trading positions. This calls for two important features: the model $x$ must be (i) flexible enough (to be able to generate various shapes) and (ii) tractable enough (to facilitate the parameters' optimization procedure). Time-homogeneous affine models like Vasicek, CIR or JCIR are very good candidates in this respect, and are widely used in interest rates and credit risk modeling. However, as such, they only feature a couple of constants and hence lack calibration flexibility. The deterministic shift extension offers an appealing solution. It consists of starting with a tractable base model $y$, that is shifted in a deterministic way with a function $\varphi$. The resulting process $x_t=y_t+\varphi(t)$ becomes fully flexible. Indeed, any discount or survival probability curve can be generated by such a model. Moreover, it has a tractability level that is very similar to that of $y$ because $\varphi$ is deterministic. Eventually, for every market curve, the shift $\varphi^\star$ that leads to the perfect fit is known in closed form, as a function of the $y$ parameters and the market curve. However, this method is less appealing when the model $x$ needs to fulfill some range constraints.  Among those, non-negativity is of primary importance when modeling interest rates (depending on the type of economy at hand), mortality rate, prepayment rate or default intensities. In the deterministic shift approach indeed, starting with a non-negative base process $y$ is not enough to guarantee that so will be $x$, without additional constraint on $\varphi$. Furthermore, this constraint becomes more and more severe when increasing the process volatility, due to the zero lower bound.\medskip

It seems obvious to rule out models allowing for ``negative volatilities''. However, surprisingly, the same does not seem to apply when it comes to ``negative intensities''. Yet, both are equally flawed.  We believe the reason is twofold: first, negative intensities do not directly generate numerical problems (in contrast with volatilities that often appear in square-roots), so that the issue is less ``obvious'', second, there is a lack of a sound alternative. The positivity constraint can be dealt with by including a non-negativity constraint on $\varphi$. However, this again raises two problems. First the parameter optimization problem becomes more difficult and second, the resulting process $x$ then features a much lower variance than without the constraint, which contradicts empirical evidences. Therefore, one often prefers to disregard the ``negative intensities'' issue, giving the priority to stochasticity and perfect fit.\medskip 

In this paper, we develop such an alternative. It simply consists of time-changing a positive homogeneous affine jump-diffusion. The model remains tractable, positive, the optimal clock is found by simple inversion and features larger implied volatility compared to the shift approach. Moreover, the perfect fit is achievable for a broad class of discount curves, including all decreasing discount curves. 
The features of the model have been illustrated on topical examples taken from credit risk, but other applications could be considered as well. This method thus proves to be a competitive challenger to the shift approach, at least under the (very common) positivity constraint, and when large volatility levels are needed. 

\section{Appendix}

\subsection{Properties of some Homogeneous Affine Jump-Diffusions}
\label{sec:AppAffine}
Let $y$ be a $\FF$-adapted jump-diffusion introduced in Definition \ref{eq:affineJD} and $Y_t:=\int_0^t y_u du$ its integrated version. We denote  $v^x(t)=v^x(t;\Xi):=\mathbb{V}[x_t]$ the variance of a stochastic process $x$ parametrized by $\Xi$ at time $t$. Without explicit mention, all the results below that are given without proofs can be found in, e.g., \cite{Brigo06}. New results are given in lemmas for further reference.

\subsubsection{Vasicek model}
The Vasicek model corresponds to the special HAJD case  $(a(t),b(t),c(t),d(t),\alpha,\omega(t))=(\kappa\beta,-\kappa,\eta^2,0,0,0)$. The $A^y,B^y$ functions in \eqref{eq:affine} are given by:
\begin{eqnarray}
A^y_s(t;\Xi)&=& \left(\beta-\frac{\eta^2}{2\kappa^2}\right)\left(B^y_s(t;\Xi)+s-t\right)-\frac{\eta^2}{4\kappa}B^y_s(t;\Xi)^2\;,\nonumber\\
B^y_s(t;\Xi)&=&\frac{1}{\kappa}\left(1-\e^{-\kappa(t-s)}\right)\;.\nonumber
\end{eqnarray}
The forward curve associated to this model is proven to be 
\beq
f_s^{\mathrm{VAS}}(t) := (1-\e^{-\kappa (t-s)})\frac{\kappa^2\beta-\eta^2/2}{\kappa^2}+\frac{\eta^2}{2\kappa^2}\e^{-\kappa (t-s)}(1-\e^{-\kappa (t-s)})+y_t\e^{-\kappa (t-s)}\;.\label{eq:fvast}
\eeq
Moreover, 
both $y$ and $Y$ are Normally distributed at all times, with
%
\beqn
\E[y_t]&=&y_0^{-\kappa t}+\beta\kappa B^y_0(t;\Xi)\;,\nonumber\\
\E[Y_t]&=&y_0B^y_0(t;\Xi)+\beta\left(t-B^y_0(t;\Xi)\right)\;,\nonumber\\
v^y(t)&=&\frac{\eta^2}{2\kappa}\left(1-e^{-2\kappa t}\right)\;,\nonumber\\
v^Y(t)&=&\frac{\eta^2}{\kappa^2}\left[t+\frac{1-e^{-2\kappa t}}{2\kappa}-2B^y_0(t;\Xi)\right]\;.\nonumber
\eeqn

\begin{lemma}\label{Lem:VarVAS}
Let $y$ be a Vasicek process and $Y$ its time integral. The functions $v^y(t)$ and  $v^Y(t)$ are increasing with respect to $t$. 
\end{lemma}

\begin{proof}
It is obvious for $v^y(t)$, and for $v^Y(t)$, a few manipulations lead to
$$\frac{d}{dt}v^Y(t)=\left(\frac{\eta}{\kappa}(1-e^{-\kappa t})\right)^2\geq 0\;.$$
\end{proof}

\subsubsection{CIR model}
\label{sec:AppAffineCIR}
The CIR model corresponds to the special HAJD case $(a(t),b(t),c(t),d(t),\alpha,\omega(t))=(\kappa\beta,-\kappa,0,\delta^2,0,0)$. The $A^y,B^y$ functions in eq. \eqref{eq:affine} are given by:
\beqn
A^y_s(t;\Xi)&=& \frac{2\kappa\beta}{\delta^2}\ln\frac{2\gamma\exp\{(\kappa+\gamma)(t-s)/2)\}}{2\gamma+(\kappa+\gamma)(\exp\{(t-s)\gamma\} - 1)} ,\nonumber\\
B^y_s(t;\Xi)&=&\frac{2(\exp\{(t-s)\gamma\} - 1)}{2\gamma+(\kappa+\gamma)(\exp\{(t-s)\gamma\} - 1)}\;.\nonumber
\eeqn
where $\gamma:=\sqrt{\kappa^2+2\delta^2}$. The forward curve associated to this model is given by \cite{Brigo06}
\beq
f_s^{\mathrm{CIR}}(t) := \frac{2\kappa\beta(e^{(t-s)\gamma}-1)}{2\gamma+(\kappa+\gamma)(e^{(t-s)\gamma}-1)} + y_t\frac{4\gamma^2e^{(t-s)\gamma}}{[2\gamma+(\kappa+\gamma)(e^{(t-s)\gamma}-1)]^2}\;.\label{eq:fcirt}
\eeq

Important characteristics of the CIR processes can be computed explicitly, (see, e.g., \cite{Dufresne90}). For instance, $y$ is distributed as a non-central chi-squared. The two first order moments of $y$ and $Y$ are respectively given by
\beqn
\E[y_t]&=&y_0e^{-\kappa t}+\beta\left(1-e^{-\kappa t}\right)\;,\nonumber\\
\E[y_t^2]&=&y_0^2e^{-2\kappa t}+\left(\frac{\delta^2}{2\kappa}+\beta\right)\left[2y_0\left(e^{-\kappa t}-e^{-2\kappa t} \right) +\beta\left(1-e^{-\kappa t} \right)^2\right]\;,\nonumber\\
\E[Y_t]&=& t\beta + \frac{(y_0-\beta)}{\kappa}(1-e^{-\kappa t})\;,\nonumber\\
\E[Y_t^2]&=&\left(\frac{y_0-\beta}{\kappa}\right)^2 + \delta^2\left(\frac{y_0}{\kappa^3}-\frac{5\beta}{2\kappa^3}\right) + t\beta\left(\frac{2y_0}{\kappa}-\frac{2\beta}{\kappa}+\frac{\delta^2}{\kappa^2}\right)+t^2\beta^2\nonumber\\
&& +\; e^{-\kappa t}\left[-2\left(\frac{y_0-\beta}{\kappa}\right)^2+\frac{2\beta\delta^2}{\kappa^3} +t\left( -\frac{2y_0\beta}{\kappa}+\frac{2\beta^2}{\kappa}-\frac{2y_0\delta^2}{\kappa^2}+\frac{2\beta\delta^2}{\kappa^2}\right)\right]\nonumber\\
&& +\; e^{-2\kappa t}\left[\left(\frac{y_0-\beta}{\kappa}\right)^2-\frac{y_0\delta^2}{\kappa^3}+\frac{\beta\delta^2}{2\kappa^3} \right]\;.\nonumber
\eeqn

In contrast with the Vasicek model, the variance of the CIR is not always increasing monotonously with time; it depends on the parameters. However, the variance of the integrated CIR is increasing. These properties are proven in the next lemma, and will be central in the proof of Theorem~\ref{th:th2}.

\begin{lemma}\label{Lem:VarCIR}
Let $y$ be a CIR process and $Y$ its time integral. 
Then,

\beqn
v^y(t)&=&\frac{\delta^2}{\kappa}\left[y_0(e^{-\kappa t}-e^{-2\kappa t})+\frac{\beta}{2\kappa}(1-e^{-\kappa t})^2\right]\;,\label{eq:VarCIR}\\
v^Y(t)&=&
\frac{\delta^2}{\kappa^3}\left[\left(2\beta -2\kappa t(y_0-\beta)-\left(y_0-\beta/2\right)e^{-\kappa t}\right)e^{-\kappa t}+t\kappa\beta+\left(y_0-5\beta/2\right)\right]\;.\label{eq:VarIntCIR}
\eeqn
The function $v^y(t)$ is increasing if $\beta\geq y_0$. Otherwise, it is first increasing up to a time $t^\star$, and then decreasing on $(t^\star,\infty)$. By contrast, $v^Y(t)$ is always increasing.
\end{lemma}

\begin{proof}
The computation of the variances is trivial from the first two moments recalled above. The derivative of the variance of the CIR is given by
\begin{eqnarray*}
\frac{d}{dt}v^y(t) &=& y_0\delta^2(-e^{-\kappa t}+2e^{-2\kappa t})+\beta\delta^2e^{-\kappa t}(1-e^{-\kappa t})\\
&=& \delta^2e^{-\kappa t}(\beta-y_0)-\delta^2e^{-2\kappa t}(\beta-2y_0)\;.
\end{eqnarray*}
This expression has a root on the positive half-line at
\begin{equation}
t^\star=\frac{1}{\kappa}\ln\left(1+ \frac{y_0}{y_0-\beta}\right)\nonumber
\end{equation}
only if $y_0>\beta$. Otherwise, $v^y(t)$ is always increasing in $t$.
The derivative of the variance of the integrated CIR with respect of time can be written, after some manipulations, as
$$\frac{d}{dt}v^Y(t)=\frac{\delta^2}{\kappa^2}\left[2y_0e^{-\kappa t}(e^{-\kappa t}-(1-\kappa t))+\beta(1-(2\kappa t +e^{-\kappa t})e^{-\kappa t})\right]\;.$$
Because $e^{-x}\geq 1-x$ for all $x\geq 0$ and $y_0,\kappa$ are positive constants, the first term is positive for all $t\geq 0$. On the other hand,  $\beta>0$, and it is enough to check that $1-g(\kappa t)\geq 0$ for all $t\geq 0$, with $g(x):=(2x+e^{-x})e^{-x}$. Clearly, $g(0)=1$ and $g'(x)=2e^{-x}(1-x-e^{-x})\leq 0$ for all $x\geq 0$. Hence, $1-g(\kappa t)\geq 0$ for all $t\geq 0$. 
\end{proof}

\subsubsection{JCIR model}\label{sec:AppJcir}

The characteristics of the JCIR can be obtained by adjusting those of the corresponding CIR, i.e., with same initial value and diffusion parameters. We note $z$ the former and $y$ the latter, and similarly for their integrated versions ($Z$ and $Y$, respectively). Hence, if the parameter set for the CIR ($y,Y$) is $\Xi_0=(\kappa,\beta,\delta,0,0,y_0)$, the parameter set of the corresponding JCIR ($z,Z$) is $\Xi_0=(\kappa,\beta,\delta,\alpha,\omega,z_0)$ with $z_0=y_0$ and $\alpha,\omega\geq 0$. The functions associated to the discount curve are given by 
\beqn
A^z_s(t;\Xi)&=&A^y_s(t;\Xi) + {\frac{\alpha\omega}{\delta^2/2-\kappa\alpha-\alpha^2}}\ln\frac{2\gamma e^{\frac{\gamma+\kappa+2\alpha}{2}(t-s)}}{2\gamma+(\kappa+\gamma+2\alpha)(e^{(t-s)\gamma }-1)}\;,\nonumber\\
B^z_s(t;\Xi)&=&B^y_s(t;\Xi)\;.\nonumber
\eeqn
From the above functions, it is easy to see that the forward curve associated to this model reads as 
\beq
f_s^{\mathrm{JCIR}}(t) := f^{\rm CIR}_s(t)+\frac{2\omega\alpha(e^{(t-s)\gamma} - 1)}{2\gamma + (\kappa+\gamma+2\alpha)(e^{(t-s)\gamma} - 1)}\;,\label{eq:fjcirt}
\eeq
where $f^{\rm CIR}_s(t)$ is given in \eqref{eq:fcirt}. For every valid parameters, $f_s^{\mathrm{JCIR}}(t) \geq f^{\rm CIR}_s(t)$ for all $t\geq s$. Regarding the moments, we have the following result.

\begin{lemma}\label{Lem:VarJCIR} Let $y$ (resp. $Y$) be a CIR (resp. integrated CIR) and $z$ (resp. $Z$) be a JCIR (resp. integrated JCIR) with same initial value, same diffusion parameters but with jumps governed by $(\omega,\alpha)$. Then,
\beqn
\E[z_t]&=&\E[y_t]+\frac{\omega\alpha}{\kappa}(1-e^{-\kappa t})\;,\nonumber\\
\E[Z_t]&=&\E[Y_t]+\frac{\omega\alpha}{\kappa^2}\left(\kappa t -(1-e^{-\kappa t})\right)\;. \nonumber\\
v^z(t)&=&v^y(t)+\omega\alpha\left[ \frac{\delta^2}{2}\left(\frac{1-e^{-\kappa t}}{\kappa}\right)^2+\alpha\frac{1-e^{-2\kappa t}}{\kappa}\right]\;,\nonumber\\
v^Z(t)&=&v^Y(t)+
\frac{\alpha\omega}{\kappa^3}\left[\frac{1-e^{-\kappa t}}{\kappa}\left(\xi(3-e^{-\kappa t})-4\delta^2\right)+2\delta^2t e^{-\kappa t}+t\left(2\alpha\kappa+\delta^2\right)\right]\;,\nonumber
\eeqn
where $\xi:=\delta^2/2-\alpha\kappa$. The function $v^z(t)$ is increasing with respect to $t$ unless $y_0>\beta+\omega\alpha/\kappa$, in which case it is first increasing up to a time $t_1$, and then decreasing on $(t_1,\infty)$. Moreover, $v^z(t)\geq v^y(t)$, $v^Z(t)\geq v^Y(t)$ and $v^Z(t)$ is always increasing.
\end{lemma}

\begin{proof}
Applying Ito's lemma we can solve the JCIR SDE \eqref{eq:JCIR} by
\begin{equation}\label{eq:jcir:sol}
    z_t = z_0e^{-\kappa t} + \beta(1-e^{-\kappa t})+\delta\int_0^te^{-\kappa(t-s)}\sqrt{z_s}dW_s + \int_0^te^{-\kappa(t-s)}dJ_s\;,
\end{equation}
and find the SDE governing the integrated JCIR process
\begin{equation}\label{eq:intjcir:sol}
    Z_t = t\beta + \frac{(z_0-\beta)}{\kappa}(1-e^{-\kappa t}) + \delta\int_0^t\int_0^se^{-\kappa(s-u)}\sqrt{z_u}dW_u ds + \int_0^t\int_0^se^{-\kappa(s-u)}dJ_u ds\;.
\end{equation}
From \eqref{eq:jcir:sol}, we can write
\begin{eqnarray*}
\mathbb{E}[z_t] &=& z_0e^{-\kappa t} + \beta(1-e^{-\kappa t}) + \mathbb{E}\left[\int_0^te^{-\kappa(t-s)}dJ_s\right]\\
&=&\mathbb{E}[y_t]+\omega\alpha\int_0^te^{-\kappa(t-s)}ds\\
&=&\mathbb{E}[y_t]+\frac{\omega\alpha}{\kappa}(1-e^{-\kappa t})\;.
\end{eqnarray*}
Using Ito isometry,
\begin{eqnarray*}
v^z(t)&=&\mathbb{E}[(z_t-\mathbb{E}[z_t])^2]\\
&=&\mathbb{E}\left[\left(\delta\int_0^te^{-\kappa(t-s)}\sqrt{z_s}dW_s + \int_0^te^{-\kappa(t-s)}dJ_s-\omega\alpha\int_0^te^{-\kappa(t-s)}ds\right)^2\right]\\
&=&\mathbb{E}\left[\left(\delta\int_0^te^{-\kappa(t-s)}\sqrt{z_s}dW_s\right)^2\right] + \mathbb{E}\left[\left(\int_0^te^{-\kappa(t-s)}(dJ_s-\omega\alpha ds)\right)^2\right]\\
&=&v^y(t)+\frac{\delta^2\omega\alpha}{\kappa}\int_0^t\e^{-2\kappa(t-s)}(1-e^{-\kappa s})ds + 2\omega\alpha^2\int_0^te^{-2\kappa(t-s)}ds\\
&=&v^y(t)+\omega\alpha\left[ \frac{\delta^2}{2}\left(\frac{1-e^{-\kappa t}}{\kappa}\right)^2+\alpha\frac{1-e^{-2\kappa t}}{\kappa}\right]\;.
\end{eqnarray*}
Using a similar procedure applied to \eqref{eq:intjcir:sol}  combined with Fubini's theorem, one can derive the expectation and variance of the integrated JCIR : 

\begin{eqnarray*}
\mathbb{E}[Z_t] &=& t\beta + \frac{(z_0-\beta)}{\kappa}(1-e^{-\kappa t}) + \mathbb{E}\left[\int_0^t\int_0^se^{-\kappa(s-u)}dJ_u ds\right]\\
&=&\mathbb{E}[Y_t]+\mathbb{E}\left[\int_0^t\int_s^te^{-\kappa u}due^{\kappa s}dJ_s\right]\\
&=&\mathbb{E}[Y_t]+\frac{\omega\alpha}{\kappa}\int_0^t(1-e^{-\kappa(t-s)})ds\\
&=&\mathbb{E}[Y_t]+\frac{\omega\alpha}{\kappa^2}\left(\kappa t -(1-e^{-\kappa t})\right)
\end{eqnarray*}
and 
\begin{eqnarray*}
v^Z(t)&=&\mathbb{E}[(Z_t-\mathbb{E}[Z_t])^2]\\
&=&\mathbb{E}\left[\left(\delta\int_0^t\int_0^se^{-\kappa(s-u)}\sqrt{z_u}dW_uds + \int_0^t\int_0^se^{-\kappa(s-u)}dJ_uds-\mathbb{E}\left[\int_0^t\int_0^se^{-\kappa(s-u)}dJ_u ds\right]\right)^2\right]\\
&=&\mathbb{E}\left[\left(\frac{\delta}{\kappa}\int_0^t(1-e^{-\kappa(t-s)})\sqrt{z_s}dW_s + \frac{1}{\kappa}\int_0^t(1-e^{-\kappa(t-s)})dJ_s-\frac{\omega\alpha}{\kappa}\int_0^t(1-e^{-\kappa(t-s)})ds\right)^2\right]\\
&=&\mathbb{E}\left[\left(\frac{\delta}{\kappa}\int_0^t(1-e^{-\kappa(t-s)})\sqrt{z_s}dW_s\right)^2\right] + \mathbb{E}\left[\left(\frac{1}{\kappa}\int_0^t(1-e^{-\kappa(t-s)})(dJ_s-\omega\alpha ds)\right)^2\right]\\
&=&v^Y(t)+\frac{\delta^2\omega\alpha}{\kappa^2}\int_0^t(1-e^{-\kappa(t-s)})^2(1-e^{-\kappa s})ds + \frac{2\omega\alpha^2}{\kappa^2}\int_0^t(1-e^{-\kappa(t-s)})^2ds\\
&=&v^Y(t)+\frac{\alpha\omega}{\kappa^3}\left[\frac{1-e^{-\kappa t}}{\kappa}\left((\delta^2/2-\alpha\kappa)(3-e^{-\kappa t})-4\delta^2\right)+2\delta^2t e^{-\kappa t}+t\left(2\alpha\kappa+\delta^2\right)\right]\;.
\end{eqnarray*}
Notice that the above results can be obtained using another procedure, namely by deriving once or twice the characteristic function  $\Psi_t(u,v)=\E[e^{uz_t+vZ_t}]$ of $(z_t,Z_t)$ which can be recovered from eq. (A.1) in \cite{Duffie01}. This procedure is however much heavier. \footnote{Be aware that there are typos in this formula. The correct expression can be found in eq. (B.9) in the draft version of Duffie and G\^arleanu's paper, that is available for download on the authors' webpage.} 

The derivative of the variance of the JCIR is given by
\begin{eqnarray*}
\frac{d}{dt}v^z(t) &=& y_0\delta^2(-e^{-\kappa t}+2e^{-2\kappa t})+\beta\delta^2e^{-\kappa t}(1-e^{-\kappa t})+\frac{\delta^2\omega\alpha}{\kappa}e^{-\kappa t}(1-e^{-\kappa t})+2\omega\alpha^2e^{-2\kappa t}\\
&=& \delta^2e^{-\kappa t}(\beta-y_0+\omega\alpha/\kappa)-\delta^2e^{-2\kappa t}(\beta-2y_0+\omega\alpha/\kappa-2\omega\alpha^2/\delta^2)\;.
\end{eqnarray*}
This expression has a root on the positive half-line at
\begin{equation}\label{eq:RootVarJCir}
t_1:=\frac{1}{\kappa}\ln\left(1+ \frac{y_0+2\omega\alpha^2/\delta^2}{y_0-\beta-\omega\alpha/\kappa}\right)
\end{equation}
only if $y_0>\beta+\omega\alpha/\kappa$. Otherwise, $v^y(t)$ is always increasing in $t$.

It seems intuitive that the variance of the JCIR cannot be smaller than that of the CIR, and similarly for the integrated versions. Because of the mean reverting effect however, this needs to be confirmed. 
It is obvious that $v^z(t)-v^y(t)\geq 0$. The term associated to $v^Z(t)-v^Y(t)$ starts at zero (since obviously $v^Z(0)=v^Y(0)=0$). This difference is increasing:
\begin{equation*}
    \frac{d}{dt} (v^Z(t)-v^Y(t)) =   \frac{\delta^2\omega\alpha}{\kappa^3}\left(1-(2\kappa t + e^{-\kappa t})e^{-\kappa t}\right) + \frac{2\omega\alpha^2}{\kappa^2}\left(1-e^{-\kappa t}\right)^2\geq 0\;.
\end{equation*}
Indeed, the second term is obviously positive and the first term takes the form $\frac{\delta^2\omega\alpha}{\kappa^3}(1-g(\kappa t))$ where the function $g(x)=(2x + e^{-x})e^{-x}$ is shown to be bounded by 1 for $x\geq 0$ in the proof of Lemma \ref{Lem:VarCIR}. This shows that $v^Z(t)\geq v^Y(t)$. Because both $v^Y(t)$ (from Lemma \ref{Lem:VarCIR}) and $v^{Z}(t)-v^Y(t)$ are increasing; $v^Z(t)$ is itself increasing. 
\end{proof}

\subsection{Some special cases where $P^{x+y}$ is a discount curve}\label{sec:AppDiscCurve}

Observe first that $P^{x+y}$ is a discount curve whenever $P^x,P^y$ are in the case where $x,y$ are independent since then $P^{x+y}_s=P^{x}_sP^{y}_s$, and the product of two time-$s$ discount curves is itself a time-$s$ discount curve. The next lemma provides sufficient conditions on $y$ for $P^{y}$ to be a discount curve in the general case. 

\begin{lemma}\label{lemma:DscCurveSum}Let $T$ be a fixed time horizon. Then, $P^y$ is a discount curve whenever $y$ is positive and $\sup_{t\in [0,T]} y_t$ is integrable.
\end{lemma}
\begin{proof} 

We start with the lemma giving sufficient conditions to swap the expectation and derivative operators which can be found in, e.g., \cite{Gill18}. 

\begin{lemma}\label{lemma:SwapED}
Let $I$ be a nontrivial interval of $\mathbb{R}$, $\cB(I)$ the Borel set of $I$ and $\Psi : I\times \Omega\to \mathbb{R}\;, (x,\omega)\mapsto \Psi(x,\omega)$ be a $\cB(I)\otimes \mathcal{G}$-measurable function. If the function $\Psi$ satisfies:
\begin{itemize}
\item[(i)] For every $x\in I$, the random variable $\Psi(x,\,\omega)\in L^1$, 
\item[(ii)] $\Psi_x(x,\omega):=\frac{\partial\Psi(x,\omega)}{\partial x}$ exists for all $x\in I$ a.s.,
\item[(iii)] There exists $Z\in L^1$ 
such that for every $x\in I$,
$$\quad \left\lvert \Psi_x(x,\omega) \right\lvert\leq Z(\omega)\;a.s.\;.$$
\end{itemize}
Then the function $\psi(x):=\mathbb{E}[\Psi(x,\,\omega)]$ is defined and differentiable at every $x\in I$ with derivative
$$\frac{d\psi(x)}{dx}=\mathbb{E}\left[\Psi_x(x,\omega)\right]\;.$$
\end{lemma}

We now proceed with the proof of Lemma \ref{lemma:DscCurveSum}.\\
Let us fix $t\leq T$. Hence,
$$\left\lvert\frac{d}{d t}e^{-\int_0^t y_u du}\right\rvert= |y_te^{-\int_0^t y_u du}|\leq y_t\leq \sup_{t\in [0,T]} y_t\;$$
for all $t\in[0,T]$. Noting that $\sup_{t\in [0,T]} y_t$ is integrable, one can use Lemma \ref{lemma:SwapED} with $\Psi(t,w)\leftarrow e^{-\int_0^t y_u(w) du}$ and $Z(\omega)\leftarrow S_T^y:=\sup_{t\in [0,T]} y_t$, justifying the swap between the derivative and expectation operators:
\beq
\frac{d }{d t} P^y(t)=\frac{d}{d t}\E\left[e^{-\int_0^t y_u du}\right]=\E\left[\frac{d}{d t}e^{-\int_0^t y_u du}\right]=-\E\left[y_te^{-\int_0^t y_u du}\right]\;,\nonumber
\eeq
where the right-hand side is bounded by the expectation of $Z$, which is integrable. This concludes the proof.
\end{proof}
In order for thE assumption about the integrability of the running supremum of $y$ to be useful in practice, it needs to be ``checkable'. Hence, we need to give simpler sufficient conditions (e.g., based on the coefficients of the SDE of $y$) that would guarantee that $S_T^y:=\sup_{t\in [0,T]} |y_t|$ satisfies $\E[S_T^y]<\infty$.
\begin{lemma}\label{lemma:DscCurveGen} Let $W$ be a Brownian motion, $J$ a compound Poisson process with constant jump intensity $\omega$ and the jump sizes are exponentially distributed with mean $\alpha$, and $y$ solving
\[
dy_t=\mu(t,y_t)dt+\sigma(t,y_t)dW_t + dJ_t
\]
where $y_0$ is positive, $\E[\int_0^T |\mu(t,y_t)|dt]<\infty$ and $\E[\int_0^T \sigma^2(t,y_t)dt]<\infty$. Then, 
\[
\E[|S^y_T|]=\E[S^y_T]<\infty~~\textrm{ where }S^y_T:=\sup_{t\in[0,T]} |y_t|\;.
\]
\end{lemma}
\begin{proof} 
The solution of the SDE is
\[
y_t=\underbrace{y_0+\int_0^t \mu(s,y_s)ds}_{A_t} + \underbrace{\int_0^t \sigma(s,y_s)dW_s}_{M_t} + J_t~\Rightarrow ~|y_t|\leq|A_t| + |M_t| + J_t\;,
\]
showing that
\[
S_T^y\leq \underbrace{\sup_{t\in [0,T]} |A_t|}_{S^A_T} + \underbrace{\sup_{t\in [0,T]} |M_t|}_{S^M_T} +  \underbrace{\sup_{t\in [0,T]} J_t}_{S^J_T}\;.
\]
We show in the sequel that $S^A_T$, $S^M_T$ and $S^J_T$ are integrable. This would conclude the proof since it would lead to
\[
\E[|S^y_T|]=\E[S^y_T]\leq \E[S^A_T] + \E[S^M_T] + \E[S^J_T]<\infty\;.
\]
Suppose that $\E[\int_0^T |\mu(s,y_s)| ds]<\infty$. Then,
\[
S^A_T=\sup_{t\in [0,T]}|y_0+\int_0^t \mu(s,y_s)ds|\leq |y_0| + \sup_{t\in [0,T]} \int_0^t |\mu(s,y_s)|ds\leq |y_0|+\int_0^T |\mu(s,y_s)|ds\;.
\]
showing that $\E[|S^A_T|]=\E[S^A_T]\leq |y_0|+\E[\int_0^T |\mu(s,y_s)|ds]<\infty$.\medskip

On the other hand, $M$ is a martingale, so that $|M|$ is a submartingale:
\[
\E[|M_t|\;|\cF_s]\geq|\E[M_t|\cF_s]|=|M_s|\;.
\]
We can then apply Doob's inequality,
\[
\E[S^M_T]=\E[\sup_{t\in [0,T]} |M_t|]\leq\frac{e}{e-1}\left(1+\E[|~|M_T|\log^+ |M_T|~|]\right)\;.
\]
Using $-\frac{1}{e}\leq x\log x\leq x^2$ for $x\geq 0$, $|x\log^+ x|=x\log^+ x\leq |x\log x|\leq \max(e^{-1},x^2)$:
\beqn
\E[|M_T|\log^+ |M_T|]&\leq& \E[\max(e^{-1},M_T^2)]\nonumber\\
&=&\E[e^{-1}1_{\{M_T^2\leq e^{-1}\}}+M_T^21_{\{M_T^2> e^{-1}\}}]\nonumber\\
&\leq&e^{-1}+\E[M_T^2]\;.\nonumber
\eeqn
Hence,
\[
\E[S^M_T]\leq\frac{e}{e-1}\left(1+e^{-1}+\E[M_T^2]\right)\;.
\]
Using Ito isometry, $\E[M_T^2]=\E\left[\int_0^T\sigma^2(t,y_t)dt\right]$ which is bounded, by assumption.\medskip

Similarly, one can prove that $\E[S^J_T]$ is finite by applying the Doob's inequality to the martingale $(J_t - \omega\alpha t), t\leq T$. Indeed,
\[J_t = |J_t - \omega\alpha t + \omega\alpha t|\leq |J_t - \omega\alpha t| + \omega\alpha t,\quad \forall\, t\leq T\;,\]
which implies that
\beqn
\E[S^J_T]&\leq &\E[\sup_{t\in [0,T]}|J_t - \omega\alpha t|] + \omega\alpha T \nonumber\\
&\leq & \frac{e}{e-1}\left(1+e^{-1}+\E[(J_T-\omega\alpha T)^2] \right) + \omega\alpha T \nonumber\\
&=& \frac{e}{e-1}\left(1+e^{-1}+2\omega\alpha^2 T \right)  + \omega\alpha T <\infty\;. \nonumber
\eeqn
\end{proof}

One can check that $P^{x+y}$ is a discount curve when $x,y$ are HAJD, possibly driven by correlated Brownian motions. Indeed, they satisfy the assumptions of Lemma~\ref{lemma:DscCurveGen}.

\subsection{Proof of Theorem \ref{th:th1}}\label{seq:ProofTh1}

Observe first that for every $\theta$ and every $t$, one gets
\beq
\int_0^t x^{\theta}_u du=\int_0^t\theta(u) y_{\theta(u)}du=\int_0^{\Theta(t)} y_udu\;.\nonumber 
\eeq
Hence, the expectation of their negative exponentials agree as well: 
\beq
P^{x^\theta}(t)= P^y(\Theta(t))\;.\nonumber
\eeq
The specific clock rate $\theta^\star$ given by the calibration equation thus satisfies, for all $t$,
\beq
P^{market}(t)= P^{x^{\theta^\star}}(t)= P^y(\Theta^\star(t))\;.\label{eq:TC-cal}
\eeq
Turning this equality in terms of instantaneous forward rates yields
\beq
\int_0^tf^{market}(u)du=\int_0^{\Theta^\star(t)} f^y(u)du\;.\nonumber
\eeq
Eq. \eqref{eq:ode} is just the differential form of the latter.\medskip 

It is not clear, in general, to determine when this ODE admits a solution. However, a simple case is when $P^{y}$ is a strictly decreasing discount curve. In this case indeed, $P^y$ admits an inverse on the positive half line, noted  $Q^y$. Apply $Q^y$ to \eqref{eq:TC-cal} yields $\Theta^\star=Q^y(P^{market}(t))$. Furthermore, the inverse of a decreasing function is decreasing, and  the combination of two decreasing functions is itself increasing. Hence, if $P^{market}$ is decreasing, $\Theta^\star(t)$ is continuous and strictly increasing. Moreover, $\Theta^\star(0)=Q^y\left(P^{market}(0)\right)=Q^y(1)=0$. Hence, $\Theta^\star$ exists, and is a clock. 

\subsection{Proof of Theorem \ref{th:th2}}\label{seq:ProofTh2}

It is known from Corollary \ref{cor:cor2} that for any (non-trivial) (J)CIR process $y$ with parameter $\Xi$, there exits a clock $\Theta^\star(t)=\Theta^\star(t;\Xi)$ that yields a perfect fit between the curves $P^x$ generated y the TC-JCIR $x_t^{\theta^\star}:=\theta^\star(t)y_{\Theta^\star(t)}$. Same holds true for the JCIR++, $x_t^{\varphi^\star}$. This means:
\beq
P^{x^{\varphi^\star}}(t;\Xi)=P^{market}(t)=P^{x^{\theta^\star}}(t;\Xi)\;,\nonumber
\eeq
or equivalently,
\beq
e^{-\int_0^t\varphi^\star(u)du}P^y(t;\Xi)=P^{market}(t)=P^y(\Theta^\star(t);\Xi)\;.\nonumber
\eeq
Because $y$ is a JCIR, it can be arbitrarilly close to 0 at any time $t$, hence the calibration constraint amounts to force $\varphi^\star(t)\geq 0$ (or equivalently, $f^{\mathrm{JCIR}}(t)\leq f^{\mathrm{market}}(t)$) $\forall\, t\geq 0$. This implies that $P^y(\Theta^\star(t);\Xi)\leq P^y(t;\Xi)$. Because $P^y(.;\Xi)$ is a decreasing function, the last inequality is equivalent to $\Theta^\star(t)\geq t$.\\
To prove 1), we start from the increasingness of $v^Y(t)$  (Lemma \ref{Lem:VarJCIR}). Hence,  $v^Y(\Theta^\star(t))\geq v^Y(t)$ since $\Theta^\star(t)\geq t$.\\

From \eqref{eq:fjcirt}, we have, after some computations,
\begin{eqnarray*}
\frac{d}{d t}f^{\mathrm{JCIR}}(t) &=& \frac{4 \gamma^2e^{t\gamma}\left[\kappa\beta(\gamma-\kappa+(\kappa+\gamma)e^{t\gamma})+y_0 \gamma(\gamma-\kappa-(\kappa+\gamma)e^{t\gamma})+\omega\alpha(2\gamma+(\kappa+\gamma)(e^{t\gamma}-1))\right]}{[2\gamma+(\kappa+\gamma)(e^{t\gamma}-1)]^3}\\
&=& \frac{4\gamma^2e^{t\gamma}\left[(\gamma-\kappa)(\kappa\beta+y_0\gamma+\omega\alpha)-2\omega\alpha^2\right]+4\gamma^2e^{2t\gamma}\left[(\gamma+\kappa)(\kappa\beta-y_0\gamma+\omega\alpha)+2\omega\alpha^2\right]}{[2\gamma+(\kappa+\gamma)(e^{t\gamma}-1)]^3}\;.
\end{eqnarray*}
From this expression, one can check that $f^{\mathrm{JCIR}}$ is strictly increasing if $y_0< \beta+\omega\alpha/\kappa$ and $y_0\gamma\leq\kappa\beta+\omega\alpha$. It is strictly decreasing if $y_0\geq \beta+\omega\alpha/\kappa$. Otherwise, i.e., if $y_0< \beta+\omega\alpha/\kappa$ and $y_0\gamma>\kappa\beta + \omega\alpha$, the derivative has a root at
\begin{equation*}
t_2:=\frac{1}{\gamma}\ln \frac{(\gamma-\kappa)(\kappa\beta+y_0\gamma+\omega\alpha)-2\omega\alpha^2}{(\kappa+\gamma)(y_0\gamma-\kappa\beta-\omega\alpha)-2\omega\alpha^2}\;.
\end{equation*}
i.e., $f^{\mathrm{JCIR}}$ is first increasing, then decreasing.


The constraint
$\varphi^\star(t)\geq 0$ for all $t$, simply means that $f^{market}(t)\geq f^{\mathrm{JCIR}}(t)$ and so $\theta^\star(t)\geq\frac{f^{\mathrm{JCIR}}(t)}{f^{\mathrm{JCIR}}(\Theta^\star(t))}$. Observe that the condition $y_0=\beta+\omega\alpha/\kappa$ corresponds to the case where $v^y(t)$ is increasing and $f^{\mathrm{JCIR}}(t)$ is decreasing, hence $(i)$ holds.\\
If $f^{market}$ is constant, we have that $f^{market}(t)\geq f^{\mathrm{JCIR}}(t)$ which implies that $f^{market}(t)\geq f^{\mathrm{JCIR}}(\Theta^\star(t))$.
Clearly, if $f^{market}(t)$ is constant or $f^{\mathrm{JCIR}}$ is decreasing, then $\theta^\star(t)\geq 1$. And if $v^y(t)$ is increasing, $v^y(\Theta^\star(t))\geq v^y(t)$ since $\Theta^\star(t)\geq t$. From the fact that $\mathbb{V}[x^{\theta^\star}_t]:=\theta^\star(t)^2v^y(\Theta^\star(t))$ and the variation of $v^y(t)$ (Lemma \ref{Lem:VarJCIR}), $(ii)$, $(iii)$ and $(iv)$ follow.\\
In particular, taking $\omega\alpha=0$, we recover the CIR case which corresponds to the TC-CIR model.

\subsection{Black model for PSO}\label{sec:BlackPSO}

In this context, the Black-Scholes model works as follows.
We start by noting that the forward start CDS can be written in terms of the difference between the fair and the agreed premium cashflows. Indeed, the former corresponds to the protection leg. Inserting \eqref{eq:ParCDS} in \eqref{eq:CDS} yields
\begin{equation}
CDS_t(a,b,k)= \indic_{\{\tau >t\}} (s_t(a,b)-k)~C_t(a,b)\;.\nonumber
\end{equation}
The Payer default swaption becomes
\beq
PSO(a,b,k) =\mathbb{E}\left[(s_{T_a}(a,b)-k)^+~C_{T_a}(a,b)D(T_a)\right]=C_0(a,b)\mathbb{E}^{(a,b)}\left[(s_{T_a}(a,b)-k)^+\right]\label{eq:CDSO-BS}
\eeq
where $\E^{(a,b)}$ stands for the expectation under the equivalent measure $\Q^{(a,b)}$, associated with the num\'eraire $C(a,b)$. Interestingly, it is clear from \eqref{eq:ParCDS} that the par spread $s(a,b)$ is a $\Q^{(a,b)}$-martingale on $[0,T_a]$. Hence, the Black-Scholes model for CDSO naturally postulates $\Q^{(a,b)}$-martingale dynamics for the par spread\\
\begin{equation*}
ds_t(a,b) = \bar{\sigma} s_t(a,b)dW^{s}_t,\quad t\leq T_a
\end{equation*}
where $W^{s}$ is a $\Q^{(a,b)}$-Brownian motion. Eventually, the expectation in \eqref{eq:CDSO-BS} is given by the standard Black-Scholes formula by setting $r\leftarrow 0$. Hence, the Black-Scholes price of the PSO is given by 
\beq
PSO^{Black}(a,b,k,\bar{\sigma}) = C_0(a,b)\left[s_0(a,b)\Phi(d_1)-k\Phi(d_2)\right]\nonumber
\eeq
where
\beq
d_1=\frac{\ln\frac{s_0(a,b)}{k}+\frac{1}{2}\bar{\sigma}^2T_a}{\bar{\sigma}\sqrt{T_a}},\qquad d_2=d_1-\bar{\sigma}\sqrt{T_a}\nonumber
\eeq
and $\Phi$ is the distribution function of a standard Normal random variable.
\bibliographystyle{plainnat}
\bibliography{./MyBib}

\ifdefined \InclFig
\end{document} 
\fi

\begin{figure}[H]
\centering
\subfigure[$Z$ curves (fit on $Z^{market}$)]{\includegraphics[width=0.32\columnwidth]{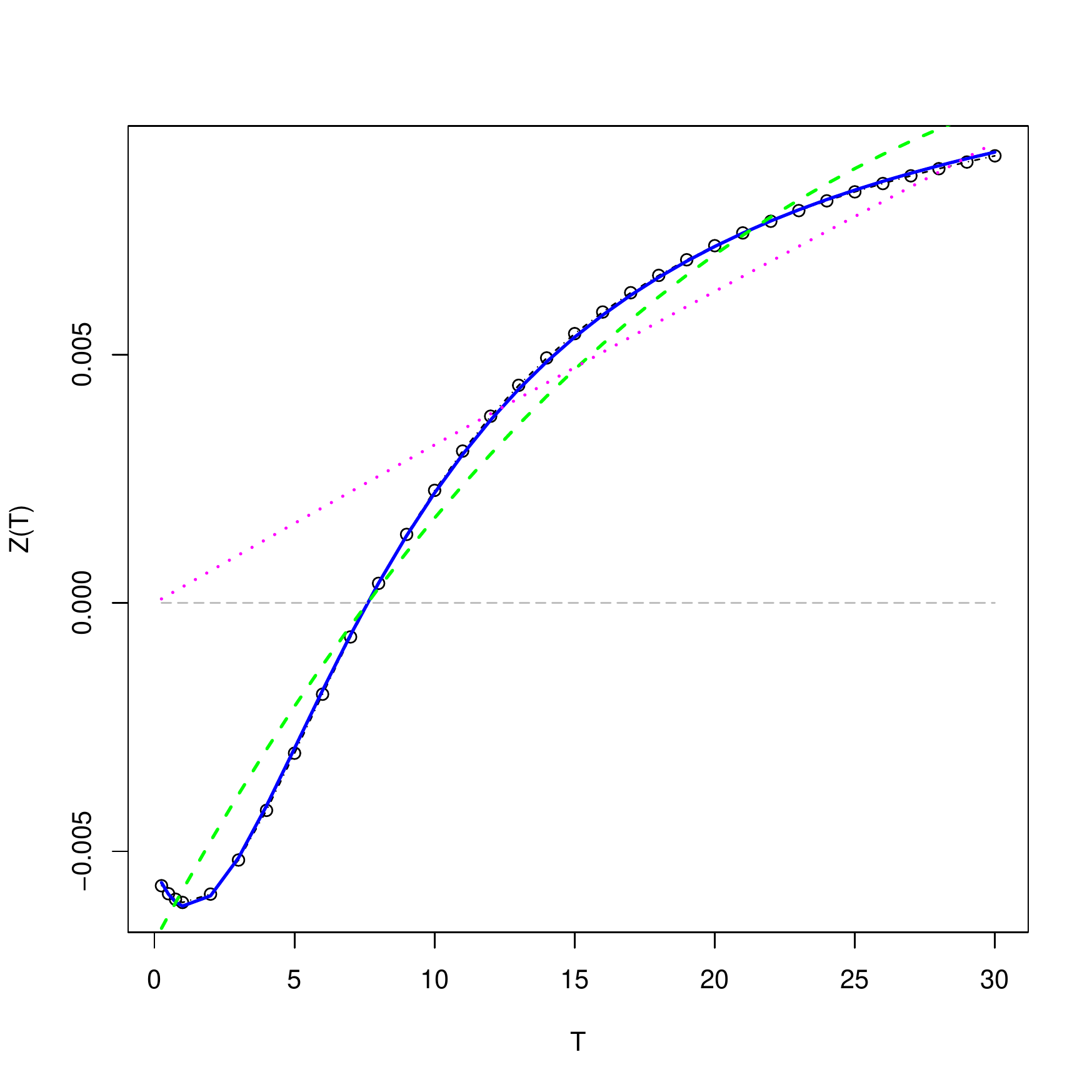}}\hspace{0.01cm}
\subfigure[$f$ curves (fit on $Z^{market}$)]{\includegraphics[width=0.32\columnwidth]{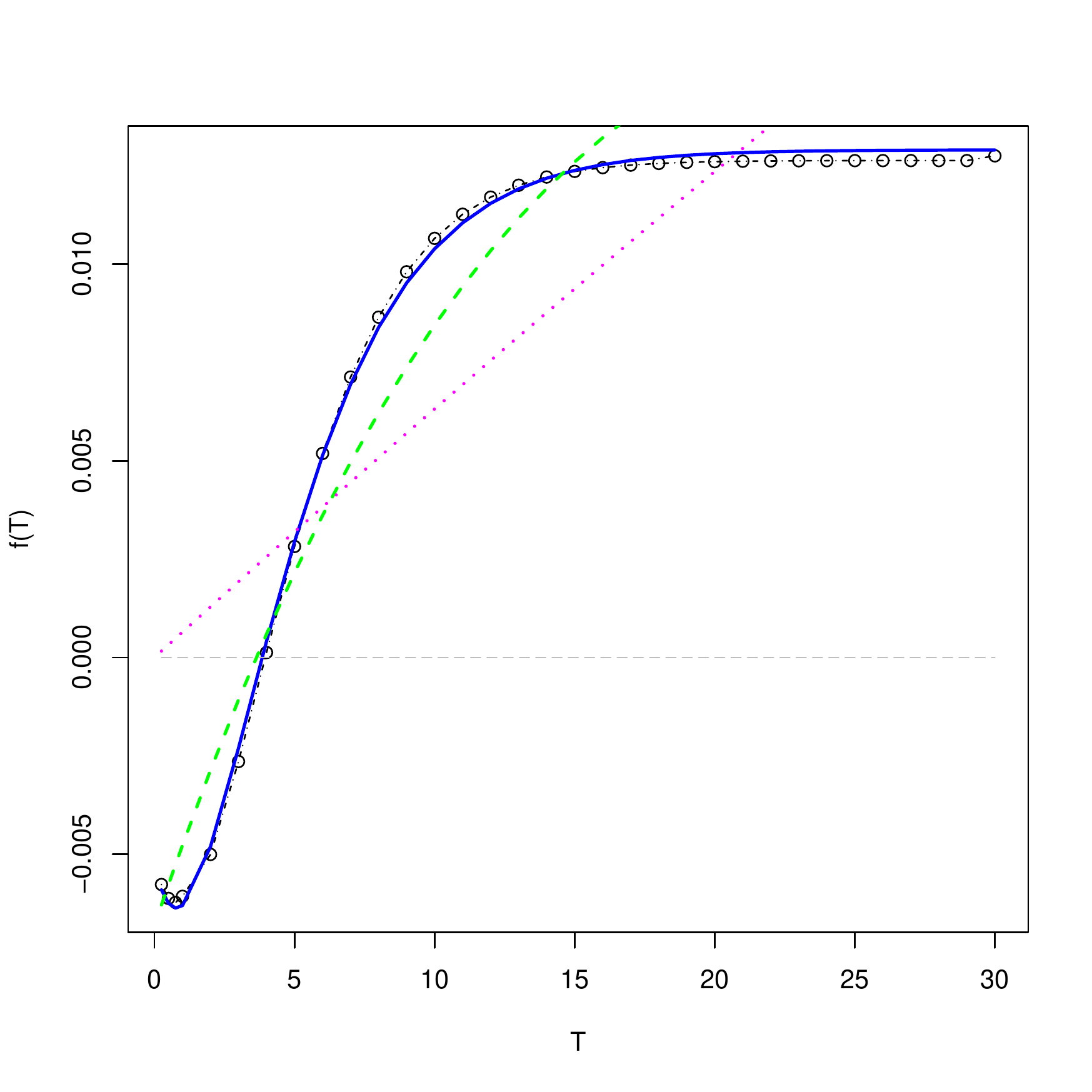}}\hspace{0.01cm}
\subfigure[$P$ curves (fit on $Z^{market}$)]{\includegraphics[width=0.32\columnwidth]{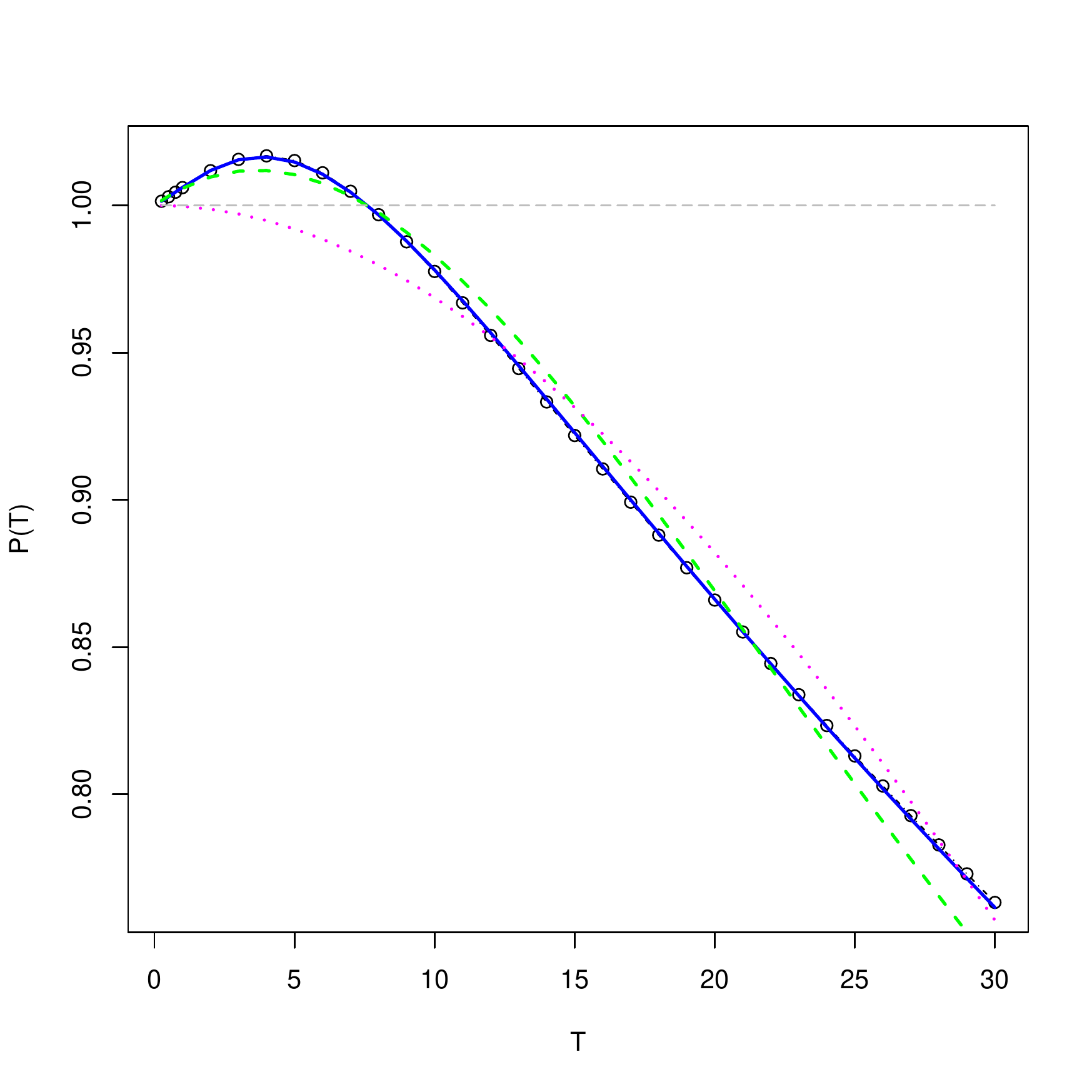}}\hspace{0.01cm}\\
\subfigure[$Z$ curves (fit on $f^{market}$)]{\includegraphics[width=0.32\columnwidth]{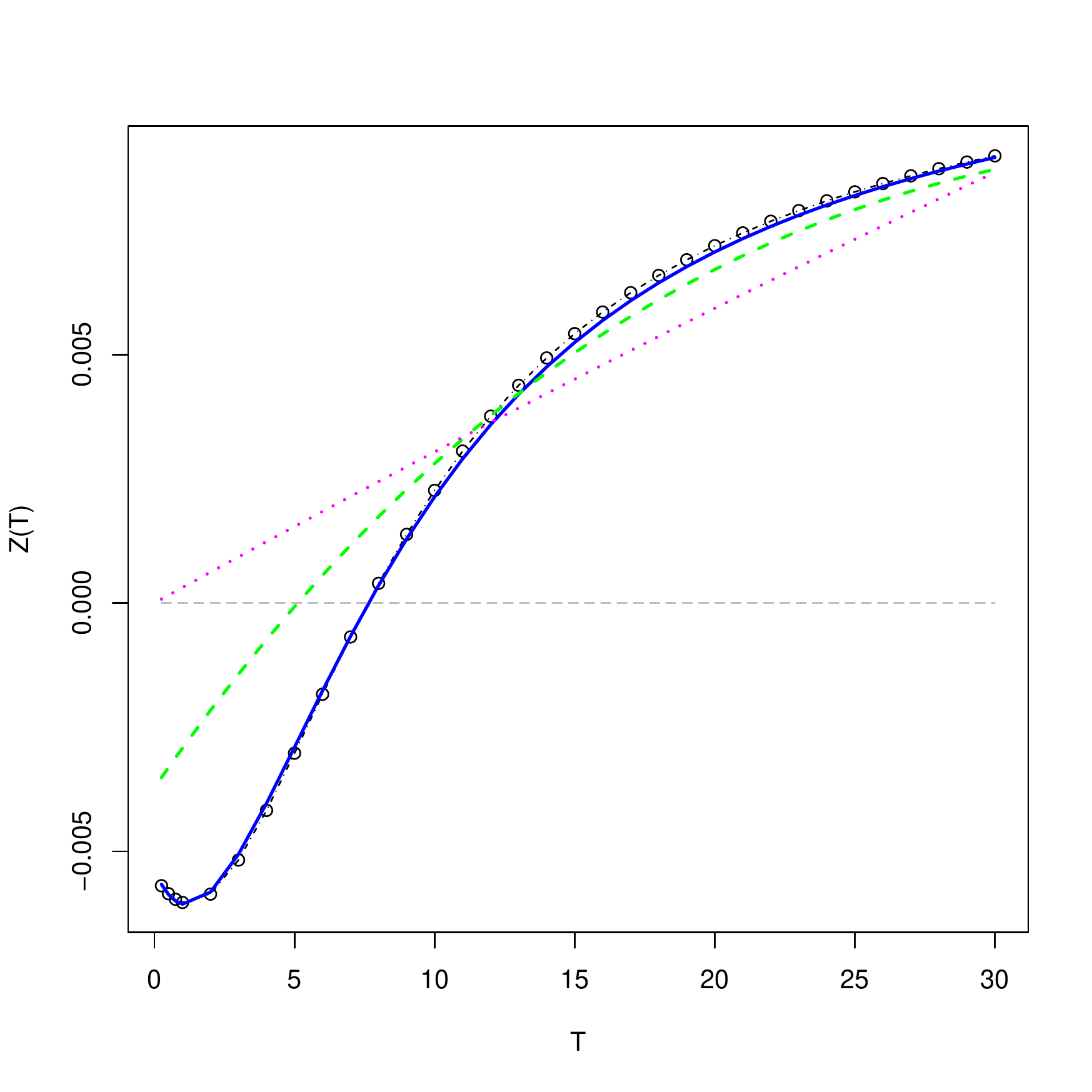}}\hspace{0.01cm}
\subfigure[$f$ curves (fit on $f^{market}$)]{\includegraphics[width=0.32\columnwidth]{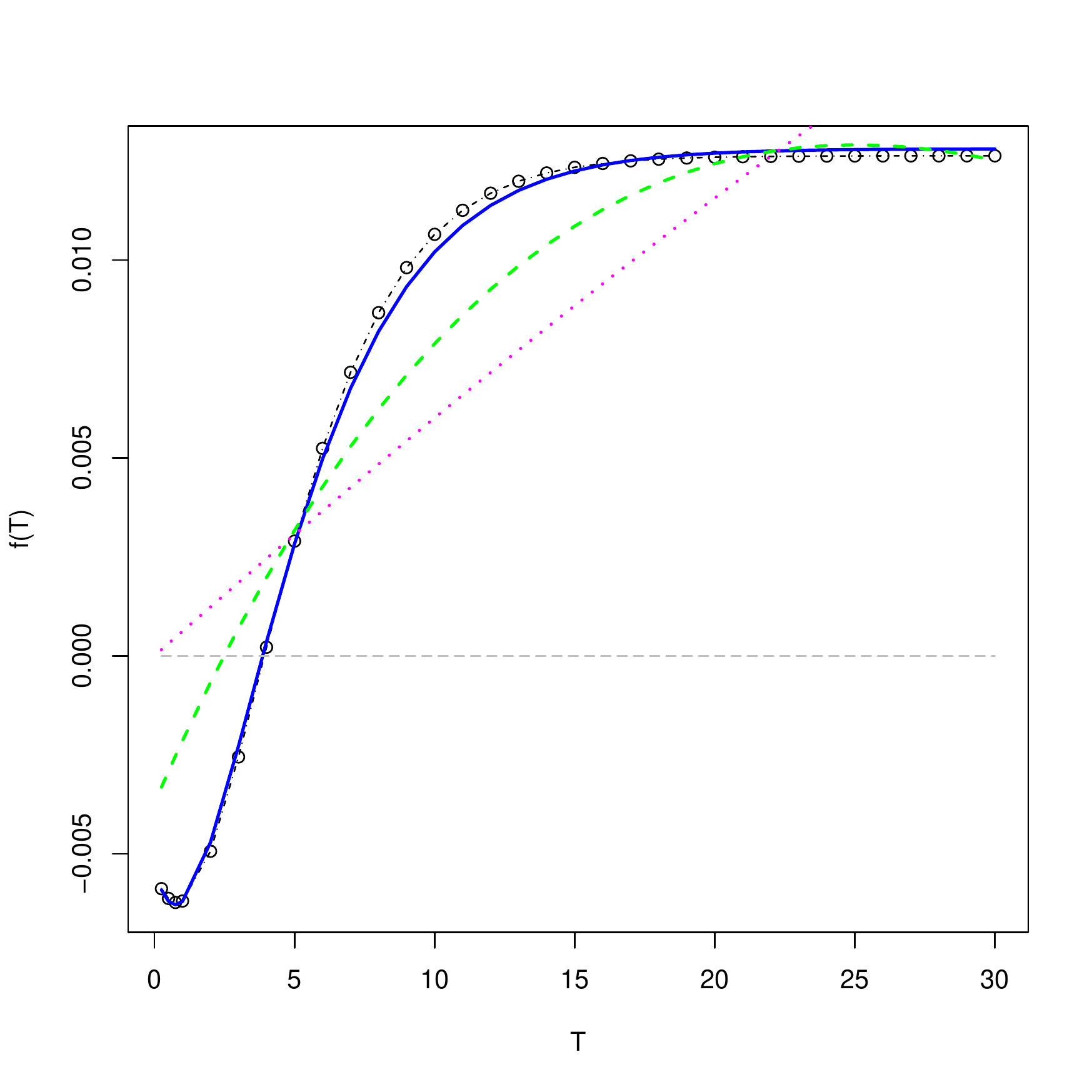}}\hspace{0.01cm}
\subfigure[$P$ curves (fit on $f^{market}$)]{\includegraphics[width=0.32\columnwidth]{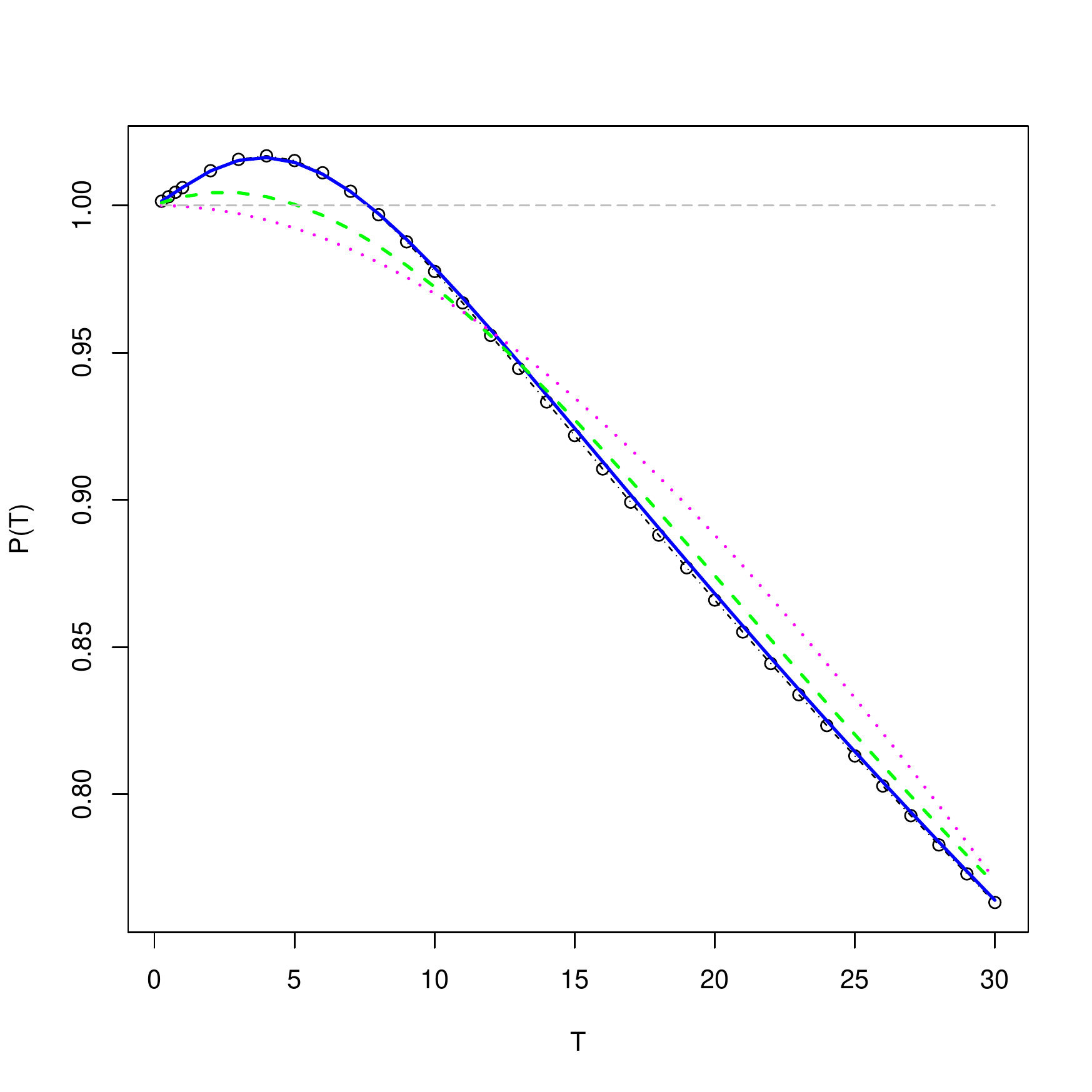}}\hspace{0.01cm}\\
\subfigure[$Z$ curves (fit on $P^{market}$)]{\includegraphics[width=0.32\columnwidth]{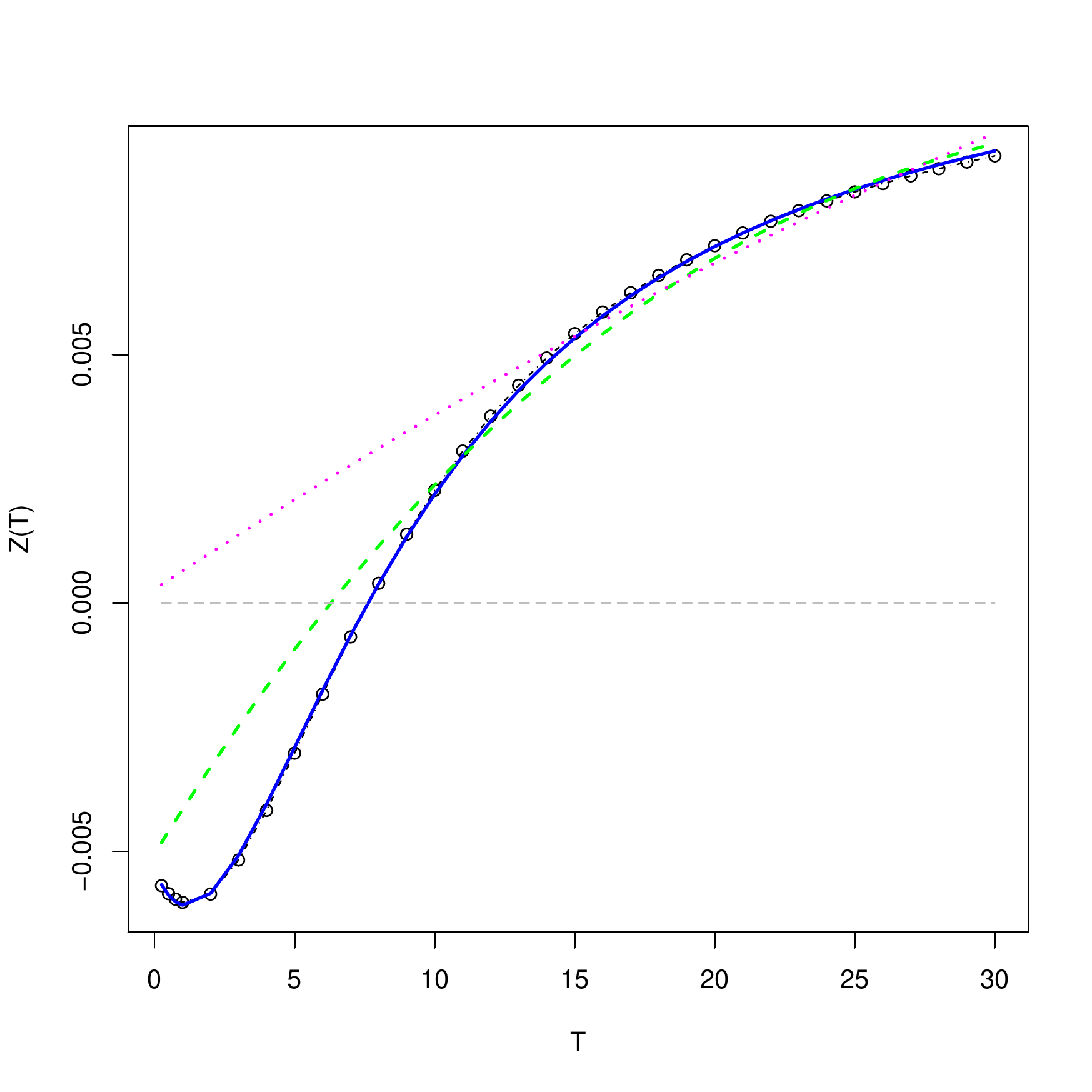}}\hspace{0.01cm}
\subfigure[$f$ curves (fit on $P^{market}$)]{\includegraphics[width=0.32\columnwidth]{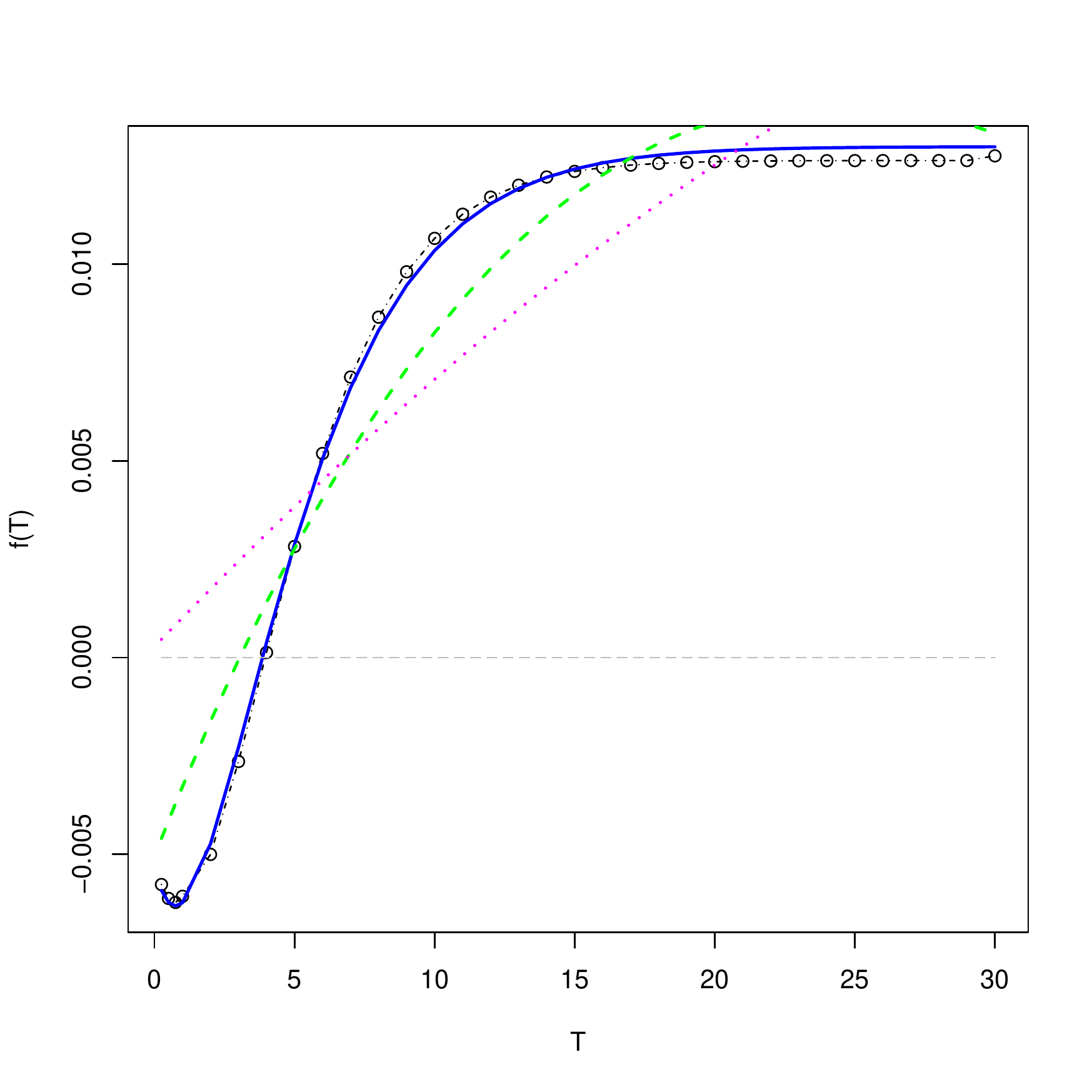}}\hspace{0.01cm}
\subfigure[$P$ curves (fit on $P^{market}$)]{\includegraphics[width=0.32\columnwidth]{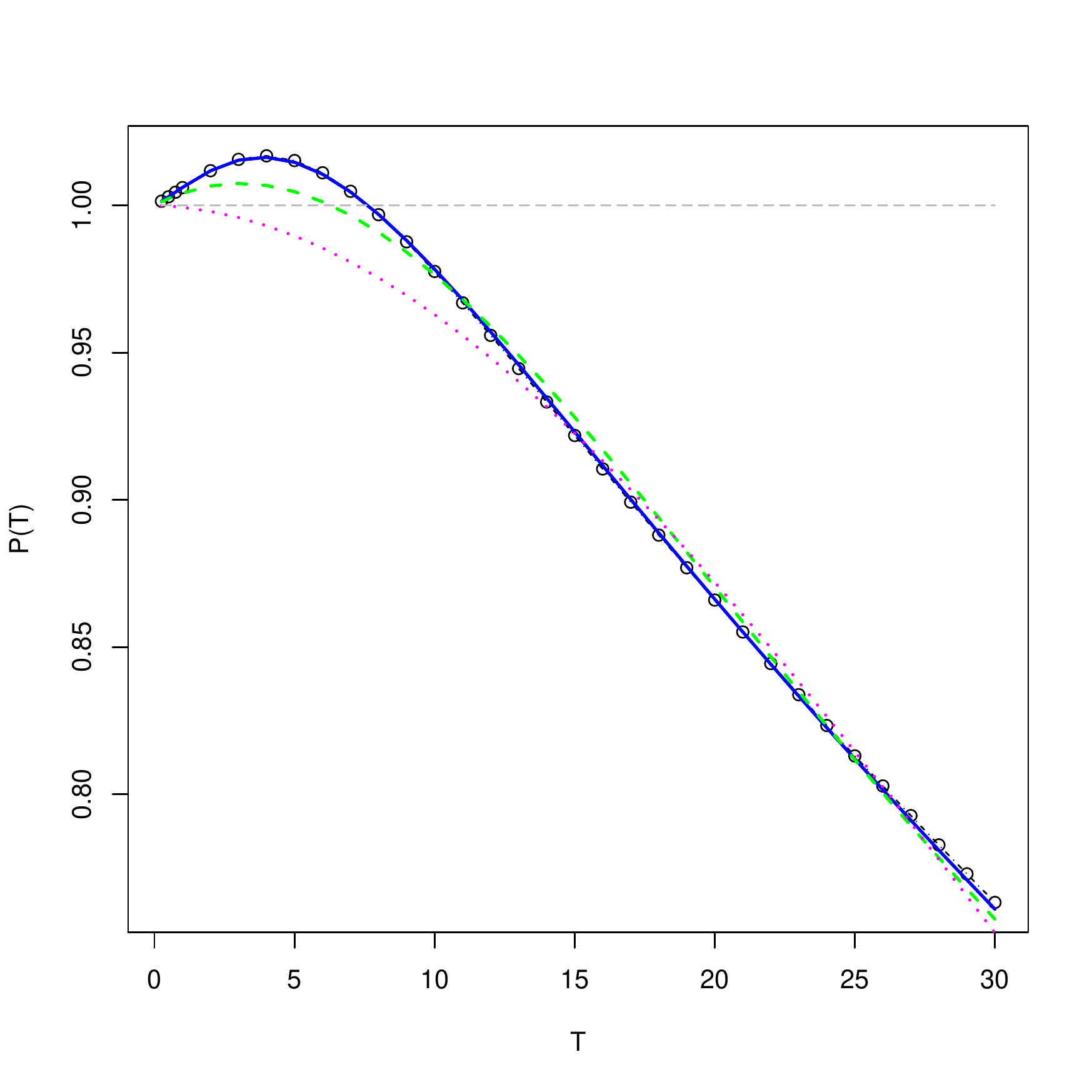}}\hspace{0.01cm}\\\caption{This pictures shows the yield curves ($Z$, left), the instantaneous forward curves ($f$, center) and ZCB price curves ($P$, right) generated by the Vasicek (green, dashed), CIR (magenta, dotted) and the Nelsen-Siegel (solid, blue) models when they  are fitted on $Z^{market}$ curve (top), $f^{market}$ curve (middle) and $P^{market}$ curve (bottom) shown in black (circles). The yield curve $Z^{market}$ downloaded from the European Central Bank website and corresponds to AAA-rated euro area central government bonds. Although Nelsen-Siegel proves to be relatively satisfactoty compared to the curves generated by the affine diffusions models, none of them provide a perfectly fit with the market curve, whatever the calibration target chose: $Z^{market}$, $f^{market}$ or $P^{market}$.}\label{fig:NS}
\end{figure}

\begin{figure}[H]
\centering
\subfigure[Piecewise constant hazard rate ($h$)]{\includegraphics[width=0.45\columnwidth]{h_cons.pdf}}\hspace{0.2cm}
\subfigure[Optimal shift ($\varphi=\varphi^\star$)]{\includegraphics[width=0.45\columnwidth]{Phi_cons.pdf}}\hspace{0.2cm}
\subfigure[Optimal clock ($\Theta=\Theta^\star$)]{\includegraphics[width=0.45\columnwidth]{Theta_cons.pdf}}\hspace{0.2cm}
\subfigure[Survival probability ($P^y,G=P^{\lambda^{\varphi}}=P^{\lambda^{\theta}}$)]{\includegraphics[width=0.45\columnwidth]{G_cons_theta.pdf}}
\caption{Fitting Ford Inc. CDS term-structure with adjusted CIR. The survival probability curve $G(t)$ is parametrized with a piecewise constant hazard rate function $h(t)$ extracted from market prices taken from Bloomberg on November 12 2018, panel (a). The base model $y$ is a CIR with parameters $\Xi=\Xi^\star$ where $\Xi^\star=(\kappa,\beta,\delta) = (0.0555, 0.3018, 0.2939)$ is obtained from eq. \eqref{eq:CalyNoConstraint} and $y_0=h_0$. The shift function $\varphi(t)=\varphi^\star(t,\Xi^\star)$ is shown in panel (b). Panel (c) gives the clock $\Theta(t)=\Theta^\star(t;\Xi^\star)$. Eventually, panel (d) yields the survival probability curves given by the market ($G(t)$, green), or associated to $\Q(\tau(\lambda)>t)$ for various intensity models $\lambda$ : the best base model $\lambda\leftarrow y$ (leading to $\Q(\tau(y)>t)=P^y(t,\Xi^\star)$, dashed blue), $\lambda\leftarrow \lambda^{\varphi}$ (CIR++) and $\lambda\leftarrow \lambda^{\theta}$ (TC-CIR model). By construction of $\varphi$ and $\Theta$, the last two curves coincide (magenta) and agree with $G(t)$.}\label{fig:Ford1}
\end{figure}

\begin{figure}[H]
\centering
\subfigure[Piecewise linear hazard rate ($h$)]{\includegraphics[width=0.45\columnwidth]{h_linear.pdf}}
\subfigure[Optimal shift ($\varphi=\varphi^\star$)]{\includegraphics[width=0.45\columnwidth]{Phi_linear.pdf}}\hspace{0.2cm}
\subfigure[Optimal clock ($\Theta=\Theta^\star$)]{\includegraphics[width=0.45\columnwidth]{Theta_linear.pdf}}\hspace{0.2cm}
\subfigure[Survival probability ( $P^y,G=P^{\lambda^{\varphi}}=P^{\lambda^{\theta}}$)]{\includegraphics[width=0.45\columnwidth]{G_lin_theta.pdf}}
\caption{Fitting Ford Inc. CDS term-structure with adjusted CIR. The survival probability curve $G(t)$ is parametrized with a piecewise linear hazard rate function $h(t)$ extracted from market prices taken from Bloomberg on Novermber 12 2018, panel (a). The base model $y$ is a CIR with parameters $\Xi=\Xi^\star$ where $\Xi^\star=(\kappa,\beta,\delta) = (0.0620, 0.2729, 0.2926)$ is obtained from eq. \eqref{eq:CalyNoConstraint} and $y_0=h_0$. The shift function $\varphi(t)=\varphi^\star(t,\Xi^\star)$ is shown in panel (b). Panel (c) gives the clock $\Theta(t)=\Theta^\star(t;\Xi^\star)$. Eventually, panel (d) yields the survival probability curves given by the market ($G(t)$, green), or associated to $\Q(\tau(\lambda)>t)$ for various intensity models $\lambda$ : the best base model $\lambda\leftarrow y$ (leading to $\Q(\tau(y)>t)=P^y(t,\Xi^\star)$, dashed blue), $\lambda\leftarrow \lambda^{\varphi}$ (CIR++) and $\lambda\leftarrow \lambda^{\theta}$ (TC-CIR model). By construction of $\varphi$ and $\Theta$, the last two curves coincide (magenta) and agree with $G(t)$.}\label{fig:Ford2}
\end{figure}

\begin{figure}[H]
\centering
\subfigure[Survival probability ($P^{y^+},G=P^{\lambda^{\varphi,+}}$)]{\includegraphics[width=0.45\columnwidth]{G_cons_phi_positif.pdf}}\hspace{0.2cm}
\subfigure[Survival probability ($P^{y^+},G=P^{\lambda^{\varphi,+}}$)]{\includegraphics[width=0.45\columnwidth]{G_cons_phi_positif_lambda0.pdf}}\hspace{0.2cm}
\subfigure[Shift function ($\varphi=\varphi^\star$)]{\includegraphics[width=0.45\columnwidth]{Phi_cons_positif.pdf}}\hspace{0.2cm}
\subfigure[Shift function ($\varphi=\varphi^{\star,+}$)]{\includegraphics[width=0.45\columnwidth]{Phi_cons_positif_lambda0.pdf}}\hspace{0.2cm}
\caption{Fitting Ford Inc. CDS term-structure using $\lambda^{\varphi^\star,+}$. Panels (a) and (c) correspond to $y^+_0=h_0$ whereas panels (b) and (d) correspond to the case where $y^+_0$ is one of the optimized parameters. The survival probability curve $G(t)$ is parametrized with a piecewise constant hazard rate function $h(t)$ extracted from market prices taken from Bloomberg on November 12 2018. The parameters $\Xi^{\star,+}$ are computed under the constraint $\varphi(t)=\varphi^\star(t;\Xi^{\star,+})\geq 0$ . The base model $y^+$ is a CIR with parameters $\Xi=\Xi^{\star,+}$ (left) and $\Xi=\Xi_0^{\star,+}$ (right).}\label{fig:FordConstr} 
\end{figure}

\begin{figure}[H]
\centering
\subfigure[$y_0=h_0$]{\includegraphics[width=0.45\columnwidth]{Var_lambda.pdf}}\hspace{0.2cm}
\subfigure[$y_0$ optimized]{\includegraphics[width=0.45\columnwidth]{Var_lambda_lmd0.pdf}}
\caption{Variances of the integrated versions of  $\lambda^{\varphi,+}$ (PS-CIR $\Xi=\Xi^{\star,+}$ (left) and $\Xi=\Xi_0^{\star,+}$ (right), solid blue), $\lambda^{\varphi}$ (S-CIR $\Xi=\Xi^\star$ (left) and $\Xi=\Xi_0^\star$ (right), dashed blue), and $\lambda^{\theta}$ (TC-CIR without constraint $\Xi=\Xi^\star$ (left) and $\Xi=\Xi_0^\star$ (right), magenta).}\label{fig:FordConstr2}
\end{figure}

\begin{figure}[H]
\centering
\subfigure[$dV_t=\nu dW^V_t$, $y_0=h_0$]{\includegraphics[width=0.48\columnwidth]{CVA_Gauss_Exp.pdf}}\hspace{0.20cm}
\subfigure[$dV_t=\nu dW^V_t$, $y_0$ optimized]{\includegraphics[width=0.48\columnwidth]{CVA_Gauss_lambda0.pdf}}\hspace{0.20cm}
\subfigure[$dV_t=\left(\gamma(T-t)-\frac{V_t}{T-t}\right) dt + \nu dW^V_t$, $y_0=h_0$]{\includegraphics[width=0.48\columnwidth]{CVA_Swp_Exp.pdf}}\hspace{0.20cm}
\subfigure[ $dV_t=\left(\gamma(T-t)-\frac{V_t}{T-t}\right) dt + \nu dW^V_t$, $y_0$ optimized]{\includegraphics[width=0.48\columnwidth]{CVA_Swp_lambda0.pdf}}
\caption{Impact of the exposure-credit correlation $\rho$ on CVA levels for prototypical forward (top) and swap exposure (bottom) with $y_0=h_0$ (left) or $y_0$ optimized (right), $\nu=8\%$ and $\gamma=0.1\%$. Under the wrong-way risk effect, the CVAs are computed using adaptive control variate \footnote{see \cite{Mbaye2018} for the implementation of the adaptive control variate applied on CVA computation.} Monte Carlo simulation (100K paths, time step 0.01) with intensities $\lambda^{\varphi,+}$ (CIR++ with positive shift constraint $\Xi=\Xi^{\star,+}$ (left) and $\Xi=\Xi_0^{\star,+}$ (right), dotted blue), $\lambda^{\varphi}$ (CIR++ without constraint $\Xi=\Xi^\star$ (left) and $\Xi=\Xi_0^\star$ (right), solid blue), and $\lambda^{\theta}$ (TC-CIR without constraint $\Xi=\Xi^\star$ (left) and $\Xi=\Xi_0^\star$ (right), dotted magenta). The case without wrong-way risk corresponds to the  cyan.}\label{fig:CVAplots}
\end{figure}

\end{document}